\documentclass[12pt]{article}

\pdfoutput=1

\usepackage[colorlinks,linkcolor=blue,citecolor=blue,bookmarks=false,pagebackref]{hyperref}
\usepackage{latexsym,amsfonts,bm,epsfig,amsmath,natbib,authblk,amsthm,thmtools,amssymb}
\usepackage{cleveref}
\usepackage{enumitem}
\usepackage{thm-restate}
\usepackage{IEEEtrantools}
\usepackage{dsfont}
\usepackage{graphicx}
\usepackage{xspace}
\usepackage{color}
\usepackage{setspace}
\usepackage{subcaption}
\usepackage{xcolor}
\usepackage{tikz}
\usetikzlibrary{positioning, arrows.meta, calc}
\usepackage{IEEEtrantools}
\usepackage{float}
\usepackage{booktabs}
\usepackage[margin=1cm]{caption}
\usepackage{xr}
\usepackage[linesnumbered,ruled,vlined]{algorithm2e}
\usepackage{algorithmic}
\usepackage{bigints}

\usepackage[toc,page,header]{appendix}
\usepackage{minitoc}
\noptcrule

\floatstyle{plain}

\makeatletter
\newcommand{\myitem}[1]{%
\item[#1]\protected@edef\@currentlabel{#1}%
}
\makeatother

\newcommand*{\x}{\bm{x}}
\newcommand*{\btheta}{\bm{\theta}}
\newcommand*{\X}{\mathcal{X}}
\newcommand*{\E}{\mathbb{E}}
\newcommand*{\R}{\mathbb{R}}

\newcommand*{\LRM}{D^{\operatorname{LRM}}}
\newcommand*{\Lhat}{\hat{\mathcal{L}}^{\operatorname{LRM}}_n}
\newcommand*{\Lexp}{\mathcal{L}^{\operatorname{LRM}}}
\newcommand*{\model}{p_{\bm{\theta}}}

\newcommand*{\AllPMFs}{\mathcal{Q}(\X)} 
\newcommand*{\AllPMFsGiven}[1]{\mathcal{Q}_{#1}(\X)} 
\newcommand*{\AllAdmPMFsGiven}[1]{\mathcal{Q}^{\operatorname{adm}}_{#1}(\X)}
\newcommand*{\modelFamily}{\mathcal{P}_{\Theta}(\X)}

\newcommand*{\basePMF}{q^{\dagger}}
\newcommand*{\tildebasePMF}{\tilde{q}^{\dagger}}
\newcommand*{\ZbasePMF}{Z^{\dagger}}

\newcommand*{\dgp}{q_{0}}

\newcommand*{\empirical}{\hat{q}^{\operatorname{emp}}}

\newcommand*{\RegCondOperator}[1]{\mathcal{R}_{#1}}

\newtheorem{theorem}{Theorem}[section]
\newtheorem{proposition}[theorem]{Proposition}
\newtheorem{corollary}{Corollary}[theorem]
\newtheorem{lemma}[theorem]{Lemma}

\theoremstyle{definition}
\newtheorem{definition}{Definition}[section]
\newtheorem{assumption}[theorem]{Assumption}
\newtheorem*{standingassumption}{Standing Assumption}

\theoremstyle{remark}

\newlist{assumpenum}{enumerate}{1}
\setlist[assumpenum]{
  label=\textbf{(\arabic*)}, 
  ref=\theassumption.\arabic*, 
  leftmargin=*, itemsep=0.25em
 }

\crefname{assumption}{Assumption}{Assumptions}
\Crefname{assumption}{Assumption}{Assumptions}
\crefname{assumpenumi}{Assumption}{Assumptions}
\Crefname{assumpenumi}{Assumption}{Assumptions}

\newlist{sassumpenum}{enumerate}{1}
\setlist[sassumpenum]{
  label=\textbf{(\arabic*)},
  ref=(\arabic*),                
  leftmargin=*, itemsep=0.25em
}

\crefname{sassumpenumi}{Standing Assumption}{Standing Assumptions}
\Crefname{sassumpenumi}{Standing Assumption}{Standing Assumptions}

\newcommand\blfootnote[1]{%
\begingroup
\renewcommand\thefootnote{}\footnote{#1}%
\addtocounter{footnote}{-1}%
\endgroup
}

\begin{document}

\def\spacingset#1{\renewcommand{\baselinestretch}%
{#1}\small\normalsize} \spacingset{1}

\title{Conjugate Generalized Bayesian Inference for Discrete Doubly Intractable Problems}

\author{
\large
 William Laplante$^{\dagger 1, 2, 4}$, Matias Altamirano$^{2}$, Jeremias Knoblauch$^{2}$, Andrew Duncan$^{3}$, and Fran\c{c}ois-Xavier Briol$^{2}$
}
\date{
\small
$^1$\textit{Department of Physics and Astronomy, University College London, London, United Kingdom}
\\
$^2$\textit{Department of Statistical Science, University College London, London, United Kingdom}
\\
$^3$\textit{Department of Mathematics, Imperial College London, London, United Kingdom}
\\
$^4$\textit{The Alan Turing Institute, London, United Kingdom}
\\[2ex] 
\today
\vspace*{-1cm}
}

\maketitle
\spacingset{1.8} 
 \blfootnote{$^\dagger$ Corresponding author. william.laplante.24@ucl.ac.uk}
\vspace{-1.2cm}
\begin{abstract}
\vspace{-0.3cm}
Doubly intractable problems occur when both the likelihood and the posterior are available only in unnormalized form, with computationally intractable normalization constants.
Bayesian inference then typically requires direct approximation of the posterior through specialized and typically expensive MCMC methods.
In this paper, we provide a computationally efficient alternative in the form of a novel generalized Bayesian posterior that allows for conjugate, closed-form or Gibbs-based MCMC inference within the class of exponential family models for discrete data.
We derive theoretical guarantees to characterize the asymptotic behavior of the generalized posterior, supporting its use for inference. 
The method is evaluated on a range of challenging intractable exponential family models, including the Conway-Maxwell-Poisson graphical model of multivariate count data,  autoregressive discrete time series models, and Markov random fields  such as the Ising and Potts models.
The computational gains are significant; in our experiments, the method is between $10$ and $6000$ times faster than state-of-the-art Bayesian computational methods.
\end{abstract}

\vspace*{-0.25cm}

\doparttoc 
\faketableofcontents 
\part{}

\vspace*{-1cm}
\section{Introduction}

An \emph{intractable} model is one for which likelihood evaluation is not feasible, such as when the normalization constant (partition function) cannot be computed or efficiently approximated. 
In discrete settings, statistical models rapidly become intractable due to normalization constants requiring summations over large, sometimes infinite, state spaces.
These intractable discrete models are prevalent, appearing in spatial statistics \citep{green2002hidden, hughes2011autologistic}, image analysis and computer vision \citep{moores2020scalable}, statistical network analysis \citep{lusher2013exponential, bouranis2018bayesian, lunagomez2021modeling}, statistical physics \citep[e.g., Ising and Potts models;][]{ising1925beitrag, potts1952some, wu1982potts, mcgrory2009variational, kim2024statistically},
and multivariate count data \citep{piancastelli2023multivariate, Sellers_2023}. 
They are especially problematic in Bayesian settings since standard sampling methods no longer apply, leading to the well-established and challenging \emph{doubly-intractable problem}, where the posterior as well as the model (or likelihood) are intractable. 
Existing approaches provide approximate and often costly solutions, using MCMC with auxiliary variables \citep{moller2006efficient, murray2006mcmc, andrieu2009pseudo, liang2016adaptive, park2018bayesian}, approximations to the normalizing constant \citep{lyne2015russian}, surrogate or composite likelihoods \citep{besag1975statistical, Lindsay1988CompositeLikelihood, pauli2011bayesian, park2021bayesian}, or variational Bayes methods \citep{ingraham2017variational, tan2020bayesian, lee2024steingradientdescentapproach}. 
Alternatively, some assume that although the intractable model cannot be evaluated, it can be simulated from, leading to simulation-based inference methods \citep{marin2012approximate, price2018bayesian, cranmer2020frontier}.
Overall, the computational burden introduced by intractable models is significant, often leading practitioners to avoid them in favor of simpler alternatives at the expense of expressiveness.
This leads directly to misspecified models, where unwanted effects are encountered, such as failing to capture data dispersion well or excluding the impact of important covariates.

For continuous domains, \citet{matsubara2022robust} introduced a method based on the \emph{kernel Stein discrepancy} (KSD), and generalized Bayesian inference: a framework that extends traditional Bayesian updating,  which relies on the likelihood, by enabling updates based on a broader class of loss functions or divergences \citep{bissiri2016general,  knoblauch2022optimization}.
The method, called KSD-Bayes, is constructed such that (1) the need to compute the statistical model's normalization constant is removed, and (2) a \emph{conjugate} generalized posterior is achieved for intractable models in the exponential family, which led to significant computational savings.

This was then followed up by \citet{matsubara2024generalized}, who attempted to replicate the approach for discrete domains using the \emph{discrete Fisher divergence} (DFD).
Although DFD-Bayes accomplishes the former property---improving on existing methods by enabling the use of out-of-the-box Markov Chain Monte Carlo (MCMC) algorithms---it does not yield conjugate posteriors.
As such, DFD-Bayes faces the typical drawbacks of MCMC sampling: it requires tuning and convergence checks, and struggles when the resulting posterior is multi-modal, the model is high-dimensional, or when the number of data points is large. 
In contrast, these issues are at the very least mitigated or entirely avoided when a conjugate posterior is obtained.

This paper improves upon DFD-Bayes by crafting a new statistical divergence, called the \emph{log-ratio matching (LRM) divergence}, which induces conjugate generalized Bayesian posteriors for intractable exponential family models of discrete data, thus mirroring the developments of KSD-Bayes in continuous domains. 
This is possible because the LRM divergence is quadratic in the natural parameters of the exponential family, and conjugacy can therefore be achieved whenever the prior is chosen to be exponentially quadratic, such as a multivariate Gaussian.
Closed-form generalized posteriors (potentially truncated) can also be obtained when the prior is either constant, linear, or quadratic in the natural parameters of the exponential family model. 
Beyond full conjugacy and closed-form posterior inference, LRM can also deliver partial (blockwise) conjugacy for models that are exponential family in a subset of their parameters.
This enables exact Gibbs updates and thus fast mixing Markov chain Monte Carlo.
Finally, under mild conditions, this LRM-Bayes posterior satisfies a Bernstein-von-Mises result.

In the cases studied, including a 65-dimensional posterior for a multivariate Conway–Maxwell–Poisson (CMP) graphical model of cancer data, and a Potts model of a satellite image on a $171 \times 171$ grid, our conjugate method is between $10$ and $6000$ times faster than existing methods, while yielding posterior distributions closely matching those from standard Bayesian methods for doubly intractable problems.  Furthermore, partial conjugacy for exact Gibbs updates is demonstrated for an autoregressive count time-series problem and for the previously mentioned CMP graphical model, but with a horseshoe prior.

\vspace{-1.em}
\section{Background}
This section first discusses how generalized Bayesian inference can enable a substantial reduction in computation cost to update belief distributions. 
We then discuss discrete divergences, along with their corresponding estimators, and emphasize their shortcomings in producing conjugate posteriors---a highly desirable feature for doubly intractable problems.

\vspace{-1em}
\subsection{Generalized Bayes For Accelerated Computation}\label{sec:gen-bayes-background}

Throughout the remainder, the probability mass function of a statistical model parametrized by some continuous parameter $\btheta \in \Theta$ and defined on data domain $\X$ will be denoted as $\model$.
In the generalized Bayesian inference framework, a loss function $\mathcal{L}:\Theta \rightarrow \mathbb{R}$ linking data to parameters is specified. 
We concern ourselves with the special case in which $\mathcal{L}$ is derived from a statistical divergence $D(\dgp\|\model)$ between the data-generating process $\dgp$ and the model $\model$, typically up to an additive constant $C(\dgp)$ independent of $\btheta$ so that  \( \mathcal{L}(\btheta) = D(\dgp \| \model) + C(\dgp) \) \citep{jewson2018principles}. 
Recall that for any two probability mass functions $q$ and $p$, a divergence satisfies \( D(q \| p) \geq 0 \) and \( D(q \| p) = 0 \) if and only if \( q = p \).
In most cases, $\mathcal{L}$ cannot be computed directly; it is instead approximated using an estimator of the loss function \( \hat{\mathcal{L}}:\Theta \times \X^n \rightarrow \mathbb{R} \) dependent on i.i.d.\ observations \( \{\mathbf{x}_i\}_{i=1}^n \sim \dgp \).
For convenience of notation, we often suppress explicit dependence on the data and define $\hat{\mathcal{L}}_n(\btheta):= \hat{\mathcal{L}}(\btheta, \{\x_i\}_{i=1}^n) $.
Given a prior \( \pi\) and scaling parameter \( \beta > 0 \), the \emph{generalized posterior} has density\footnote{The semi-colon notation emphasizes that the posterior is a function of the data but is not conditioned on it in the standard Bayesian sense.}
\vspace{-1em}
{\setlength{\abovedisplayskip}{3pt}
 \setlength{\belowdisplayskip}{4pt}
\begin{equation}
    \pi_{\mathcal{L}}^\beta \left(\bm{\theta} ; \;  \{\bm{x}_i\}_{i=1}^n \right) \propto \exp \left (-\beta n \, \hat{\mathcal{L}}_n(\btheta)  \right) \pi (\bm{\theta}),\label{eq:gen-bayes-posterior}
\end{equation}
}
where we assume that $\int_{\Theta} \exp (-\beta n \hat{\mathcal{L}}_n(\btheta)) \pi(\btheta) d\btheta < \infty$ so that the density can be normalized; however, the normalization constant need not be a marginal likelihood.
The estimator $\hat{\mathcal{L}}_n(\btheta)$ is typically chosen to be strongly consistent so that $\hat{\mathcal{L}}_n(\btheta ) \longrightarrow \mathcal{L}(\btheta)$ almost surely as $n \rightarrow \infty$ uniformly in $\btheta \in \Theta$.
With a few additional but standard conditions, this uniform convergence ensures that the generalized posterior concentrates around minimizers of $\mathcal{L}$ as more data is observed.
Further details and theoretical implications of this requirement are discussed in \Cref{sec:theory}.
Importantly, this framework recovers standard Bayesian inference as a special case when $\hat{\mathcal{L}}_n(\theta)$ corresponds to the negative log-likelihood, which approximates the Kullback-Leibler divergence.

Generalized posteriors have been studied primarily in the context of robustness against model misspecification \citep{hooker2014bayesian, ghosh2016robust, knoblauch2018doubly, miller2019robust, husain2022adversarial}.
However, there has recently been a focus on their use to enable more efficient computation or sampling of a posterior over $\btheta$.
Notably, \emph{conjugate} generalized posteriors have been obtained through carefully chosen divergence-based losses---a strategy first proposed in \citet{matsubara2022robust}---facilitating efficient inference in complex modeling tasks such as Gaussian process regression \citep{altamirano2024robust, laplante2025robust}, online changepoint detection \citep{altamirano2023robust}, and Kalman filtering \citep{duran2024outlier}.
In particular, \citet{matsubara2022robust} obtain a loss from the KSD, which is a statistical divergence for continuous domains that does not require the computation of the model's normalization constant.
This loss is quadratic for the natural parameters of any exponential family model
{\setlength{\abovedisplayskip}{4pt}
 \setlength{\belowdisplayskip}{4pt}
\begin{equation}
    p^{\exp}_{\btheta}(\x) := \exp \left( \bm{\eta}(\btheta)^\top \mathbf{T}(\x) + B(\x) - \log Z(\btheta) \right),\label{eq:exp-family-model}
\end{equation}
}
with natural parameters $\bm{\eta} : \Theta \rightarrow \mathbb{R}^p$, sufficient statistic $\mathbf{T} : \X \rightarrow \mathbb{R}^p$, base measure $B : \X \rightarrow \mathbb{R}$, and normalization constant $Z: \Theta \rightarrow \mathbb{R}$.
Specifically, the loss becomes 
$
\hat{\mathcal{L}}^{\operatorname{KSD}}_n(\btheta) = \bm{\eta}(\btheta)^\top \bm{\Lambda}_n \bm{\eta}(\btheta) + \bm{\eta}(\btheta)^\top \bm{\nu}_n$,
where $\bm{\Lambda}_n \in \mathbb{R}^{p \times p}$ and $\bm{\nu}_n \in \mathbb{R}^p$ depend on observations $\{\x_i\}_{i=1}^n$, $B$, and $\mathbf{T}$, but \emph{not} on $Z(\btheta)$---a feature that is important for intractable models where $Z(\btheta)$ cannot be computed.
This quadratic form allows for a conjugate update on $\bm{\eta}$ with the generalized posterior in \Cref{eq:gen-bayes-posterior} when the prior is exponentially quadratic, such as a multivariate Gaussian. 
It was also shown in \citet{altamirano2023robust} that with the same model class, a quadratic loss arises from the score-matching divergence of \citet{hyvarinen2005estimation}.

These results are strictly applicable to continuous domains.
In contrast, in discrete settings, there is to date no known statistical divergence that yields a quadratic loss for exponential family models, and thus no analogous conjugate generalized posterior. 
In the next section, we introduce existing divergences for discrete intractable models, highlighting their shortcomings in achieving conjugacy.

\vspace{-1em}
\subsection{Divergences for Discrete Intractable Models} \label{sec:discrete-divergences}

Suppose that $\model$ is defined on a discrete domain $\X$ and has a normalization constant that cannot be computed. That is, 
{\setlength{\abovedisplayskip}{4pt}
 \setlength{\belowdisplayskip}{4pt}
\begin{equation}
    \model(\bm{x}) = \frac{\tilde{p}_{\bm{\theta}}(\bm{x})}{Z(\bm{\theta})}, \quad Z(\bm{\theta}) := \sum_{\bm{x} \in \mathcal{X}} \tilde{p}_{\bm{\theta}}(\bm{x}),
    \label{eq:intractable-model}
\end{equation}
}
where, unlike the intractable  $Z(\bm{\theta})>0$, $\tilde{p}_{\bm{\theta}}(\bm{x})$ is easily evaluated. 
Most common statistical divergences for models on discrete domains (such as the total variation distance,  Wasserstein distance, maximum mean discrepancy, and Kullback-Leibler divergence) unfortunately depend on the intractable $Z(\bm{\theta})$. We instead focus on divergences that bypass the computation of this constant.

\citet{hyvarinen2007some} first investigated such divergences via an estimation of models of binary data, where  $\mathcal{X} = \{-1, 1\}^d$. 
The strategy is to force the ratios $\model(\bm{x})/\model(\bm{x}_{-i})$ and $\dgp(\bm{x})/\dgp(\bm{x}_{-i})$ to be equal, where $\bm{x}_{-i} = (x_1, x_2, \dots, -x_i, \dots, x_d)$.
For a parametric model $\model$, the ratio-matching divergence proposed is given by
{\setlength{\abovedisplayskip}{4pt}
 \setlength{\belowdisplayskip}{5pt}
\begin{IEEEeqnarray*}{rl}
 D^{\operatorname{RM}}(\dgp \| \model) & \\
 := \quad & \mathbb{E}_{\x \sim \dgp} \Bigg[ \sum_{i=1}^d 
    \Bigg( g\!\left(\frac{\dgp(\x)}{\dgp(\x_{-i})} \right)
        - g\!\left(\frac{\model(\x)}{\model(\x_{-i})} \right) \Bigg)^2 + \Bigg( g\!\left(\frac{\dgp(\x_{-i})}{\dgp(\x)} \right)
        - g\!\left(\frac{\model(\x_{-i})}{\model(\x)} \right) \Bigg)^2
    \Bigg]
\\[1ex]
  = \quad & \underbrace{\mathbb{E}_{\x \sim \dgp}
    \left[\sum_{i=1}^d g^2 \!\left(\frac{\model(\x)}{\model(\x_{-i})} \right)\right]}
    _{\mathcal{L}^{\operatorname{RM}}(\btheta)} + C^{\operatorname{RM}}(\dgp)
\end{IEEEeqnarray*}
}
where $g(u) = 1 / (1 + u)$ serves as a bounded transformation that prevents numerical instability when probabilities are small; the constant $C^{\operatorname{RM}}(\dgp) \in \mathbb{R}$ does not depend on $\btheta$; and the second line follows from Theorem 1 of \citet{hyvarinen2007some}.
Importantly, $D^{\text{RM}}$ induces a loss $\mathcal{L}^{\text{RM}}$ not depending on $Z(\bm{\theta})$ and that can be estimated using only samples from $\dgp$.
Subsequent studies have proposed similar objectives for fitting unnormalized models, including energy-based and generative models \citep{Lyu2009, gutmann2011bregman, pang2020efficient, Meng2022concrete, schroder2023energy}, and have been applied to tasks such as regression \citep{xu2022generalized, gan2025generalized}, gradient estimation \citep{shi2022gradient}, and goodness-of-fit testing \citep{yang2018goodness}. 

\citet{matsubara2024generalized} first used this type of divergence to construct a generalized posterior. More precisely, they used a discrete Fisher divergence (DFD) relying on a generalization of the Fisher divergence for continuous domains (see \citet{Lyu2009}).
Assuming an ordering of the domain $\mathcal{X}$ (see Definition 1 of \citet{matsubara2024generalized} for additional details) with $x^{\pm}$ denoting the previous or next element of $x$ in the domain $\X$, and given some $h :\mathcal{X} \rightarrow \mathbb{R}$, we define the backward difference operator as $\nabla^{-} h(\x) := \left (h(\x) - h(\x^{1-}), \dots, h(\x) - h(\x^{d-}) \right)^\top$ for $\x^{j\pm }:=(x_1, \dots, x_j^{\pm}, \dots, x_d)$. Armed with this operator, the DFD is defined as
{\setlength{\abovedisplayskip}{6pt}
 \setlength{\belowdisplayskip}{6pt}
\begin{equation*}
\begin{split}
    D^{\operatorname{DFD}}(\dgp \| \model) &:= \mathbb{E}_{\x \sim \dgp} \left[ \left \| \frac{\nabla^{-} \model(\x)}{\model(\x)} - \frac{\nabla^{-}\dgp(\x)}{\dgp(\x)}  \right \|^2 \right] \\
    &= \underbrace{\mathbb{E}_{\x \sim \dgp} \left[\sum_{j=1}^d \left (\frac{\model(\x^{j-})}{\model(\x)} \right)^2 - 2\left (\frac{\model(\x)}{\model(\x^{j+})} \right)  \right]}_{\mathcal{L}^{\text{DFD}}(\btheta)} + C^{\operatorname{DFD}}(\dgp),
\end{split}
\end{equation*}
}
where $C^{\operatorname{DFD}}(\dgp)$ does not depend on $\btheta$, and where the second equality follows from Proposition 1 of \citet{matsubara2024generalized}. When constraining the model class to exponential families as in \Cref{eq:exp-family-model}, the DFD loss is given by 
{\setlength{\abovedisplayskip}{6pt}
 \setlength{\belowdisplayskip}{6pt}
\begin{IEEEeqnarray*}{rCl}
\mathcal{L}^{\mathrm{DFD}}(\btheta)
& = & \mathbb{E}_{\x\sim \dgp}\!\Bigg[
\sum_{j=1}^d
    e^{2\bm{\eta}(\btheta)^\top \Delta \mathbf{T}_j^-(\x)
    + 2 \Delta B_j^-(\x) }
    - 2 e^{\bm{\eta}(\btheta)^\top \Delta \mathbf{T}_j^+(\x)
    + \Delta B_j^+(\x)}
\Bigg].
\end{IEEEeqnarray*}
}
 with 
\(
\Delta \mathbf{T}_j^-(\x):=\mathbf{T}(\x^{j-})-\mathbf{T}(\x)
\),
\(
\Delta B_j^-(\x):=B(\x^{j-})-B(\x)
\),
\(
\Delta \mathbf{T}_j^+(\x):=\mathbf{T}(\x)-\mathbf{T}(\x^{j+})
\),
\(
\Delta B_j^+(\x):=B(\x)-B(\x^{j+})
\). 
The resulting loss is not quadratic in $\bm{\eta}(\btheta)$ and therefore does not lead to a generalized posterior conjugate with exponentially quadratic priors.
More generally, the exponentiated $\exp(-\beta n \hat{\mathcal{L}}_n^\text{DFD}(\btheta))$ does not lead to an exponential family representation; consequently, no choice of prior on $\bm{\eta}(\btheta)$ would yield conjugacy \citep{diaconis1979conjugate}.
Unfortunately, the same issue arises with all aforementioned divergences for discrete distributions.
The next section presents our approach to resolving this issue.

\vspace{-1.5em}
\section{Methodology} \label{sec:methods}
This section introduces a novel \emph{log-ratio matching divergence}, from which a loss function can be obtained to construct a generalized posterior---the \emph{LRM-Bayes} posterior. 
We demonstrate that this posterior is \emph{conjugate} for exponential family models, and discuss its hyperparameter selection.

Before doing so, we briefly introduce notation. The domain $\X$ is a countable space, and $2^{\X}$ denotes the set of all subsets of $\X$.
$\AllPMFs$ denotes the set of probability mass functions (PMFs) on $\X$.
For any $q \in \AllPMFs$, the set of PMFs whose support is contained in that of $q$ is $\AllPMFsGiven{q}:=\{p \in \AllPMFs \; | \; \operatorname{supp}(p) \subseteq \operatorname{supp}(q) \}$, where $\operatorname{supp}(q) := \{ \x \in \X : q(\x) > 0 \}$.
For $r \in \mathbb{N}$, we also denote $L^r(q, \mathbb{R}) := \{f : \X \rightarrow \mathbb{R} \;| \; \mathbb{E}_{\x \sim q}\left[|f(\x)|^r \right] < \infty \}$ the space of functions that are $r$-integrable against $q$.
\vspace{-1em}
\subsection{Log-Ratio Matching Divergence}\label{sec:log-ratio-matching-div}

To define a divergence between  $q \in \AllPMFs$ and $p \in \AllPMFsGiven{q}$, we introduce a \emph{matching set} $M : \X \rightarrow 2^{\X}$. 
Intuitively, the matching set is a set-valued function that enables \emph{local} comparisons between $p$ and $q$. 
The main properties of $M$ are that (1) $ M(\x) \subseteq \operatorname{supp}(q) \; \forall \x \in \X$ and (2) $|M(\x)|=m<\infty \; \forall \x \in \X$, where $|\cdot|$ denotes the cardinality of a set. 
While one can construct models and matching sets where $|M(\x)|$ varies across $\x$, we restrict to constant size to simplify the presentation.
\Cref{fig:matching-set-4x4} illustrates an example construction on a $4 \times 4$ lattice, relevant for settings such as Markov random fields (see \Cref{sec:experiments}).
The selection of the matching set is context-dependent, typically guided by the structure of the data, as well as by properties specific to the model.

\begin{figure}[t]
\centering
\begin{subfigure}[b]{0.32\textwidth}
\centering
\begin{tikzpicture}[
  scale=0.6, transform shape,
  node/.style={circle, draw, minimum size=4.8mm, inner sep=0pt}
]
\foreach \i in {0,1,2,3}{\foreach \j in {0,1,2}{\draw (\i,\j)--(\i,\j+1);}}
\foreach \i in {0,1,2}{\foreach \j in {0,1,2,3}{\draw (\i,\j)--(\i+1,\j);}}
\foreach \i in {0,1,2,3}{
  \foreach \j in {0,1,2,3}{
    \ifodd\numexpr\i+\j\relax
      \node[node, fill=white] at (\i,\j) {};
    \else
      \node[node, fill=black] at (\i,\j) {};
    \fi
  }
}
\end{tikzpicture}
\caption{$\x \in \X$}
\end{subfigure}
\hspace{-2cm}
\begin{subfigure}[b]{0.69\textwidth}
\centering
\begin{tikzpicture}[
  scale=0.6, transform shape,
  node/.style={circle, draw, minimum size=4.8mm, inner sep=0pt}
]

\begin{scope}[xshift=-9.5cm]
  \foreach \i in {0,1,2,3}{\foreach \j in {0,1,2}{\draw (\i,\j)--(\i,\j+1);}}
  \foreach \i in {0,1,2}{\foreach \j in {0,1,2,3}{\draw (\i,\j)--(\i+1,\j);}}
  \foreach \i in {0,1,2,3}{
    \foreach \j in {0,1,2,3}{
      \ifodd\numexpr\i+\j\relax \node[node, fill=white] at (\i,\j) {};
      \else \node[node, fill=black] at (\i,\j) {};
      \fi
    }
  }
  \node[node, fill=black] at (0,3) {};
  \draw[draw=yellow, line width=0.5mm] (0,3) circle[radius=0.25];
\end{scope}

\begin{scope}[xshift=-5cm]
  \foreach \i in {0,1,2,3}{\foreach \j in {0,1,2}{\draw (\i,\j)--(\i,\j+1);}}
  \foreach \i in {0,1,2}{\foreach \j in {0,1,2,3}{\draw (\i,\j)--(\i+1,\j);}}
  \foreach \i in {0,1,2,3}{
    \foreach \j in {0,1,2,3}{
      \ifodd\numexpr\i+\j\relax \node[node, fill=white] at (\i,\j) {};
      \else \node[node, fill=black] at (\i,\j) {};
      \fi
    }
  }
  \node[node, fill=white] at (1,3) {};
  \draw[draw=yellow, line width=0.5mm] (1,3) circle[radius=0.25];
\end{scope}

\node at (-1.3,1.5) {\scriptsize$\bullet$};
\node at (-1.0,1.5) {\scriptsize$\bullet$};
\node at (-0.7,1.5) {\scriptsize$\bullet$};

\begin{scope}[xshift=0cm]
  \foreach \i in {0,1,2,3}{\foreach \j in {0,1,2}{\draw (\i,\j)--(\i,\j+1);}}
  \foreach \i in {0,1,2}{\foreach \j in {0,1,2,3}{\draw (\i,\j)--(\i+1,\j);}}
  \foreach \i in {0,1,2,3}{
    \foreach \j in {0,1,2,3}{
      \ifodd\numexpr\i+\j\relax \node[node, fill=white] at (\i,\j) {};
      \else \node[node, fill=black] at (\i,\j) {};
      \fi
    }
  }
  \node[node, fill=black] at (3,0) {};
  \draw[draw=yellow, line width=0.5mm] (3,0) circle[radius=0.25];
\end{scope}

\end{tikzpicture}
\caption{$M(\x)$}
\end{subfigure}

\caption{\textit{A matching set on a \(4 \times 4\) lattice.} 
Here \(\X = \{-1, +1\}^{16}\), with \(-1\) shown in black and \(+1\) in white. 
(a) Example configuration \(\x \in \X\); 
(b) construction of the matching set \(M(\x)\) by flipping each yellow-highlighted node successively, yielding elements \(M_1(\x), M_2(\x), \dots, M_m(\x)\).
}
\label{fig:matching-set-4x4}
\end{figure}

We are now ready to introduce our \emph{log-ratio matching} (LRM) divergence. 
For the regularity conditions that follow, we equip each matching set $M(\x)$ with an arbitrary index set $\mathcal{J}$ of size $m$ and write $M(\x)=\{M_j(\x)\}_{j \in \mathcal{J}}$, where each map $M_j:\mathcal{X} \to \mathcal{X}$ selects one element of the matching set.
Importantly, this indexing is noninformative, does not depend on $\x$, and does not impose any ordering.
Using these maps $M_j$, we define the log-ratio operators $\RegCondOperator{j}[q](\x) := \log \frac{q(M_j(\x))}{q(\x)}$ for $ j \in \mathcal{J}$, with which we define the set of admissible PMFs for our proposed divergence.
For any reference PMF $q \in \mathcal{Q}(\mathcal{X})$, this consists of the PMFs that are positive on the support of $q$ and satisfy an $L^2$-integrability condition relative to $q$: \(
\AllAdmPMFsGiven{q} := 
\Big\{
  p \in \AllPMFsGiven{q}
  \mid
  \RegCondOperator{j}[p] \in L^2(q,\mathbb{R}) \;\forall j \in \mathcal{J}
\Big\}.
\)

\begin{definition}[The LRM Divergence]
Suppose $q \in \AllPMFs$.
The \emph{log-ratio matching (LRM) divergence} between $q,p \in \AllAdmPMFsGiven{q}$ is defined as: 
{\setlength{\abovedisplayskip}{4pt}
 \setlength{\belowdisplayskip}{4pt}
\begin{equation}
\begin{split}
        \LRM(q \| p) :&= \mathbb{E}_{\x \sim q} \left [ \frac{1}{|M(\x)|} \sum_{\x' \in M(\x)}   \left(\log \frac{p(\x')}{p(\x)} - \log \frac{q(\x')}{q(\x)} \right)^2 \right] \\[0.5em]
        &=\mathbb{E}_{\x \sim q}\left[ \frac{1}{m} \sum_{j\in \mathcal{J}} \left(\RegCondOperator{j}[p] - \RegCondOperator{j}[q] \right)^2 \right]
    \label{eq:log-ratio-div}
\end{split}
\end{equation}
}
\label{def:dfd_div}
\end{definition}

\vspace{-1.5em}
We note that these $L^2$-integrability conditions, required for both $q$ and $p$, are similar to those in \citet{matsubara2024generalized}, and are not automatically satisfied by $q$ since they depend on the choice of the matching set  $M$.
The proposed divergence $\LRM$ is reminiscent of $D^{\text{DFD}}$.
Indeed, one can even express $D^{\text{DFD}}$ in terms of matching sets: $\LRM$ compares log-ratios for arbitrary matching sets, whereas $D^{\text{DFD}}$ compares ratios for the fixed matching set $M^{\operatorname{DFD}}(\x):=\{\x^{1-}, \ldots, \x^{d-}\}$.
$\LRM$ is also closely related to the ratio matching divergence $D^{\text{RM}}$ of \Cref{sec:discrete-divergences} for the choice of $g(\cdot) = \log(\cdot)$.
Critically, this choice of $g$ enables analytical simplifications we expand on in \Cref{sec:conjugate_posterior}.
Before concerning ourselves with computational considerations, we will first impose suitable assumptions on $M$ that will facilitate a proof that $\LRM$ is a valid statistical divergence.
\begin{assumption}[Graph Connectedness]
    The matching set $M$ induces a graph $G:=(\mathcal{X}, E)$ with edges $E := \bigcup_{\bm{x} \in \mathcal{X}}\{\{\x, \x' \} : \x' \in M(\bm{x}) \}$
    which is \emph{connected}. That is, for every pair of vertices $\bm{x}, \bm{y} \in \X$ there exists a \emph{path} between them: a finite sequence of vertices $(\bm{x}_0, \bm{x}_1, \ldots, \bm{x}_k)$ with $\bm{x}_0 = \bm{x}$, $\bm{x}_k = \bm{y}$, and $(\bm{x}_{i-1},\bm{x}_i) \in E$ for all $i=1,\ldots,k$.
    \label{assumption:graph-connect}
\end{assumption}
\begin{theorem}[LRM is a Statistical Divergence]
\label{theorem:divergence}
  Suppose $q \in \AllPMFs$ and \Cref{assumption:graph-connect} holds. Then, for $p,q \in \AllAdmPMFsGiven{q}$, we have $\LRM(q \| p) = 0 \Leftrightarrow  q = p.$
\end{theorem}
The proof for \Cref{theorem:divergence} can be found in \Cref{appendix:proof-divergence} and elaborates on the mild requirement that the matching set must induce a connected graph.  This assumption is the discrete equivalent of the connected domain assumption in Theorem 1 of \cite{Zhang2022}, and is similar to assumptions used for other discrete divergences; see Theorem 1 in \cite{Meng2022concrete}.
We briefly discuss how graph connectedness is achieved in \Cref{appendix:select-matching-set-for-connected-graph}.
For clarity, we further note that any graphical-model structure imposed on $\x$ is distinct from the graph $G$ appearing in Assumption 3.1.

\vspace{-1em}
\subsection{Inference with Log-Ratio Matching}

From the log-ratio matching divergence $\LRM$, we will construct a loss function $\Lexp$ and a corresponding estimator $\Lhat$ to use for inference. Throughout, given a data-generating process $\dgp$, our parametric model of interest will be $\modelFamily$:= $\{\model \in \AllPMFs:\btheta \in \Theta\}$, with parameter space $\Theta \subseteq \mathbb{R}^p$.
Estimating the loss will require a PMF estimate $\hat{q}$ of $\dgp$, along with assumptions that will be maintained throughout the rest of this paper.
\begin{standingassumption} Inference with log-ratio matching requires the following:
    \vspace{-1em}
    \begin{sassumpenum}
    \item \label{assumption:expectation}
    The domain $\X$ is a countable space, and the data $\{\bm{x}_i \}_{i=1}^n \overset{i.i.d.}{\sim} \dgp$ for a data-generating process $\dgp \in \AllPMFs$ such that $\dgp \in \AllAdmPMFsGiven{\dgp}$.
    
    \item \label{assumption:parametric-model}
    Every model $\model \in \modelFamily$ belongs to the class
$\mathcal{Q}^{\mathrm{adm}}_{\dgp}(\mathcal{X})$.

    \item \label{assumption:pmf-estimator}
    For each sample size $n \in \mathbb{N}$, the estimate $\hat q$ of $\dgp$ satisfies $\hat q \in \mathcal{Q}^{\mathrm{adm}}_{\dgp}(\mathcal{X})$ almost surely. 
\end{sassumpenum}
\end{standingassumption}

Part \ref{assumption:expectation} of the Standing Assumption ensures that the local ratio terms are square-integrable, guaranteeing that all expectations appearing in $\LRM$ are finite.
As previously discussed, this assumption is automatically satisfied for finite sample spaces and is only mildly restrictive in countable discrete settings.
Part \ref{assumption:parametric-model} of the Standing Assumption imposes the same integrability condition for the parametric model family and excludes models that assign zero probability to possible observations.
This strict positivity requirement is met by all models considered in this paper and by many other discrete models.
Similarly, Part \ref{assumption:pmf-estimator} of the Standing Assumption requires the estimated PMF $\hat q$ to be strictly positive on the support of $\dgp$ and square-integrable under $\dgp$.
This condition is typically satisfied by standard PMF estimators such as Laplace smoothing \citep{chen1999empirical}.

We now present the log-ratio matching loss.
\begin{definition}[The LRM Loss]
    The log-ratio matching loss is given by:
    {\setlength{\abovedisplayskip}{6pt}
\setlength{\belowdisplayskip}{6pt}
    \begin{equation}
    \Lexp(\bm{\theta}) := \mathbb{E}_{\x \sim \dgp}\left[  \frac{1}{|M(\bm{x})|} \sum_{\bm{x}' \in M(\bm{x})}  \left(\log \frac{\model(\bm{x}')}{\model(\bm{x})}\right)^2 -2 \log\frac{\model(\bm{x}')}{\model(\bm{x})}\log \frac{\dgp(\bm{x}')}{\dgp(\bm{x})}  \right].
    \label{eq:expected-log-ratio-matching-loss}
\end{equation}
}
\end{definition}
The loss from \Cref{eq:expected-log-ratio-matching-loss} is straightforwardly obtained from $\LRM$ by expanding the square and dropping the term not depending on $\btheta$.
For some divergence-based losses, such as $\mathcal{L}^{\operatorname{DFD}}$ and $\mathcal{L}^{\operatorname{RM}}$, $\dgp$ only appears as the measure over which the loss' expectation is taken; therefore, these losses can be approximated with samples from $\dgp$.
However, this is not the case for LRM: estimating $\Lexp$ requires a PMF estimate $\hat{q}$ of $\dgp$, since $\dgp$ also appears inside the log-ratio term.
This reliance on an estimate of $\dgp$ has precedent in works such as \citet{jewson2018principles} and \citet{hooker2014bayesian}. 
With an estimate $\hat{q}$, we can now define the estimator of the log-ratio matching loss as follows:
{\setlength{\abovedisplayskip}{6pt}
\setlength{\belowdisplayskip}{6pt}
\begin{equation}
\Lhat(\bm{\theta}) := \frac{1}{n} \sum_{i=1}^n \frac{1}{|M(\bm{x}_i)|} \sum_{\bm{x}' \in M(\bm{x}_i)}  \left(\log \frac{\model(\bm{x}')}{\model(\bm{x}_i)}\right)^2 -2 \log\frac{\model(\bm{x}')}{\model(\bm{x}_i)}\log \frac{\hat{q}(\bm{x}')}{\hat{q}(\bm{x}_i)} \label{eq:log-ratio-matching-loss}, 
\end{equation}
}
where the dependency on the estimate $\hat{q}$ is implied from the hat notation. 
In \Cref{sec:theory}, we discuss the theoretical requirements, in particular those needed for $\hat{q}$ to ensure that $\Lhat (\btheta) \longrightarrow \Lexp(\btheta)$ almost surely pointwise and uniformly in $\btheta$.
Although not the focus of this paper, we note that the loss from \Cref{eq:log-ratio-matching-loss} could also be used in a frequentist setting,
In \Cref{appendix:theoretical-assessment}, we study some of the properties of the resulting frequentist estimator.
As we will demonstrate in \Cref{sec:conjugate_posterior}, for exponential family models, the loss function is quadratic, which would lead to a closed-form minimum distance estimator.
Additionally, the loss can be extended to include data-dependent weights, providing robustness against outliers. 
We elaborate on this extension in \Cref{appendix:robust-loss}.

Finally, the loss $\Lhat$ allows us to build the \textit{log-ratio matching generalized Bayesian (LRM-Bayes) posterior}, presented below, which has desirable theoretical properties that we elaborate on in \Cref{sec:theory}.
\begin{definition}[The LRM-Bayes Posterior]
    The LRM-Bayes posterior has density:
    {\setlength{\abovedisplayskip}{4pt}
 \setlength{\belowdisplayskip}{4pt}
    \begin{equation*}
    \hat{\pi}_M^\beta(\btheta ;\;  \{\x_i\}_{i=1}^n) \propto \exp \left (-\beta n \Lhat(\btheta) \right) \pi(\btheta).
\end{equation*}
}
\end{definition}

\subsection{Conjugate Posterior for Exponential Family Models}\label{sec:conjugate_posterior}

Similar to the work of \citet{matsubara2024generalized} (DFD-Bayes), we develop a generalized posterior for any discrete intractable model.
However, in stark contrast to their approach, when $\modelFamily$ is restricted to the exponential family model outlined in \Cref{eq:exp-family-model}, the LRM-Bayes posterior is \emph{conjugate}.
This is demonstrated in the following \Cref{prop:exp_fam}. 

\begin{proposition}[Exponential Family Models and Quadratic Losses]
    Suppose $\model \in \modelFamily$ is an exponential family, i.e. it is of the form in \Cref{eq:exp-family-model}.
    Then, 
    {\setlength{\abovedisplayskip}{4pt}
 \setlength{\belowdisplayskip}{4pt}
    \begin{align*}
    \Lhat(\bm{\theta}) = \bm{\eta}(\bm{\theta})^\top \bm{\Lambda}_n \bm{\eta}(\bm{\theta}) - 2 \bm{\eta}(\bm{\theta})^\top \bm{\nu}_n + C^{\operatorname{LRM}}(\dgp), 
    \end{align*} }
    for a constant $C^{\operatorname{LRM}}(\dgp)$ independent of $\btheta$, and
    {\setlength{\abovedisplayskip}{8pt}
 \setlength{\belowdisplayskip}{8pt}
    \begin{equation*}
    \begin{split}
        &\bm{\Lambda}_n :=  \frac{1}{n}\sum_{i=1}^n \frac{1}{|M(\bm{x}_i)|}\sum_{\bm{x}' \in M(\bm{x}_i)} \left(\mathbf{T}(\bm{x}') - \mathbf{T}(\bm{x}_i)\right) \left(\mathbf{T}(\bm{x}') - \mathbf{T}(\bm{x}_i)\right)^\top  \in \mathbb{R}^{p \times p}  \\
        &\bm{\nu}_n := \frac{1}{n} \sum_{i=1}^n \frac{1}{|M(\bm{x}_i)|} \sum_{\bm{x}' \in M(\bm{x}_i)} \left(\mathbf{T}(\bm{x}') - \mathbf{T}(\bm{x}_i)\right) \left (\log \frac{\hat{q}(\bm{x}')}{\hat{q}(\bm{x}_i)} - \left(B(\bm{x}') - B(\bm{x}_i) \right)   \right) \in \mathbb{R}^p.
    \end{split}
    \end{equation*}
    }
    For an exponentially quadratic prior 
    \(\pi(\bm{\eta}) \propto \exp (-\frac{1}{2}(\bm{\eta} - \bm{\mu} )^\top \mathbf{\Sigma}^{-1}  (\bm{\eta} - \bm{\mu} ) ),
    \)
    with $\bm{\mu} \in \mathbb{R}^p, \mathbf{\Sigma} \in \mathbb{R}^{p\times p}$, and $\mathbf{\Sigma}$ positive definite, $\beta > 0$, the LRM-Bayes posterior $\hat{\pi}_M^\beta$ on $\bm{\eta}$ simplifies to $\mathcal{N}(\bm{\mu}_n, \mathbf{\Sigma}_n)$ for
    \(
    \bm{\mu}_n := \mathbf{\Sigma}_n \left ( \mathbf{\Sigma}^{-1} \bm{\mu} + 2 \beta n  \bm{\nu}_n\right) \) and \(
        \mathbf{\Sigma}_n := \left ( \mathbf{\Sigma}^{-1} + 2 \beta n \bm{\Lambda}_n \right)^{-1} 
    \).
    \label{prop:exp_fam}
\end{proposition}

The derivation can be found in \Cref{appendix:exp_fam}. 
We note that for exponential family models, Part \ref{assumption:parametric-model} of Standing Assumption holds if, for each $M_j$ ($j=1, \dots, m$): (i) for all $i=1,\dots, p$, \(\Delta T_{i,j}:=T_i \circ M_j - T_i \in L^2(\dgp, \mathbb{R})  \), where $\mathbf{T}(\x):=(T_1(\x), \dots, T_p(\x))^\top$; and (ii) $\Delta B_j:= B \circ M_j - B \in L^2(\dgp, \mathbb{R})$, where $\circ$ denotes function composition.

\Cref{prop:exp_fam} implies significant computational gains over the DFD-Bayes posterior of \citet{matsubara2024generalized} for exponential family models, as  MCMC sampling is no longer required.
Supposing  $m \propto d$ (for instance, see \Cref{exp:intractable-models-on-lattices}), computing $\bm{\Lambda}_n$ and $\bm{\nu}_n$ incurs a naive cost $\mathcal{O}(ndp^2)$, which includes the computation of $\hat{q}$, typically requiring one pass through all samples.
Obtaining the LRM-Bayes posterior then only requires computing $\mathbf{\Sigma}_n$ and $\bm{\mu}_n$, which has complexity $\mathcal{O}(ndp^2 +p^3)$.
In contrast, the DFD-Bayes posterior is approximated through MCMC, and the complexity scales with both $p$ and the number of MCMC steps $T$, implying a cost $\mathcal{O}(ndpT)$.
For fixed $p$ and $n \gg p$, the dominant costs are $\mathcal{O}(nd)$ for LRM-Bayes and $\mathcal{O}(ndT)$ for DFD-Bayes: LRM-Bayes then effectively reduces computational scaling by a factor of $T$.
This advantage becomes increasingly pronounced as $n$ increases, since DFD-Bayes' cost scales linearly in the number of MCMC iterations $T$.
We now show that LRM-Bayes' computational advantage extends beyond conjugacy, starting with the following corollary.

\vspace{1em}
\begin{corollary}[Closed-form truncated Gaussian generalized posterior]
\label{cor:closed-form-trunc-Gauss}
    Under the conditions of Proposition 3.3, suppose that the prior is \(\pi(\bm{\eta})\propto \exp\left\{-\frac12\bm{\eta}^\top \mathbf{A}_0\bm{\eta}+\mathbf{b}_0^\top \bm{\eta} \right\} \mathbf{1}_{\mathcal C}(\bm{\eta}),
\) for some set $\mathcal C\subseteq\mathbb R^p$, symmetric positive semidefinite matrix $\mathbf{A}_0$, and vector $\mathbf{b}_0\in\mathbb R^p$. If $\mathbf{A}_0+2\beta n\bm{\Lambda}_n$ is positive definite, then the LRM posterior $\hat\pi_M^\beta$ on $\bm{\eta}$ simplifies to $\mathcal{N}_{\mathcal{C}}(\bm{\mu}_n, \mathbf{\Sigma}_n)$, the multivariate Gaussian truncated to $\mathcal{C}$, for \(\mathbf{\Sigma}_n=\left(\mathbf{A}_0+2\beta n \bm{\Lambda}_n\right)^{-1}\), \(\bm{\mu}_n= \mathbf{\Sigma}_n\left(\mathbf{b}_0+2\beta n\bm{\nu}_n\right).\)
\end{corollary}

\Cref{cor:closed-form-trunc-Gauss} shows that LRM-Bayes can retain analytic tractability and accelerate inference even when the prior is not conjugate in the classical sense.
In particular, once the LRM-Bayes mean and covariance have been computed, the resulting truncated-Gaussian posterior can be sampled efficiently using standard methods \citep{botev2017normal}.
Priors satisfying the conditions of \Cref{cor:closed-form-trunc-Gauss} include exponential, continuous Bernoulli, and uniform priors.
\Cref{appendix:add-details-univariate-CMP} illustrates the effect of these prior choices empirically using the univariate Conway-Maxwell-Poisson model of \Cref{sec:experiments}.

Beyond conjugacy and direct closed-form inference for the full parameter $\btheta$, LRM-Bayes can also accelerate inference through partial conjugacy and Gibbs or Metropolis-Hastings-within-Gibbs sampling.
Suppose that $\btheta=(\btheta_1,\btheta_2)$ and write the posterior as 
\(
    \pi_{\mathcal L}^\beta(\btheta_1,\btheta_2 ; \x) \propto \exp\{-\beta n\hat{\mathcal L}(\btheta_1,\btheta_2, \x)\} \pi(\btheta_1,\btheta_2). 
\)
One useful setting arises when the loss is quadratic in at least one of $\btheta_1$ or $\btheta_2$.
To be concrete, we take that it is quadratic in $\btheta_1$, as follows:
\(
\hat{\mathcal L}(\btheta_1,\btheta_2, \x) = \btheta_1^\top\bm{\Lambda}_n(\btheta_2)\btheta_1
-2\btheta_1^\top\bm{\nu}_n(\btheta_2)+C(\btheta_2).
\)
Provided that $\pi(\btheta_1\mid\btheta_2)$ satisfies the conditions on the prior of \Cref{prop:exp_fam} or \Cref{cor:closed-form-trunc-Gauss}, this gives a conjugate or closed-form update for $\pi_{\mathcal L}^\beta(\btheta_1 ; \btheta_2,\x)$.
If the loss is also quadratic in $\btheta_2$, the same argument provides an exact update for $\pi_{\mathcal L}^\beta(\btheta_2 ; \btheta_1,\x)$, yielding a Gibbs sampler.
Otherwise, $\btheta_2$ can be updated using Metropolis-Hastings, yielding a Metropolis-Hastings-within-Gibbs sampler.
This occurs, for example, when the loss is quadratic only in a subset of the model parameters, as in the count time-series model of \Cref{sec:experiments}.

Another important case arises when $\btheta_2$ is a prior hyperparameter or auxiliary variable that does not enter the loss.
The full conditionals then satisfy \( \pi_{\mathcal L}^\beta(\btheta_1 ; \btheta_2,\x) \propto \exp\{-\beta n\hat{\mathcal L}(\btheta_1, \x)\} \pi(\btheta_1 \mid \btheta_2), \) and \( \pi_{\mathcal L}^\beta(\btheta_2 ; \btheta_1,\x) \propto \pi(\btheta_1\mid\btheta_2)\pi(\btheta_2). \)
The first conditional is available in closed form or is conjugate under the same conditions as above, while the second is determined entirely by the prior construction.
For example, this setting includes Gaussian scale-mixture priors of the form $\btheta_1\mid\btheta_2\sim\mathcal N(0,\Sigma(\btheta_2))$, for which conditioning on $\btheta_2$ retains the Gaussian conjugacy established above.
This class includes the Laplace prior used in the Bayesian lasso, Student-$t$ priors, including the Cauchy as a special case, and the horseshoe prior, while related constructions yield spike-and-slab priors.
We consider these last two priors for the graphical Conway-Maxwell-Poisson model in \Cref{sec:experiments}.

More broadly, by replacing generic proposal tuning and rejection steps with exact Gibbs updates whenever the relevant conditionals are tractable, these approaches extend the computational benefits of LRM-Bayes to models beyond exponential families and to a vast class of priors used in practice.

\vspace{-1em}

\subsection{Hyperparameter Selection}\label{sec:calibration-selection}
To construct $\Lhat$ and the LRM-Bayes generalized posterior $\hat{\pi}_M^\beta$, we require an estimate $\hat{q}$ of $q_0$.
The empirical PMF $\empirical (\x) = \frac{1}{n}C_n(\x)$, where $C_n(\x):= \sum_{i=1}^n \delta(\x_i = \x)$, is the canonical nonparametric estimator, and can be used for LRM-Bayes; see \Cref{appendix:eval-loss-empirical-PMF} for more details.
However, it can be unstable when $n$ is small \citep{chen1999empirical}, and it does not guarantee strict positivity on all of $\X$, and would therefore break Part \ref{assumption:pmf-estimator} of the Standing Assumption.
To resolve this, we use \emph{Laplace additive smoothing} \citep[e.g., see][]{ zhai2017study} as a regularized extension of the empirical PMF.
For $\alpha \in [0, 1]$ and a base PMF $\basePMF(\x):= \tildebasePMF(\x) / \ZbasePMF$, where $\ZbasePMF := \sum_{\x \in \X} \tildebasePMF(\x) < \infty$, Laplace additive smoothing defines
\begin{equation}
    \hat{q}_{\alpha}(\x) := \frac{C_n(\x) + \alpha \tildebasePMF(\x)}{n + \alpha \ZbasePMF}.
    \label{eq:PMF-estimator}
\end{equation}
This form can be interpreted as the posterior mean of a multinomial likelihood with a Dirichlet prior and base PMF $\basePMF$. The PMF $\hat{q}_{\alpha}$ satisfies Part \ref{assumption:pmf-estimator} of the  Standing Assumption whenever the base PMF $\tildebasePMF \in \AllAdmPMFsGiven{\dgp}$.
When $\X$ is finite, we take $\basePMF$ to be uniform.
When $\X$ is countably infinite, we typically select $q^{\dagger}$ to be a mixture between a uniform distribution covering the majority of the mass of $\dgp$ and a distribution with the same support as $\dgp$; see  \Cref{appendix:select-base-PMF} for more details.
Finally, there remains to estimate the value of $\beta$.
We achieve this by following the procedure outlined in \citet{syring2019calibrating}, which aims to obtain an approximate nominal frequentist coverage probability.
This approach to estimating $\beta$ is summarized in \Cref{appendix:posterior-calibration} and adopted throughout the paper.

\vspace{-1.5em}
\section{Experiments}\label{sec:experiments}
We investigate LRM-Bayes for intractable models of count data and for Markov random fields on lattices, demonstrating that it achieves results comparable to other generalized and standard Bayes methods while incurring significantly lower computational costs thanks to conjugacy. 
In the figures, we refer to LRM-Bayes and DFD-Bayes as \textcolor[HTML]{6495ED}{\textbf{LRM}} and \textcolor[HTML]{008000}{\textbf{DFD}}.
We primarily compare our method to DFD-Bayes because, prior to the present work, it was the most computationally efficient method for discrete doubly intractable problems.
In some contexts, we also compare to a Bayesian posterior based on the pseudo-likelihood \citep[see][]{pensar2017marginal}, denoted by \textcolor[HTML]{FFA500}{\textbf{PL}}, as it is among the fastest approximate methods available.
Finally, where computationally feasible, we also compare to a standard Bayesian posterior approximated through MCMC as a reference. This will typically be based on standard Metropolis-Hastings with a truncated normalization constant, or auxiliary variable MCMC  \citep[see][]{moller2006efficient}, also known as exchange MCMC, denoted by \textcolor[HTML]{9400D3}{\textbf{MCMC-Approx}} and \textcolor[HTML]{9400D3}{\textbf{MCMC-Aux}} respectively. 
Additional details on the experiments can be found in \Cref{appendix:exp-add-details}.
The code to reproduce all experiments is available at  
 \url{https://github.com/williamlaplante/LRM-Bayes}.
All computations were performed on a 13-inch MacBook Pro (2020) with an Apple M1 processor and 8 GB unified memory.

\subsection{Intractable Models of Count Data}\label{exp:intractable-models-of-count-data}
Many prominent cases of discrete intractable likelihood arise in the context of count data. 
Replacing simpler approaches grounded in Poisson or Negative Binomial distributions with more flexible models can allow us to better capture dispersion or dependence, but also leads to intractable normalization constants.
In this subsection, we study
a popular extension of the Poisson distribution, which suffers under this problem called the Conway-Maxwell-Poisson (CMP) model \citep{benson2021bayesian, Sellers_2023, inouye2017review, piancastelli2023multivariate}. We first consider its univariate form, then extensions to graphical models and autoregressive time series.
Experimental details for this model class can be found in \Cref{appendix:intract-count-data}.

\subsubsection{Univariate Conway-Maxwell-Poisson Model}\label{sec:CMP_univariate}

We start our experiments with synthetic data from a univariate CMP, considering both the over- and under-dispersed cases from Section 4.1 of \citet{matsubara2024generalized}. This synthetic example will be used to benchmark all methods, and to demonstrate that our method is not overly sensitive to the choice of matching set and hyperparameters. The CMP with $\X = \mathbb{N} \cup \{0\}$ has probability mass function \(\model(\x) \propto (\theta_1)^x (x!)^{-\theta_2} \), where $\btheta:=(\theta_1, \theta_2) \in \Theta = (0,\infty)^2 \cup ((0,1) \times \{0 \})$ and the normalization constant is $Z(\btheta)=\sum_{x=0}^\infty (\theta_1)^x (x!)^{-\theta_2}$. The latter is an intractable infinite sum and has no analytical form, unless $\theta_2=1$, in which case the Poisson distribution is recovered.
For this model, conjugacy is obtained on $\bm{\eta}(\btheta) = (\log \theta_1, \theta_2)$, and the posterior on $\btheta$ follows in closed form by change of variables.

\begin{figure}[t]
    \centering
    \includegraphics[width=0.9\linewidth]{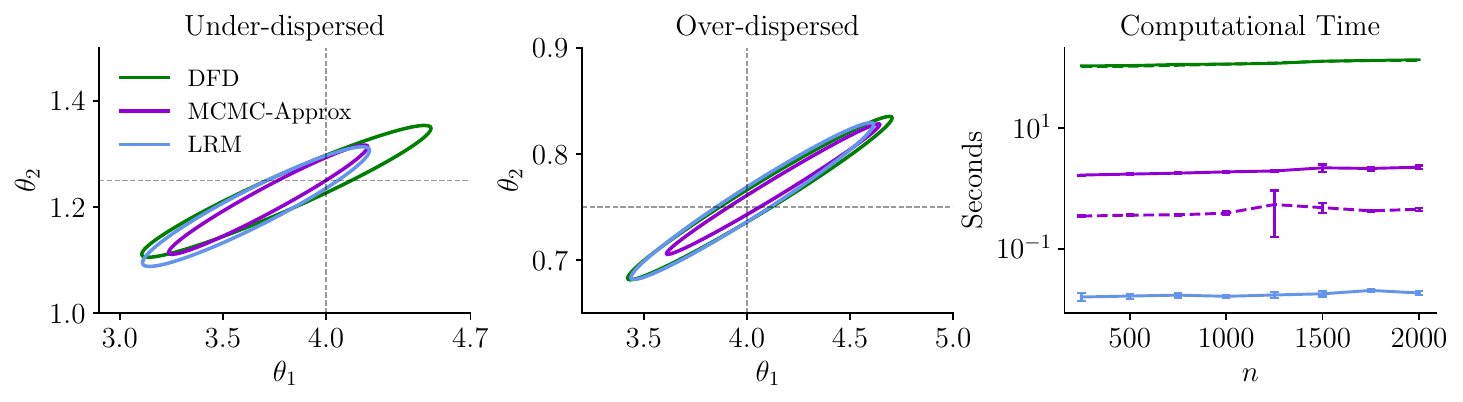}
    \caption{\textit{Posterior distributions over CMP parameters}. 
    The left and center panels show the 95\% credible regions. 
    The posterior samples for both Bayes and DFD-Bayes were approximately Gaussian, which justifies representing the credible regions as ellipses.
    The true parameter values are indicated by dotted lines. 
    The right panel reports computational time as a function of $n$, averaged over 10 runs. 
    In this panel, for MCMC-Approx and DFD-Bayes, dotted lines correspond to $1000$ MCMC samples and solid lines to $5000$ samples.
    Computational cost includes estimating $\beta$ for DFD-Bayes and LRM-Bayes.
    The estimates of $\beta$ are $0.56, 1.38$ with empirical coverages $94\%, 92\%$ for the under-dispersed and over-dispersed cases, respectively.
    }
    \label{fig:CMP-1d}
\end{figure}

We will compare three methods: LRM-Bayes, DFD-Bayes, and standard Bayes with approximate MCMC, with all methods sharing the same multivariate normal prior.
For Bayes and DFD-Bayes, we run 5000 MCMC samples with a very large number of burn-in steps \citep[as in][]{matsubara2024generalized}.
For LRM-Bayes, the default matching set will be  $M(x) = \{x+1\}$, which leads to $\RegCondOperator{1}[\model] = \log \theta_1 - \theta_2 \log (1+x)$.
In \Cref{appendix:select-matching-set-for-connected-graph}, we discuss how graph connectedness for this matching set is satisfied.
Further, for Part \ref{assumption:parametric-model} of the Standing Assumption to hold, we require $\mathbb{E}_{x \sim \dgp}[\log (1 + x)^2] < \infty$, which is extremely weak and, for example, satisfied by any distribution $q_0$ with finite mean.
The dataset has $n=2000$ samples and the results are provided in Figures \ref{fig:CMP-1d} and \ref{fig:sensitivity-alpha-neighbours-1d-cmp}.
Details of the experiment can be found in \Cref{appendix:add-details-univariate-CMP}.

In \Cref{fig:CMP-1d}, the left and middle panels show the 95\% posterior credible regions obtained by each method, while the right panel reports computational time as a function of $n$. 
The benchmarking results indicate that all three methods yield nearly identical posterior credible regions. 
However, their computational costs differ substantially.
At $n=2000$, Bayes required approximately 0.45 seconds with 1000 MCMC samples and 2.25 seconds with 5000 samples, LRM took about 0.02 seconds, while DFD-Bayes required roughly 131.5 seconds (1000 samples) and 135.2 seconds (5000 samples).
Note that DFD-Bayes' cost is dominated by the task of estimating $\beta$. 
At modest sample sizes, LRM is then substantially faster than Bayes (approximately 23$\times$ at 1{,}000 draws and $>100 \times$ at 5{,}000 draws), and orders of magnitude faster than DFD-Bayes ( $>6000 \times$).
\begin{figure}[t]
  \centering
  \begin{subfigure}[t]{0.33\textwidth}
    \centering
    \includegraphics[width=0.9\linewidth]{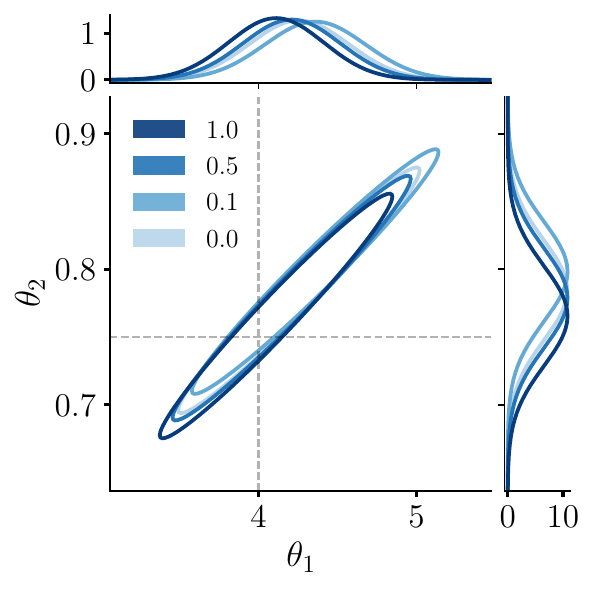}
    \caption{Varying $\alpha$}
    \label{fig:alpha-vs-posterior}
  \end{subfigure}\hfill
  \begin{subfigure}[t]{0.33\textwidth}
    \centering
    \includegraphics[width=0.9\linewidth]{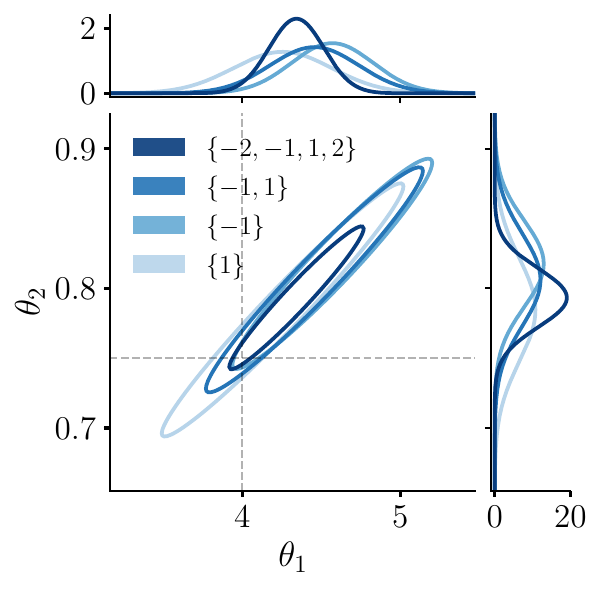}
    \caption{Varying $M$}
    \label{fig:neighbours-vs-posterior}
  \end{subfigure}
  \begin{subfigure}[t]{0.33\textwidth}
    \centering
    \includegraphics[width=0.9\linewidth]{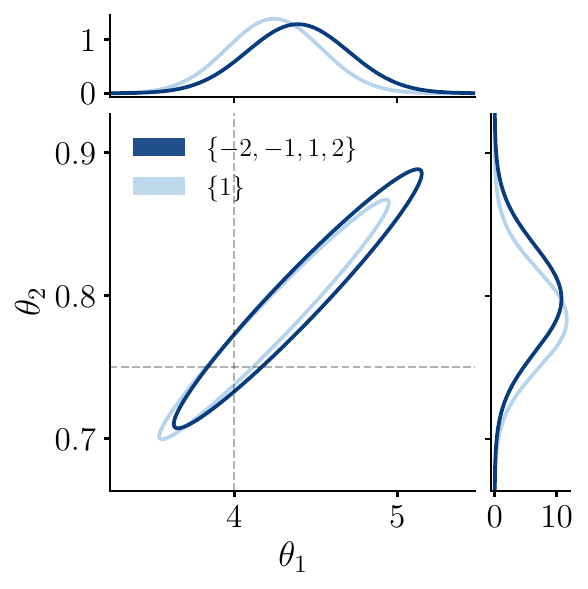}
    \caption{Vary $M$ with estimated $\beta$}
    \label{fig:neighbours-vs-posterior-calibrated}
  \end{subfigure}
  \caption{\textit{LRM 95\% credible regions for a 1d CMP model.} The true parameters are shown in dotted lines. 
  \Cref{fig:alpha-vs-posterior} shows how the posterior varies with $\alpha$, and uses $M(\x):=\{x+1 \}$. 
  \Cref{fig:neighbours-vs-posterior} shows how the posterior varies with $M(\x)$ (the labels are shorthand for $\{x\pm i\}$; for example, $\{-1, 1\} \mapsto \{x-1, x+ 1\}$), and uses $\alpha=0.0$.  
  \Cref{fig:alpha-vs-posterior} and \Cref{fig:neighbours-vs-posterior} both fix $\beta=1$; in \Cref{fig:neighbours-vs-posterior-calibrated}, $\beta$ is instead estimated.}
  \label{fig:sensitivity-alpha-neighbours-1d-cmp}
\end{figure}

In \Cref{fig:sensitivity-alpha-neighbours-1d-cmp}, we then conduct a sensitivity analysis of the proposed method with respect to $\alpha$ and $M$, holding all other factors constant.
To do so, we plot the 95\% credible region of the LRM-Bayes posterior for $(\theta_1, \theta_2)$.
We observe in \Cref{fig:alpha-vs-posterior} that the LRM-Bayes posterior is not very sensitive to $\alpha$, because in this setting, $n$ is large enough to dampen the effect of $\alpha$-smoothing.
Relative to changes in $\alpha$, \Cref{fig:neighbours-vs-posterior} shows that varying $M$ leads to greater variability across posteriors, and a larger $|M|$ is associated with tighter or narrower posteriors before estimating $\beta$.
Unlike $\alpha$, which has a limited effect as $n$ increases, the choice of $M$ can affect the region where the posterior concentrates. 
However, \Cref{fig:neighbours-vs-posterior-calibrated} shows that estimating $\beta$ helps reduce variability induced by different choices of $M$.

\subsubsection{Graphical Model for Breast Cancer Data}\label{sec:graphical-model-breast-cancer}

Next, we apply a CMP graphical model to a breast cancer gene expression dataset previously studied in \citet{matsubara2024generalized} and \citet{inouye2017review}.
This is a substantially more challenging task compared to the univariate case: the data consist of $n=878$ patients across $d=10$ dimensions, with a parameter space of size $p=65$.
In this regime, standard Bayesian inference is effectively computationally infeasible, since estimating the normalizing constant requires truncated summations whose complexity grows combinatorially with $d$ (approximately $K^{10}$ terms for truncation level $K$).
We show that despite this, LRM-Bayes provides a reliable posterior distribution at very low cost compared to state-of-the-art methods.
The CMP graphical model we consider is given by:
{\setlength{\abovedisplayskip}{6pt}
 \setlength{\belowdisplayskip}{6pt}
\begin{equation*}
\begin{split}
    \model (\bm{x}) \propto \exp \left (\sum_{i=1}^d \theta_i x_i - \sum_{i=1}^d \sum_{i < j} \theta_{i,j} x_i x_j - \sum_{i=1}^d \theta_{0,i} \log(x_i!)\right) 
    \label{eq:CMP-model}
\end{split}
\end{equation*}
}
where $\bm{x} \in \mathbb{N}_0^d$,  $\theta_i \in \mathbb{R}$, and $\theta_{i,j}, \theta_{0,i} \in \mathbb{R}^+$.
The dimensionality of $\Theta$ scales with $d$ as $p = 2d + d(d-1)/2$.
The model can be written as an exponential family with \(\bm{\eta}(\btheta) := \left((\theta_i)_{i=1}^d,(\theta_{i,j})_{i <j}, (\theta_{0,i})_{i=1}^d\right)^\top \), \(\mathbf{T}(\x):=\left(( x_i)_{i=1}^d, (-x_i x_j)_{i<j},  (-\log(x_i!))_{i=1}^d \right)^\top\), and $B(\x)=0$, where $(\cdot)_{i=1}^n$ corresponds to vector notation. 
Note that the conjugate posterior on $\bm{\eta}$ for this model is a \emph{truncated} multivariate normal due to the constraints on $\btheta$. We therefore use the fast minimax tilting sampler of \citet{botev2017normal} to draw samples from this truncated Gaussian posterior.

Since standard Bayes is infeasible, we compare LRM-Bayes with DFD-Bayes. 
For LRM-Bayes, the CMP  model must satisfy Part \ref{assumption:parametric-model} of the Standing Assumption, which holds if $\Delta T_{i,j} \in L^2(\dgp, \mathbb{R})$ for all $i=1, \dots, p$ and $j=1,\dots, m$ ($B=0$, so $\Delta B \in L^2(\dgp, \mathbb{R})$ automatically).
In the lower-sample-size, higher-dimensional setting of this experiment, we use a larger matching set $M(\x)$ to reduce variability in the LRM estimating terms, helping stabilize inference.
Specifically, $M(\x)$ has elements $M_j$ defined as all points obtained by adding offsets $\{-2,-1,1,2\}$ to a single coordinate of $\x \in \mathcal{X}$ (one coordinate at a time).
Then, for $\Delta T_{i,j} \in L^2(\dgp, \mathbb{R})$ to hold, we require $\dgp$ to have finite second moment $\mathbb{E}_{\x \sim \dgp}[x_i^2] < \infty$ for all $i = 1, \dots, d$.

\begin{figure}[t]
    \centering
    \includegraphics[width=\linewidth]{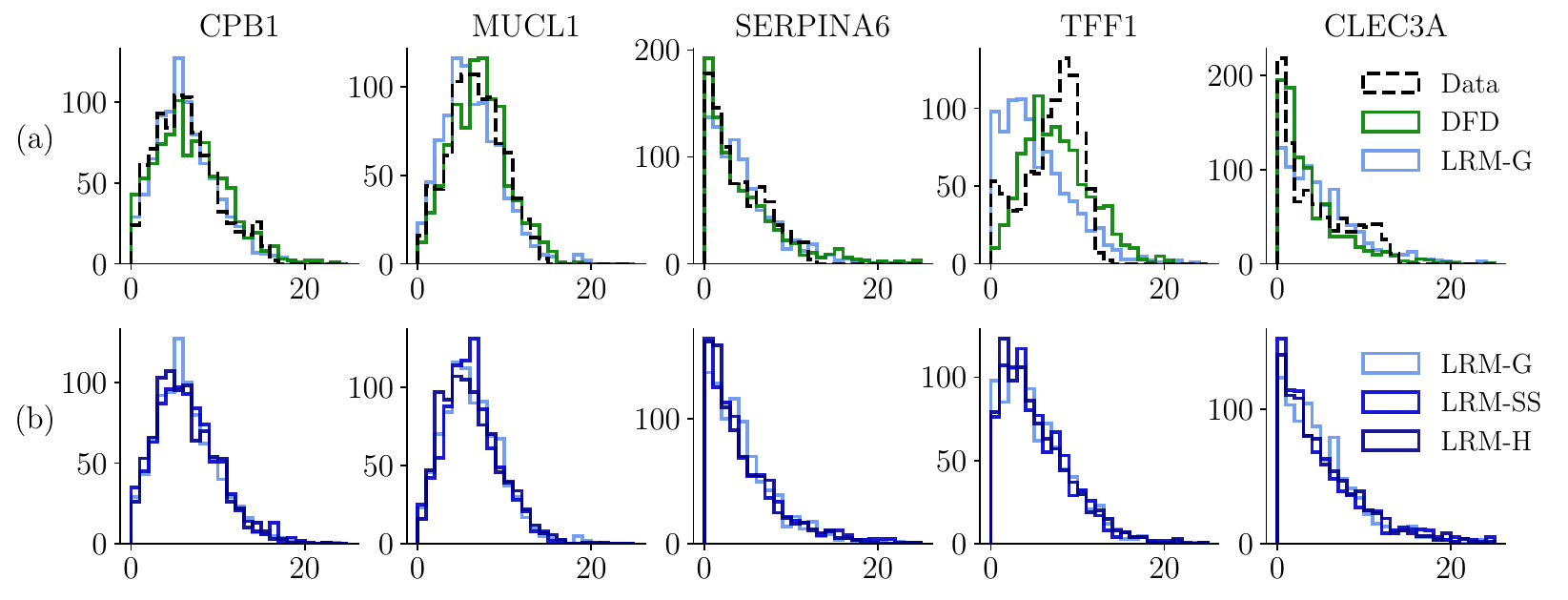}
    \caption{\textit{Posterior predictive distributions for the multivariate Conway-Maxwell-Poisson model.}
    Each column shows one of the 10 observed dimensions in the breast cancer dataset; the remaining 5 are reported in \Cref{fig:CMP-graphical-model-appendix} of \Cref{appendix:add-details-graph-model-breast-cancer}.
    Dashed histograms denote the observed data, while the remaining histograms show posterior predictive samples from DFD-Bayes and LRM-Bayes.
    LRM-G, LRM-H, and LRM-SS refer to LRM posterior predictives under Gaussian, horseshoe, and spike-and-slab priors, respectively.
    Panel~(a) shows that DFD-Bayes and LRM-Bayes agree with the data, while panel~(b) compares LRM-Bayes predictives under shrinkage priors.
    The attained empirical coverage is at $96\%$ with an estimate of $\beta$ of 1.21.
    }
    \label{fig:CMP-graphical-model}
\end{figure}

The posterior predictive results are shown in \Cref{fig:CMP-graphical-model}, and experimental details are in \Cref{appendix:add-details-graph-model-breast-cancer}.
Across most dimensions, as shown in panel (a), the two methods yield nearly identical posterior predictive distributions; the only clear discrepancy occurs in TFF1, where the observed data are bimodal and therefore poorly captured by the CMP family.
This suggests that the predictive differences stem from model misspecification rather than from the inference method itself.
Computationally, our method provides a substantial gain: estimating $\beta$ for the LRM posterior requires about $48 \pm 2$ seconds, with only $2.2 \pm 0.1$ seconds for posterior computation, while \citet{matsubara2024generalized} report $1896$ seconds ($\approx 31.6$ minutes).
Although runtimes are not perfectly comparable, since DFD-Bayes' MCMC costs depend on the number of samples, chains, and hardware, the order-of-magnitude gap ($>35 \times$ faster) demonstrates that LRM offers a significant computational advantage.

In panel (b), we extend the experiment to LRM-Bayes with horseshoe and spike-and-slab priors, which admit Gaussian scale-mixtures and thus preserve exact Gibbs updates.
Computationally, Gibbs-based LRM-Bayes remains highly efficient: inference took roughly 330 seconds and 230 seconds with the horseshoe and spike-and-slab priors, respectively, for 5000 retained samples after 1000 burn-in iterations.
Spike-and-slab LRM-Bayes induces substantial sparsity and shrinkage: \(44.8\%\) of interaction-parameter draws are exactly zero on average, and posterior means are reduced by \(44.5\%\) in median absolute value relative to the Gaussian-prior posterior.
The most frequently excluded interactions are SCGB2A2-SCGB1D2, SCGB2A2-TFF1, SCGB2A2-UGT2B11, CLEC3A-SERPINA6, and SCGB2A2-CLEC3A, each set to zero in more than half of posterior draws.
Despite these differences, panel~(b) of \Cref{fig:CMP-graphical-model} shows that the posterior predictive distributions remain similar across priors, suggesting that the impact of the prior on the predictives was limited in this experiment proposed by \citet{matsubara2024generalized}.
This experiment shows that LRM-Bayes can accommodate non-Gaussian priors, induce sparsity, and retain efficient posterior updating through exact Gibbs sampling.

\subsubsection{Autoregressive Count Time Series Model for Crime Data} \label{sec:autoregressive-model-crime-data}
This last CMP-based experiment investigates a time-series model for discrete count data.
Such models have been studied extensively in the literature, with notable applications in econometrics, finance and sociology
\citep{neal2007mcmc, fried2015retrospective, drovandi2016exact, chen2016generalized, chen2017bayesian, weiss2022softplus, huang2025exponential}. 
For this model class, we highlight the broad applicability of the LRM-Bayes method and demonstrate that it admits partial conjugacy, which can lead to substantial computational benefits for  MCMC methods and well beyond closed-form computations.
To this end, we focus on a CMP variant of the popular Poisson integer-valued GARCH model 
\citep[see][for example]{weiss2022softplus, chen2016generalized, huang2025exponential}; the former was also introduced in \citet{huang2025exponential}. 
The model is:
{\setlength{\abovedisplayskip}{4pt}
 \setlength{\belowdisplayskip}{4pt}
\begin{equation}
    x_t \mid x_{1:t-1} \sim \text{CMP}(\lambda_t, \theta_3), \quad\quad
    \log \lambda_t = \theta_1 + \varphi \log \lambda_{t-1} + \theta_2 \log (1 + x_{t-1}) 
    \label{eq:INGARCH-CMP-model}
\end{equation}
}
for some fixed initial $\lambda_0>0$  and $t \in [0, T]$.
The complete model $p_{(\btheta, \varphi)}(x_t \mid x_{1:t-1})$ is \emph{not} in the exponential family class since
{\setlength{\abovedisplayskip}{4pt}
 \setlength{\belowdisplayskip}{4pt}
\[
\log \lambda_t = \varphi^t \log\lambda_0 +  \theta_1 \sum_{k=0}^{t-1} \varphi^k + \theta_2 \sum_{k=0}^{t-1}\varphi^{t-k-1} \log(1 + x_k) := c_t + \theta_1 b_t + \theta_2 a_t,
\]
}
where $c_t, b_t, a_t$ depend on $\varphi$.
However, for fixed $\varphi$, $\model (x_t \mid x_{1:t-1})$ is an exponential family model with $\bm{\eta}(\btheta) = (\theta_1, \theta_2, -\theta_3)^\top$,  $\mathbf{T}(x_t; x_{1:t-1}, \varphi) = (b_tx_t, a_t x_t , \log(x_t!))^\top $, and $B(x_t; x_{1:t-1}, \varphi) = c_t x_t$.
This enables partial conjugacy, which in turn allows the usage of algorithms such as Metropolis-within-Gibbs.
In contrast, even for the Poisson variant, standard Bayesian updating with a Gaussian prior does not produce a conjugate posterior. 
We apply LRM-Bayes to a dataset previously studied by \citet{chen2017bayesian} on the monthly number of sexual offenses from the New South Wales Bureau of Crime Statistics and Research, covering January 1995 to June 2025.
This dataset of size $n=366$ is important as it helps identify shifts in offense incidence, which is relevant for policy response.
We compare LRM-Bayes against Approximate MCMC.
For LRM-Bayes, we use a Metropolis-within-Gibbs sampler, where $\btheta$ is updated in a conjugate Gibbs step, while $\varphi$ is updated using a random-walk Metropolis-Hastings (MH) step conditional on $\btheta$.
The approximate MCMC baseline for standard Bayes uses random walk MH with the full model for $(\btheta, \varphi)$.
For both methods, we sample until chains have mixed well, which is assessed via the Gelman–Rubin statistic \citep{gelman1992inference}.
The vanilla MH algorithm for approximate MCMC takes approximately $100,000$ samples before mixing is achieved.
The chains for LRM-Bayes mix well within $1000$ samples.
For this model, Part \ref{assumption:expectation} of the Standing Assumption is violated because the data are conditionally, rather than marginally, i.i.d. This is standard in time-series settings, where inference is based on the conditional likelihood \citep{hamilton1994time}, so it does not affect our empirical analysis.
Furthermore, satisfying Part \ref{assumption:parametric-model} of the Standing Assumption amounts to the same as fulfilling the same conditions as those for the 1D CMP model previously outlined.

\begin{figure}[t]
  \centering
  \begin{subfigure}[t]{0.42\textwidth}
    \vspace{0pt} 
    \centering
    \includegraphics[width=\linewidth]{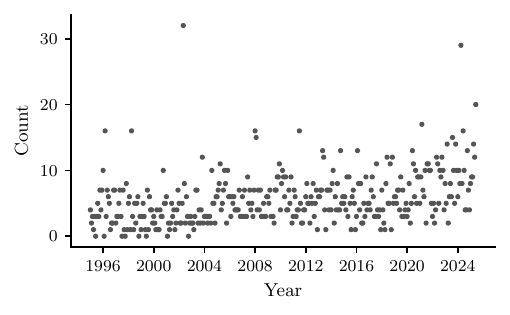}
  \end{subfigure}
  %\hfill
  \hspace{-0cm} 
  \begin{subfigure}[t]{0.45\textwidth}
    \vspace{0pt} 
    \centering
    \scriptsize
    \setlength{\tabcolsep}{2pt}
    \renewcommand{\arraystretch}{1.435}
    \begin{tabular}{@{}lcccc@{}}
    \toprule
    & \multicolumn{2}{c}{\textbf{Full Model}} 
    & \multicolumn{2}{c}{\textbf{Partial Model}} \\
    \cmidrule(lr){2-3} \cmidrule(l){4-5}
    & \tiny{\textbf{LRM}} 
    & \tiny{\textbf{MCMC-Approx}} 
    & \tiny{\textbf{LRM}} 
    & \tiny{\textbf{MCMC-Approx}} \\
    \midrule
    \midrule
    $\theta_1$ & $0.42\,(0.10)$ & $0.44\,(0.10)$ & $0.42\,(0.10)$ & $0.39\,(0.08)$ \\
    $\theta_2$ & $0.08\,(0.05)$ & $0.09\,(0.03)$ & $0.08\,(0.05)$ & $0.09\,(0.03)$ \\
    $\theta_3$ & $0.34\,(0.05)$ & $0.37\,(0.04)$ & $0.34\,(0.04)$ & $0.36\,(0.04)$ \\
    $\varphi$  & $0.0\,(0.01)$  & $0.0\,(0.10)$  & $-$           & $-$ \\
    \bottomrule
    \end{tabular}
  \end{subfigure}
  \vspace{-0em}
  \caption{Sexual offense data (left) and posterior mean and standard deviation (right) from MCMC-Approx and LRM for the full and partial models (i.e., fixed $\varphi=0$). The estimate of $\beta$ is 0.87 with empirical coverage $94\%$.
  }
  \label{fig:crim-offence-analysis}
\end{figure}

The results for the sexual offenses dataset are shown in \Cref{fig:crim-offence-analysis}, where we plot the data on the left and report on the right under the column ``full Model'' the posterior means of LRM and approximate MCMC, along with their standard deviations.
Details of the experiment can be found in \Cref{appendix:autoregressive-model-crime-data}.
In total, MCMC takes about 20 minutes per chain, whereas LRM takes roughly $1$ minute.
The posterior estimates of $(\btheta, \varphi)$ from LRM and approximate MCMC closely agree, with similar levels of uncertainty, but LRM achieves this roughly $20 \times$ faster than MCMC. 

Finally, we compute as an additional baseline the conjugate LRM-Bayes posterior and approximate MCMC for the partial model with $\varphi=0$.
This is motivated by the fact that the posterior mean for $\varphi$ is zero for both LRM and MCMC-Approx, and so a simplified model could have been used. 
In this case, partial LRM takes $\lesssim  0.1$ seconds, whereas approximate MCMC takes about 2 minutes before the MCMC chains have mixed.
In this partial model setting, LRM is therefore $>1200 \times$ faster than MCMC-Approx on the full model, and both posteriors also concentrate on the same values.

\subsection{Intractable Models on Lattices}\label{exp:intractable-models-on-lattices}
We now move to discrete \emph{Markov random fields}, another well-established class of generally intractable models.
We demonstrate that our method performs well in these settings, which are notoriously difficult for standard Bayesian inference reliant on conventional MCMC-based approaches \citep[see][]{moller2006efficient, park2018bayesian, middleton2020unbiased}. 
Details for the experiments that follow can be found in \Cref{appendix:intract-on-lattices}.

Suppose the domain $\mathcal{X}$ is a square lattice made of sites $j= 1, \ldots, d$, where at each site, we observe $x_j$, which takes values in a finite set of states $ \mathcal{S}$.
Therefore, $\x = (x_1, \dots, x_d) \in \X = \mathcal{S}^d$.
By the Markovian assumption, each site is solely dependent on its nearest neighbors.
We denote this with $\text{nb}(x_j):= \{x_{j'}: j' \text{ is a neighboring site of } j \}$. The Markov random field model is given by:
{\setlength{\abovedisplayskip}{5pt}
 \setlength{\belowdisplayskip}{5pt}
\begin{equation}
    p_{\bm{\theta}}(\bm{x}) = \frac{1}{Z(\bm{\theta})} \exp \left (  \theta_1 \sum_{j =1}^d \psi(x_j)  + \theta_{2} \sum_{x_{j'} \in \text{nb}(x_{j})}  \; \phi(x_j, x_{j'}) \right),\label{eq:mrf-model}
\end{equation}
}
where $\psi: \mathcal{S} \rightarrow \mathbb{R}$ denotes a scalar function on the sites and $\phi : \mathcal{S} \times \mathcal{S} \rightarrow \mathbb{R}$ over any two sites. Here, $\mathbf{T}(\x):= (\sum_{j=1}^d \psi(x_j), \sum_{ x_{j'} \in \text{nb}(x_j) } \phi(x_j, x_{j'}))^\top$ and $\bm{\eta}(\btheta) = (\theta_1, \theta_2)^\top$. 
Because $|\X| < \infty$, Part \ref{assumption:parametric-model} of the Standing Assumption is satisfied so long as $\psi, \phi$ are bounded.

For this data structure, the matching set of $\x$ is typically obtained by iteratively altering states at each site.
That is, $M(\x):=\{\x' : \x'_{-j} = \x_{-j} \; \text{and} \; x'_j = s \; \text{for} \; s \in \mathcal{S} \setminus \{x_j \}, \; \text{for} \; j=1,...,d \}$, such that for any $\x$, $|M(\x)| = d \cdot (|\mathcal{S}|-1)$.
Note that with this construction, the graph connectedness condition of \Cref{theorem:divergence} is directly satisfied.

These models are of particular interest for this paper because $\Lexp$ simplifies greatly when we make similar assumptions as the model $p_{\btheta}$ does. 
First, suppose that for $\x \in \X$, we have $\x' \in M(\x)$ such that $\x'_{-k} = \x_{-k}$ and $x_k \neq x'_k$, following the matching set convention previously defined. Then, the model's log-ratio reduces to:
{\setlength{\abovedisplayskip}{4pt}
 \setlength{\belowdisplayskip}{4pt}
\begin{equation*}
    \log \frac{p_{\bm{\theta}}(\bm{x}')}{p_{\bm{\theta}}(\bm{x})} = \theta_1 \left (\psi (x_k') - \psi(x_k) \right) + \theta_{2} \sum_{ x_j \in \text{nb}(x_k) }  (\phi(x'_k, x_j) - \phi(x_k, x_j) ),
\end{equation*}
}
which is linear in $\btheta=(\theta_1, \theta_2)$. 
We now assume a Markov blanket for $\dgp$ as well, i.e. $\dgp(x_k \mid \x_{-k}) = \dgp(x_k \mid \text{nb}(x_k))$.
Then, for $\x, \x'$ defined as before, the log-ratio of $\dgp$ is
{\setlength{\abovedisplayskip}{4pt}
 \setlength{\belowdisplayskip}{4pt}
\[
\log \frac{\dgp(\bm{x}')}{\dgp(\bm{x})} = \log \frac{\dgp (x'_k \mid \x_{-k}) \dgp(\x_{-k})}{\dgp (x_k \mid \x_{-k})\dgp(\x_{-k})} = \log \frac{\dgp(x_k' \mid \text{nb}(x_k)) }{\dgp(x_k \mid \text{nb}(x_k))}.
\]
}
The Markov property for $\dgp$ holds when the model in \Cref{eq:mrf-model} is well specified, but may or may not hold for misspecified models.
The model also assumes stationarity, meaning all sites share the same conditional distribution. 
These two assumptions are crucial for LRM-Bayes: rather than estimating the full conditional $\dgp(x_k \mid \x_{-k})$, which becomes exponentially harder as $d$ grows and typically requires multiple observed lattices, we can instead target the local conditional $\dgp(x_k \mid \text{nb}(x_k))$.
Consequently, $\Lexp$ can be estimated from a single observed lattice (or more, if available). 
Further details on this simplification can be found in \Cref{appendix:LRM-for-MRFs}.

\subsubsection{Ising Model}\label{exp:Ising-model}

 We first study the joint parameter estimation problem for \emph{Ising models} \citep[e.g., see][]{ghosal2020joint}.
Ising models were initially developed to study magnetism and later applied in many other scientific fields such as neuroscience \citep{das2014highlighting},  geophysics \citep{ma2019ising}, as well as for modeling public opinion survey data \citep{Avalos-Pacheco2025} or social network data  \citep{bhattacharya2018inference}.
We study this problem to benchmark LRM-Bayes, and examine the impact of grid dimensionality. 
The model is specified by \Cref{eq:mrf-model} with $\mathcal{S} = \{-1, +1 \}$,  $\phi(x_i, x_j) = x_i x_j$ and $\psi(x_i) = x_i$.

We compare performance for increasing grid sizes.
For each grid size, we conduct inference given a single grid observation from the Ising model with fixed parameter value $\btheta = (0.15, 0.30)$, and we compute posterior means for each of the methods mentioned below.
We then repeat this $20$ times and report, for each grid size and method, the mean and standard deviation of these estimates, together with their computational cost.
The methods we focus on are DFD, LRM and pseudo-likelihood (PL). 
Auxiliary MCMC is omitted because it is computationally prohibitive in this setting, and a reliable ground truth is already available.
To generate Ising-model datasets, we run Gibbs sampling and verify convergence by monitoring the trace of the average magnetization $m_t = \frac{1}{d}\sum_{i}x_i^{(t)}$.
For the DFD and PL, we draw 5,000 MCMC samples after an adequate burn-in period.

The results are shown in \Cref{fig:ising-model}, with details outlined in \Cref{appendix:Ising-model}.
The experiment demonstrates that LRM can recover the Ising model parameters as well as other methods while being substantially faster: even when including the cost of estimating $\beta$, LRM is on average $10\times$ faster than PL and DFD-Bayes with fixed $\beta$. 
Since the model is well-specified and thus the Markov and stationarity assumptions hold, increasing the grid size provides more local information, leading to lower uncertainty in the average of posterior means.

\begin{figure}[t]
    \centering
    \includegraphics[width=1.0\linewidth]{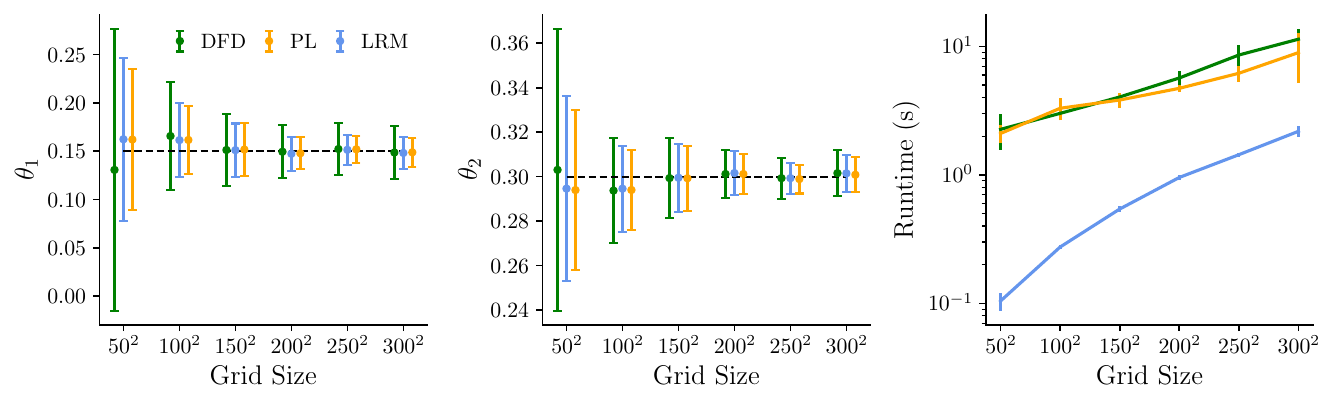}
    \caption{\textit{Posterior mean for the Ising model parameters}. 
    For each grid size and method, we run 20 independent Ising model simulations and compute one posterior mean per simulation, using MCMC sampling for DFD and PL, and the conjugate posterior for LRM. The points show the empirical average of these posterior means, and the error bars show their variability across simulations. The dashed horizontal lines indicate the true parameter values.
    We also show the runtime per grid size in the right-most plot. 
    Importantly, the runtime for DFD does \emph{not} include the estimation of $\beta$, whereas LRM \emph{does}.  
    }
    \label{fig:ising-model}
\end{figure}

\subsubsection{Potts Model}\label{exp:Potts-model}

As a second case study of a Markov random field, we investigate how LRM performs for estimating the temperature parameter of a Potts model.
The Potts model has been used to model protein sequences \citep{ekeberg2013improved}, or for image segmentation \citep{Rosu2015, chakraborty2022ordered}, where inferring the inverse temperature parameter is of interest as it governs the strength of spatial cohesion \citep{moores2020scalable}.
The model is specified via \Cref{eq:mrf-model} with $\theta := \theta_2$, $\theta_1 := 0$, $\phi(x_i, x_j) = \delta(x_i = x_j)$.

We analyze this model in the context of an Antarctic ice-thickness dataset from \textit{Bedmap2} \citep{fretwell2013bedmap2}.
The methods we investigate are auxiliary MCMC, pseudo-likelihood MCMC, DFD-Bayes and LRM-Bayes.
MCMC methods retrieve 2000 samples with a sufficient number of burn-in steps.
This dataset was previously studied by \citet{lee2024steingradientdescentapproach}, where the same model was used.

We plot the data we investigate in \Cref{fig:a-data-ice-bedsheet}, and show the posteriors produced by each method in \Cref{fig:b-bedmap-fit-lrm-etc}; further details of the experiment can be found in 
\Cref{appendix:Potts-model}.
Both PL and DFD are biased toward higher $\theta$ relative to auxiliary MCMC, while LRM yields a posterior that lies roughly between PL and auxiliary MCMC.

LRM is substantially more computationally efficient, requiring \(\sim0.5\) seconds including the estimation of \(\beta\), compared to \(\sim20\) seconds for PL and DFD with the latter fixing \(\beta\), and \(>35\) minutes for auxiliary MCMC.
The disagreement between the competing posteriors suggests that the fitted homogeneous Potts model is misspecified for this data.
This is especially evident in Figure~7a: local regions exhibit different degrees of spatial smoothness, so the data are not homogeneous.
Under such misspecification, there is no general criterion under which one posterior update should be preferred; the choice of update depends on the inferential or predictive feature one wishes to prioritize.
This experiment therefore demonstrates the speed of LRM-Bayes, while showing that its posterior remains broadly in line with competing methods in a misspecified setting.
We complement this analysis with a simulated Potts model experiment in \Cref{appendix:add-exp-potts}, where we show that LRM-Bayes' posterior covers the ground truth across different regimes induced by \(\theta\).

\begin{figure}[t]
    \centering
    \begin{subfigure}[t]{0.36\textwidth}
        \centering
        \includegraphics[width=0.9\textwidth]{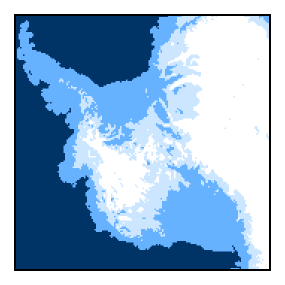}
        \caption{\textit{Antarctic ice sheet dataset.} We plot ice thickness for a $171 \times 171$ lattice. The thickness, $y$, is classified into four categories: $x_i=0$ for no ice, $x_i=1$ for $0 <y \leq 1000$, $x_i=2$ for $1000 < y \leq 2000 $, $x_i=3$ for $y>2000$.}
        \label{fig:a-data-ice-bedsheet}
    \end{subfigure}
    \hfill
    \begin{subfigure}[t]{0.6\textwidth}
        \centering
        \includegraphics[width=\textwidth]{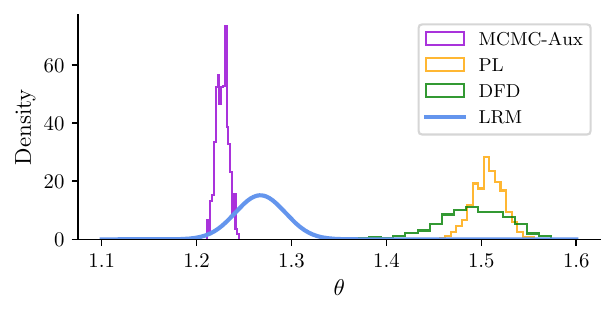}
        \caption{\textit{Posterior results for ice thickness data.} We show the auxiliary variable MCMC (MCMC-Aux), which displays similar results as in \citet{lee2024steingradientdescentapproach}. PL and DFD-Bayes MCMC samples are also displayed, along with the conjugate LRM-Bayes posterior.}
    \label{fig:b-bedmap-fit-lrm-etc}
    \end{subfigure}
    \caption{\textit{Analysis of the Antarctic ice sheet dataset.} Panel (a) provides a plot of the data, and panel (b) shows the posterior results obtained. The estimate of $\beta$ is of 0.005 with estimated coverage $98\%$.}
    \label{fig:ice_comparison}
\end{figure}

 We also investigate the sensitivity of LRM-Bayes to the smoothing parameter $\alpha$ for the Potts model.
 While its effect has been mild in the main experiments so far, \Cref{appendix:alpha-vs-small-grids} shows that it becomes more visible on smaller grids, where fewer local configurations are observed.
 This illustrates that the impact of $\alpha$ is tied to sparsity: with fewer observations, Laplace smoothing has a stronger influence on $\hat{q}$ and hence on the posterior.

\section{Theoretical Assessment}\label{sec:theory}

We provide a concise theoretical analysis regarding posterior consistency and Bernstein–von Mises type behavior for LRM-Bayes.
Our theory relies on the arguments of \citet{miller2021asymptotic}, and thus applies equally to the well-specified and misspecified regimes.
We confine our presentation to the case of exponential families, since they are the main model class considered in this paper, but 
\Cref{appendix:theoretical-assessment} extends the results to more general settings for completeness. 
This appendix also includes proofs and additional technical lemmas that may be of interest for the frequentist study of LRM-based point estimators.
The proofs for the presented theoretical assessment are in \Cref{appendix:proofs-of-theory}.

The main results presented hereafter demonstrate that LRM-Bayes enjoys the same asymptotic guarantees as other generalized posteriors \citep{matsubara2022robust,cherief2020mmd,matsubara2024generalized}.
Here, the technical complication relative to prior work is the LRM estimator's reliance on the PMF estimate $\hat{q}_{\alpha}$.
For our results to hold, we will require that the estimated log ratios used within LRM-Bayes can be estimated arbitrarily well as $n\to\infty$ on every truncated domain where 
$\dgp$ is bounded away from zero.
The next lemma establishes this property.
\begin{lemma}
    For $\alpha >0$, for every $\varepsilon>0$, the PMF estimator $\hat{q}_\alpha$ in \eqref{eq:PMF-estimator} satisfies:
    {\setlength{\abovedisplayskip}{4pt}
 \setlength{\belowdisplayskip}{4pt}
    \begin{align}
             \sup_{\x \in \X: \dgp(\x)\geq\epsilon} \max_{\x' \in M(\x)} \left | \log \left (\frac{\hat{q}_{\alpha}(\x')}{\hat{q}_{\alpha}(\x)} \right) - \log \left( \frac{\dgp(\x')}{\dgp(\x)} \right) \right| \overset{a.s.}{\longrightarrow} 0. 
             \label{eq:sup-max-q}
    \end{align}
    }
    \label{lemma:laplace-estimator}
\end{lemma}
\vspace{-2em}
This allows us to overcome the challenges posed by the PMF estimator $\hat{q}_{\alpha}$. To guarantee the remaining pointwise convergence properties of \( \Lhat  \), we impose regularity conditions.
\vspace{-2em}
\begin{assumption}
One of the following conditions holds: 
\vspace{-1em}
\begin{enumerate}
    \item [(i)] $|\X| < \infty$;
    \item [(ii)] $\mathcal{X}$ is countably infinite, and $\dgp$ is a subexponential distribution. Moreover, there exists constants $K_{\dgp}>0$, $\gamma>1$ and a sequence $(K_{\hat{q}_\alpha,n})_{n\geq1}$ of non negative constants with $\limsup_{n\to\infty} K_{\hat{q}_\alpha,n}<\infty$, such that for all $n\in\mathbb N$ , $\x\in\X$, $j\in\mathcal{J}$, $|\mathcal{R}_j[\dgp](\x)| \leq K_{\dgp}(1+\|\x\|^\gamma)$ and $|\mathcal{R}_j[\hat{q}_{\alpha}](\x)| \leq K_{\hat{q}_\alpha,n}(1+\|\x\|^\gamma)$  
\end{enumerate}
\vspace{-1em}
\label{assumption:laplace-estimator}
\end{assumption}

In \Cref{assumption:laplace-estimator},  (i) yields pointwise convergence straightforwardly, and (ii) allows unbounded spaces so long as $\dgp$ does not put too much mass on a large number of states, and that log-ratios do not grow too quickly, the latter being ensured by defining $M$ so that each $\x'\in M(\x)$ differs from $\x$ by a bounded change. We note that this assumption holds for all experiments in the paper. The discrete Markov random fields all satisfy (i), while for the count data models, we have provided experiments with all three model classes where the part of (ii) dependent on $\hat{q}_\alpha$ holds, and the part of (ii) dependent on $q_0$ can be verified for all synthetic experiments.

To obtain a Bernstein-von Mises theorem and posterior concentration for non-exponential family models, stronger modes of convergence are required---we refer the reader to \Cref{appendix:theoretical-assessment}.
However, for the special case of exponential family likelihoods, these results follow immediately from convexity in the natural parameters. 
\begin{theorem} [Consistency \& Bernstein–von Mises]
Suppose $\model$ is an exponential family as in \Cref{eq:exp-family-model} with natural parameter $\eta(\btheta)=\btheta$. Then $\Lhat(\btheta)$ and $\Lexp(\btheta)$ are convex in $\btheta$, and therefore $\Lexp$ has a minimizer $\btheta_\star \in \Theta$.
In addition, suppose that \Cref{assumption:laplace-estimator} holds and that the prior $\pi$ has a density continuous at $\btheta_\star$ with $\pi(\btheta_\star)>0$.
Let $B_\epsilon(\btheta_\star)=\{\btheta\in\Theta\,:\, \|\btheta-\btheta_{\star}\|_{2}\leq\epsilon\}$. Then, for any $\epsilon>0$,
{\setlength{\abovedisplayskip}{4pt}
 \setlength{\belowdisplayskip}{4pt}
\begin{equation*}
    \int_{B_\epsilon(\btheta_\star)}\hat{\pi}_M^\beta(\btheta)d\btheta \;\xrightarrow{\text{a.s}}\; 1.
\end{equation*}
}
Let $\tilde{\pi}_M$ be the p.d.f. of the random variable $\tilde{\btheta}_n := \sqrt{n}(\btheta - \hat\btheta_n)$ for $\btheta \sim \hat{\pi}_M^\beta$, viewed as a p.d.f. on $\mathbb{R}^{p}$.
Let $\mathbf H_{\star} := \beta\nabla^2_{\btheta}\,\Lexp (\btheta)\big|_{\btheta=\btheta_\star}$. If $\mathbf H_\star$ is non-singular,
\begin{equation*}
\int_{\mathbb R^p}
\left|
\tilde\pi_{M}(\tilde{\btheta}_n) - \frac{1}{\det(2\pi \mathbf H_\star^{-1})^{1/2}}
\exp\!\left(-\tfrac{1}{2}\btheta^\top \mathbf H_\star\btheta\right)
\right|\,d\btheta \;\xrightarrow{\text{a.s}}\; 0.
\end{equation*}
    \label{theorem:bvm}
\end{theorem}
\vspace{-2em}
The above result shows that under mild regularity conditions on the prior, the LRM–Bayes posterior is asymptotically Gaussian.
Such results are now considered standard in the generalized Bayes literature, where they have several important uses.
Most immediately, the Gaussian limit provides a convenient approximation to the posterior.
More interestingly, when the underlying model is well specified, discrepancies between the asymptotic covariance matrix $\mathbf H_\star$ and the inverse Fisher information have been used to assess posterior calibration or loss of efficiency relative to standard Bayesian procedures.

We also note that \Cref{theorem:bvm} treats $\beta$ as fixed, whereas in practice we estimate it from data following the procedure of \citet{syring2019calibrating}. This discrepancy can be reconciled by estimating $\beta$ on an initial subsample of fixed size and then holding it fixed at $\hat{\beta}$: conditionally on this calibration subsample, $\hat{\beta}$ is constant and \Cref{theorem:bvm} applies directly. Extending this asymptotic result to a data-driven $\hat{\beta}_n$ would be a promising direction for future work, as was done in \citet{ray2026statistical} for the power posterior.
\vspace{-2em}
\section{Conclusion}

This paper introduces a new divergence for discrete-valued data we call the log-ratio matching (LRM) divergence.
Combined with generalized Bayesian inference (GBI), this produces LRM-Bayes: a conjugate generalized posterior for discrete intractable exponential family models, which is consistent and satisfies a Bernstein-von-Mises theorem.
The resulting procedure is substantially more computationally efficient than standard Bayes and the related GBI-based approach in \citet{matsubara2024generalized}, with computational gains between $10\times$ and $6000\times$ while producing posterior distributions that closely match those obtained by other baseline methods.

The scope of LRM-Bayes is already broad, spanning graphical count data models and Markov random fields, and several further extensions follow naturally from the present formulation.
One direction concerns settings in which part of the data is continuous, such as regression tasks with discrete responses and continuous covariates \citep{hosmer2013applied}.
The Laplace smoothing estimator used in this paper would not be appropriate in this context without further adaptation, but other strategies could be employed instead in line with ideas introduced by \citet{hooker2014bayesian}.
More importantly, conditional dependence and regression settings raise non-trivial questions about how the matching set should be defined, and specifically whether it should be constructed jointly across the continuous and discrete components or restricted to the discrete space. 
More broadly, one could also ask whether $M$ can be selected optimally by studying which choices lead to better statistical properties of the resulting LRM divergence, such as identifiability, efficiency, and robustness.

Relatedly, the LRM-Bayes construction depends on the data-generating process only through the local log-ratios that enter the loss.
Therefore, future work could develop estimators targeted directly at these local ratios by adapting ideas from direct density-ratio estimation \citep{sugiyama2012densityinml}, including KL-divergence-based methods, least-squares methods, and more general Bregman-divergence density-ratio matching approaches.
This could avoid estimating the full probability mass function and may be especially useful in large or sparse discrete supports.

LRM-Bayes also extends naturally to latent-variable models, such as mixture models and hidden Markov models. In these settings, the partial conjugacy developed in this paper may bring substantial computational benefits, since the LRM-Bayes update could be applied to different parts of the latent-variable formulation, such as the latent model, the observation model conditional on the latent variables, or to marginalizations of the joint model.
This raises the question of how to choose the divergence, matching set, and local-ratio estimator to exploit the latent-variable structure and preserve reliable posterior inference.

A second direction arises from the fact that, unlike much of the literature on generalized Bayesian inference, this paper has not focused on robustness: the proposed framework is motivated instead by computational considerations---an initiative which has recently gained traction in the generalized Bayes literature \citep[e.g., see][]{matsubara2022robust, matsubara2024generalized, altamirano2023robust, laplante2025robust}.
The computational gains achieved by LRM-Bayes suggest a broader question of where generalized posteriors designed for tractability might offer similar advantages in other data spaces.
Settings such as simulation-based inference \citep{price2018bayesian, beaumont2019approximate, cranmer2020frontier}, models on manifolds \citep{mardia2009directional}, or other likelihood-free problems with complex structure may benefit from constructing posterior distributions in this spirit, provided an appropriate divergence can be specified.
Pursuing such extensions would provide a promising avenue for future work.

Finally, LRM-Bayes raises the question of how generalized posteriors in discrete-data settings should be compared.
Unlike in standard Bayes, their normalizing constants are not marginal likelihoods, so standard Bayes factors do not directly apply.
Building on recent work by \citet{jewson2022general} in continuous settings, future work could develop analogous criteria for comparing posteriors arising from different losses, hyperparameters, or modeling choices in discrete unnormalized models.

\vspace{-1em}
 \subsection*{Acknowledgments}
 William Laplante was partially supported through the project `Bayesian robustness in filtering algorithms' funded by The Alan Turing Institute. Matias Altamirano was supported by a Bloomberg Data Science PhD
 fellowship. The authors would like to thank the Isaac Newton Institute for Mathematical Sciences, Cambridge, for support and hospitality during the programme `Representing, calibrating \& leveraging prediction uncertainty from statistics to machine learning' where work on this paper was undertaken.
 Jeremias Knoblauch and Fran\c{c}ois-Xavier Briol were supported by the EPSRC grant [EP/Y011805/1].

\subsection*{Supplementary Material \& Data Availability}
The supplementary material contains proofs of all technical results presented in the article, as well as additional numerical details and results. The location of the data that support the findings of this study is specified in the supplementary material or is available upon request from the corresponding author. The data are not always publicly available due to size constraints.

 {\spacingset{1.0}

 {\footnotesize
 \bibliographystyle{asa}
 \bibliography{bibliography}
 }
 }

\newpage
\appendix

\addcontentsline{toc}{section}{Appendix} 
\part{\LARGE Supplementary Material} 

In \Cref{appendix:implementation-details}, we provide more details on the implementation of the methodology presented in \Cref{sec:methods}.
In \Cref{appendix:exp-add-details}, we give additional details on the experiments presented in \Cref{sec:experiments}.
In \Cref{appendix:robust-loss}, we study a robust extension of the LRM loss.
In \Cref{appendix:theoretical-assessment}, we investigate theoretical results complementing those in \Cref{sec:theory} for the general case of models outside the exponential family. 
In \Cref{appendix:proofs-of-theory}, we provide proofs for the theoretical results presented in \Cref{sec:theory}.

\section{Additional Implementation Details}\label{appendix:implementation-details}

Details in \Cref{appendix:select-matching-set-for-connected-graph} elaborate on the concepts introduced in \Cref{sec:log-ratio-matching-div}, and the material in \Cref{appendix:eval-loss-empirical-PMF}, \Cref{appendix:select-base-PMF}, and \Cref{appendix:posterior-calibration} provides further technical details for \Cref{sec:calibration-selection}.

\subsection{Selecting $M$ to Satisfy Graph Connectedness} \label{appendix:select-matching-set-for-connected-graph}
When $|\X| < \infty$, selecting $M$ so that \Cref{assumption:graph-connect} holds is typically straightforward, as we demonstrate in the experiments of \Cref{sec:experiments}.
For infinite support, so long as the matching set is reasonably specific, we may rely on the underlying undirected graph.
For example, when $\mathcal{X} = \mathbb{N}$ and $M(x):= \{x+1\}$, the graph may be traversed from $x$ to $x+1$ or the opposite to achieve connectedness.

\subsection{Evaluating the Loss Under the Empirical PMF}\label{appendix:eval-loss-empirical-PMF}

We first note that our theoretical results require $\alpha>0$, ensuring that $\hat q_\alpha$ assigns positive mass to every configuration involved in the loss.
Nevertheless, for sensitivity analysis and implementation testing, we also consider the limiting choice $\alpha=0$, for which $\hat q_\alpha$ reduces to the empirical PMF $\empirical$.
In this case, $\empirical$ may assign zero mass to some configurations, making the corresponding log-ratio terms in $\Lhat$ undefined.
For this exploratory setting, we omit such terms and restrict the inner sum in $\Lhat$ to pairs $(\x_i,\x')$ for which both $\empirical(\x_i)$ and $\empirical(\x')$ are positive.
Thus, only comparisons between configurations belonging to the empirical support are retained.
This convention allows us to evaluate the same loss structure at $\alpha=0$, although this case lies outside the assumptions of our theoretical results.

\subsection{Selecting the Base PMF for Infinite Support}\label{appendix:select-base-PMF}
For infinite but countable $\X$, we use a mixture of a uniform distribution on a subdomain $\tilde{\X} \subset \X$ and a  PMF $r \in \AllAdmPMFsGiven{\dgp}$. 
Concretely, the base PMF is defined as $\basePMF (\x) = (1-\epsilon) \text{Unif}(\tilde{\X})(\x) + \epsilon r(\x)$, where $\epsilon > 0$ is taken to be small. 
We select this base PMF because (i) it satisfies the requirement of  \Cref{assumption:pmf-estimator} and (ii) the uniform component has proven effective for inference and performs well empirically (see \Cref{sec:experiments}).
In practice, the set $\tilde{\X} \subset \X$ is chosen to contain all observed values.
The parameter $\alpha$ regulates the influence of the base PMF $\basePMF$.
In moderate-sample regimes where the empirical distribution is sufficiently representative of $\dgp$, positivity violations are negligible for our purposes, and the empirical PMF ($\alpha=0$) is adequate and preferred for simplicity.
When data are sparse or unrepresentative, we set $\alpha>0$ to ensure stability and positivity.

\subsection{Estimating $\beta$}\label{appendix:posterior-calibration}
The procedure we use to estimate $\beta$, taken from \citet{syring2019calibrating}, is outlined as follows. To simplify the notation, we refer to a dataset $\{\x_i\}_{i=1}^n$ as $\mathcal{D}$, and a resampled version of this dataset, that is, a bootstrap sample, as $\mathcal{D}_b$, where $b$ indexes the resampled sample.

\begin{enumerate}
    \item With dataset $\{\x_i \}_{i=1}^n$, estimate $\hat{\btheta}_n:= \arg \min_{\btheta \in \Theta} \Lhat(\btheta) = \bm{\Lambda}_n^{-1} \bm{\nu}_n$.
    \item For $\beta > 0$ and fixed coverage $1-\delta$, do the following:
    \begin{enumerate}
        \item We generate bootstrap datasets $\mathcal{D}_b$ for $b=1,\dots,B$, by sampling with replacement.
        We then compute the posteriors $\{\hat{\pi}_M^\beta(\btheta;  \mathcal{D}_b)\}_{b=1}^B$.
        Recall that each posterior $\hat{\pi}_M^\beta(\btheta;  \mathcal{D}_b)$ is conjugate when the model is an exponential family; hence, we may only need to compute the means and covariances using \Cref{prop:exp_fam}.
        \item Estimate the $1-\delta$ coverage region $\mathcal{C}\left(\hat{\pi}_M^\beta(\btheta; \mathcal{D}_b),\; 1-\delta \right)$ for $b=1,\dots, B$. When the posteriors $\hat{\pi}_M^\beta(\btheta;  \mathcal{D}_b)$ for $b=1, \dots, B$ are Gaussian (from the conjugacy of \Cref{prop:exp_fam}), the procedure simplifies. In particular,
        \[
        \mathcal{C}\left(\hat{\pi}_M^\beta(\btheta; \mathcal{D}_b), 1-\delta \right) = \{ \btheta : (\btheta - \bm{\mu}_n(\beta))^\top \mathbf{\Sigma}^{-1}_n(\beta) (\btheta - \bm{\mu}_n(\beta)) \leq \chi^2_{p, 1-\delta} \},
        \]
        where $\chi^2_{p, 1-\delta}$ is the $1-\delta$ quantile of a chi-squared distribution with $p$ degrees of freedom.
        
        \item Estimate the frequency $1-\hat{\delta}$ at which $\hat{\btheta}_n \in \mathcal{C}(\hat{\pi}_M^\beta(\btheta;  \mathcal{D}_b), 1- \delta)$ using the bootstrapped datasets.
    \end{enumerate}
    \item Repeat 2. for a different $\beta>0$ until coverage $1-\hat{\delta} \approx 1-\delta$.
\end{enumerate}
The outlined procedure above defines a stochastic optimization objective over $\beta$. That is, we may use a standard black-box optimizer to optimize over $\beta$.
In some cases, specifying a grid of $\beta$ values is simpler and sufficient, as one of the $\beta$'s produces a posterior that attains the desired coverage.

If the posterior is Gaussian or truncated Gaussian, the procedure remains computationally tractable.
Otherwise, only Steps~2(a)-2(b) change: Step~2(a) requires sampling from the \(b=1,\dots,B\) posteriors, rather than storing their already computed means and covariances, and Step~2(b) uses sample-based quantile estimation rather than exact Gaussian calculations.
All other steps remain unchanged.

Thus, non-Gaussian posteriors make the procedure more computationally demanding, mainly through the posterior sampling required in Step~2(a).
The overhead depends on sampler efficiency, and hence on the dimension of \(\btheta\) and \(\x\), as well as the geometry of the posterior.
This highlights a key computational benefit of our method: the conjugacy and Gaussian posteriors obtained in our setting substantially accelerate the standard generalized-Bayes procedures for selecting \(\beta\).

\section{Additional Details on the Numerical Experiments} \label{appendix:exp-add-details}
This section presents additional information regarding the experiments of \Cref{sec:experiments}.
In \Cref{appendix:intract-count-data}, we present the experiments on the univariate CMP model from \Cref{sec:CMP_univariate}, the graphical model from \Cref{sec:graphical-model-breast-cancer}, and the autoregressive count data model from \Cref{sec:autoregressive-model-crime-data}.
In \Cref{appendix:intract-on-lattices}, we present the experiments on the Ising model from \Cref{exp:Ising-model} and the Potts model from \Cref{exp:Potts-model}.
Further, we will denote the $d$-dimensional identity matrix by  $I_d$.

\subsection{Intractable Models of Count Data from \Cref{exp:intractable-models-of-count-data}} \label{appendix:intract-count-data}

\subsubsection{Additional Details for the Univariate Conway-Maxwell-Poisson Model in \Cref{sec:CMP_univariate}} \label{appendix:add-details-univariate-CMP}

In the first part of the \textbf{Univariate Conway-Maxwell-Poisson Model} experiment, we simulate $n=2000$ data points from a Conway-Maxwell-Poisson model with $(\theta_1, \theta_2)$ set to $(4.0, 0.75)$, which is over-dispersed, and $(4.0 ,1.25)$, which is under-dispersed, as was done in \citet{matsubara2024generalized}.
The sampling is done by rejection sampling, where the proposal distribution is a Poisson distribution with the same $\theta_1$, and the CMP model is evaluated by computing a truncated normalization constant $Z(\btheta) \approx \sum_{y=0}^{99} \model(y)$. 
This is feasible in this setting since the data is univariate.
The prior used is a $\mathcal{N}((3.,3.), I_2)$ on $\bm{\eta} = (\log \theta_1, \theta_2)$.
The dataset produced from the CMP model, and the fit for the Bayes and DFD methods, are both obtained using code from \citet{matsubara2024generalized}, available at \url{https://github.com/takuomatsubara/Discrete-Fisher-Bayes/tree/main} (see the CMP folder, ``code\_posterior.py'' file, and the file producing the Bayes and DFD results is ``benchmark-methods.py'' in our repository). 
For LRM-Bayes, we use $\alpha=1$
and $M(x) = \{x+1\}$, and estimate $\beta$ with $B=50$ bootstrap samples, following the procedure from \Cref{appendix:posterior-calibration}. 
The ellipses in \Cref{fig:CMP-1d} correspond to $95\%$ credible regions from the corresponding Gaussian posterior.
The computational times are obtained on a CMP model with $(4.0, 1.25)$, and include the estimation of $\beta$ (for the generalized Bayes methods). 
The DFD and Bayes methods require MCMC samples. 
We retain 1000 samples with 5000 burn-ins. 
The former specifications match exactly those of \citet{matsubara2024generalized}. 
For DFD, we estimate $\beta$ using the code from \citet{matsubara2024generalized}. 
For LRM, the $\alpha, M, \beta$ are obtained in the same way as for the fit.

In the second part of the \textbf{Univariate Conway-Maxwell-Poisson Model} experiment, we investigate the impact of $\alpha$ and $M$ on the posterior of a 1D CMP model.
We use the over-dispersed data from \Cref{fig:CMP-1d} found in \Cref{sec:CMP_univariate}, with the same prior.
We also fix $\beta=1.0$, except for \Cref{fig:neighbours-vs-posterior-calibrated} where we estimate $\beta$. Finally, since the posterior is on $(\log \theta_1, \theta_2)$, we transform the posterior to $(\theta_1, \theta_2)$. This random variable transformation is in closed form. 
In \Cref{fig:alpha-vs-posterior}, we fix $M(\x):=\{x+1\}$, and vary $\alpha \in \{0.0,0.1, 0.5, 1.0 \}$. For $\alpha=0.0$, $\hat{q}_\alpha$ becomes the empirical $\empirical$, and is treated as outlined in \Cref{appendix:eval-loss-empirical-PMF}. 
In \Cref{fig:neighbours-vs-posterior}, we fix $\alpha=0.0$, and vary $M(\x)$ to be $ \{x+1\}, \{x-1\}, \{x-1, x+1\}, \{x-2, x-1, x+1, x+2\}$. 
Finally, in \Cref{fig:neighbours-vs-posterior-calibrated}, we use identical specifications as for \Cref{fig:neighbours-vs-posterior}; however, we focus on $M_1(\x):=\{x+1\}, M_2(\x):= \{x-2, x-1, x+1, x+2\}$, and we estimate $\beta$ for the two posteriors with $B=50$ bootstrap samples.

We now conduct an experiment displaying the closed-form, truncated Gaussian LRM-Bayes posterior for the univariate CMP model when the prior satisfies the condition of \Cref{cor:closed-form-trunc-Gauss}.
We compare three priors: a uniform prior on $[2,6]\times[0.5,2]$, an independent exponential prior $\mathrm{Exp}(2)$ for $\theta_1, \theta_2$, and a continuous-Bernoulli prior with $\lambda=(0.8,0.8)$ truncated to $[0,10]\times[0,10]$. 
The data consist of $n=100$ observations generated from the CMP model with $\btheta=(4,1.25)$, and we select $\alpha=1$.
The contours show kernel density estimates of the $95\%$ highest-density credible regions. 
The results are in Figure~\ref{fig:CMP-different-priors}.
The differences between the three posteriors reflect the expected influence of the prior in this small-sample setting. 
The uniform prior allocates the most mass around the plausible parameter range, and hence yields the most constrained posterior; the exponential prior favors smaller parameter values, and therefore pulls the posterior slightly downward; and the continuous-Bernoulli choice favors larger parameter values over a wider truncated domain.

\begin{figure}
    \centering
    \includegraphics[width=0.6\linewidth]{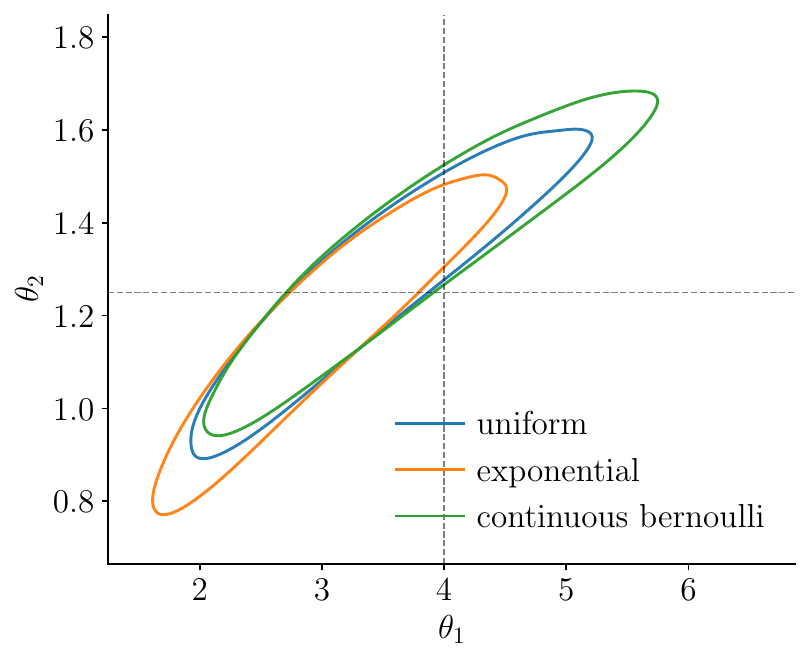}
    \caption{\textit{Comparing LRM-Bayes posteriors induced by different priors for the univariate CMP model.}}
    \label{fig:CMP-different-priors}
\end{figure}

\subsubsection{Comparing PMF estimators for LRM-Bayes under the univariate CMP model}

We assess the role of $\hat q$ via inference with the univariate CMP model, comparing the LRM-Bayes posterior obtained from three choices of \(\hat q\).
The first is Laplace smoothing, the default choice in this paper.
The second is the Good-Turing estimator, following the classical approach studied by \citet{orlitsky2003always}.
The third is the nonparametric maximum likelihood estimator of \citet{han2025besting}, denoted NPMLE-EB; NPMLE-EB estimates a latent distribution over probabilities and uses this fitted distribution to regularize the observed frequencies.

Since the Good-Turing and NPMLE-EB estimators are designed for finite domains, we apply them after truncating to \(\{0,\ldots,20\}\).
In the CMP setting, this corresponds to working with a truncated CMP model.
The results are in Figure~\ref{fig:cmp-vary-q-estimator}.
This experiment shows that LRM-Bayes can incorporate alternative estimators of \(\hat q\) and retain good performance so long as these estimators are appropriate for the domain and modeling task.
Further, studying the impact of these estimators on the asymptotic arguments would be an interesting direction for future work.

\begin{figure}
    \centering
    \includegraphics[width=1.0\linewidth]{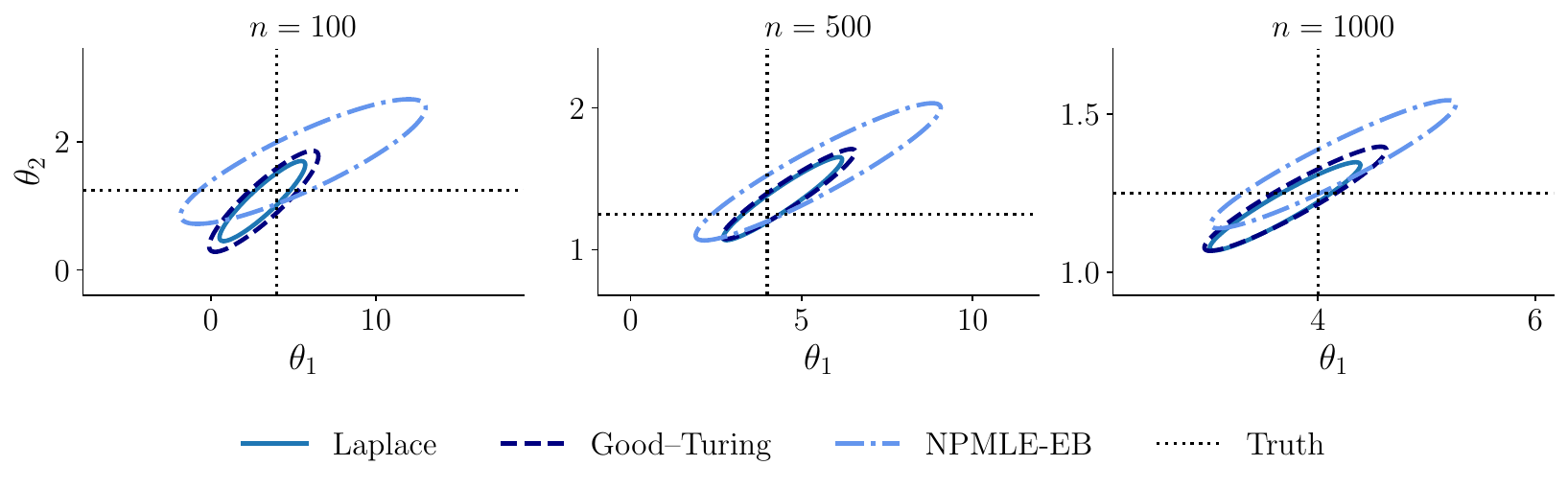}
    \caption{\textit{Varying PMF estimator under the univariate CMP model.} We simulate data from a CMP model with true parameter \(\btheta=(4.0,1.25)\), and compare the resulting LRM-Bayes posterior using three estimators of \(\hat q\): Laplace smoothing with \(\alpha=1\), the Good-Turing estimator \citep{orlitsky2003always}, and the NPMLE empirical-Bayes estimator of \citet{han2025besting}.
        Since Good-Turing and NPMLE are finite-domain estimators, we apply them on the truncated support \(\{0,\ldots,20\}\), corresponding to inference under a truncated CMP model.
        Each panel shows the Gaussian \(95\%\) posterior credible contour for a sample size \(n\in\{100,500,1000\}\), using the independent Gaussian prior \(\theta_1,\theta_2\sim \mathcal N(3,10)\).
    }
    \label{fig:cmp-vary-q-estimator}
\end{figure}

\subsubsection{Additional Experiment on the Univariate CMP Model for Increasing $n$}

We next study the behavior of LRM-Bayes as the sample size increases in a well-specified univariate CMP model.
Independent datasets are generated from a CMP distribution with true parameter $\theta=(4,1.25)$, and LRM-Bayes is fit for sample sizes $n=100,\ldots,204800$. 
We use a $\mathcal{N}((3., 3.), I_2)$ prior on $(\log \theta_1, \theta_2)$.
For each $n$, the learning rate $\beta$ is selected using the coverage-calibration procedure with $B=100$ bootstrap replicates.
The results (\Cref{fig:uni-CMP-vs-sample-size}) demonstrate that the LRM-Bayes posterior uncertainty contains the truth across sample sizes and narrows around the true parameter as $n$ increases.

\begin{figure}
    \centering
    \includegraphics[width=0.9\linewidth]{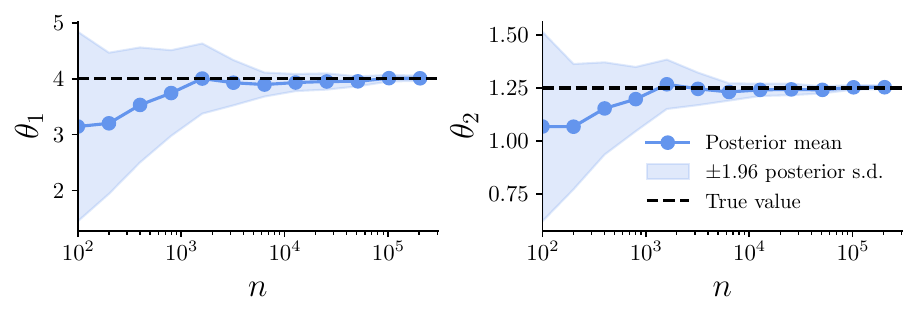}
    \caption{\textit{LRM-Bayes posterior for increasing sample sizes in the univariate CMP model.} Solid lines show posterior means for $\theta_1$ and $\theta_2$, shaded regions give pointwise $\pm 1.96$ posterior standard deviations, and dashed lines denote the true parameter values.}
    \label{fig:uni-CMP-vs-sample-size}
\end{figure}

\subsubsection{Additional Details for the Graphical Model for Breast Cancer Data in \Cref{sec:graphical-model-breast-cancer}}\label{appendix:add-details-graph-model-breast-cancer}

In the \textbf{Graphical Model for Breast Cancer Data} experiment, the data is from \url{https://github.com/takuomatsubara/Discrete-Fisher-Bayes/tree/main/Graphical/Data/} and can be found in the ''brca\_10.npy'' file. 
The prior for all methods is set to be a $\mathcal{N}(0, I_{p})$.
For the LRM posterior, we set $\alpha = 0.1$ to introduce a small amount of smoothing given the limited sample size, which helps stabilize the estimates.
In practice, $\alpha$ can be selected according to the amount of data available, with more smoothing used when the sample size is small.
For completeness, we provide the posterior predictive results of LRM for $\alpha=0.1, 0.5, 1.0$ in \Cref{fig:BRCA-alpha-sensitivity}, showing that greater smoothing assigns more mass to unseen outcomes. 
Note that the LRM posterior is a truncated multivariate normal, since $\theta_{i,j}\geq 0$ for all $i<j$ and $\theta_{0,j}>0$ for all $j=1,\ldots,d$ are required for the CMP graphical model to be normalizable.
For our fit, the matching set $M(\x)$ is defined as all axis-aligned neighbors of $\x \in \X$ obtained by adding offsets $j \in \{-2,-1,1,2\}$ to a single coordinate.
The $\beta$ hyperparameter is fitted with $B=50$ bootstrap samples.
The posterior predictive samples for DFD are provided in \url{https://github.com/takuomatsubara/Discrete-Fisher-Bayes/tree/main/Graphical/}.
For LRM, posterior predictive samples are obtained by drawing parameter values from the truncated multivariate normal posterior and, conditional on each draw, simulating from the CMP model with a Metropolis–Hastings sampler.
We use a single-site random-walk proposal that selects one coordinate at random and increases or decreases it by one.
The proposal is symmetric, aligns with the discrete structure of the CMP, and requires no further tuning.
After burn-in, the chain is thinned, and one synthetic draw per observation is retained to form the predictive distribution.

\begin{figure}
    \centering
    \includegraphics[width=1.0\linewidth]{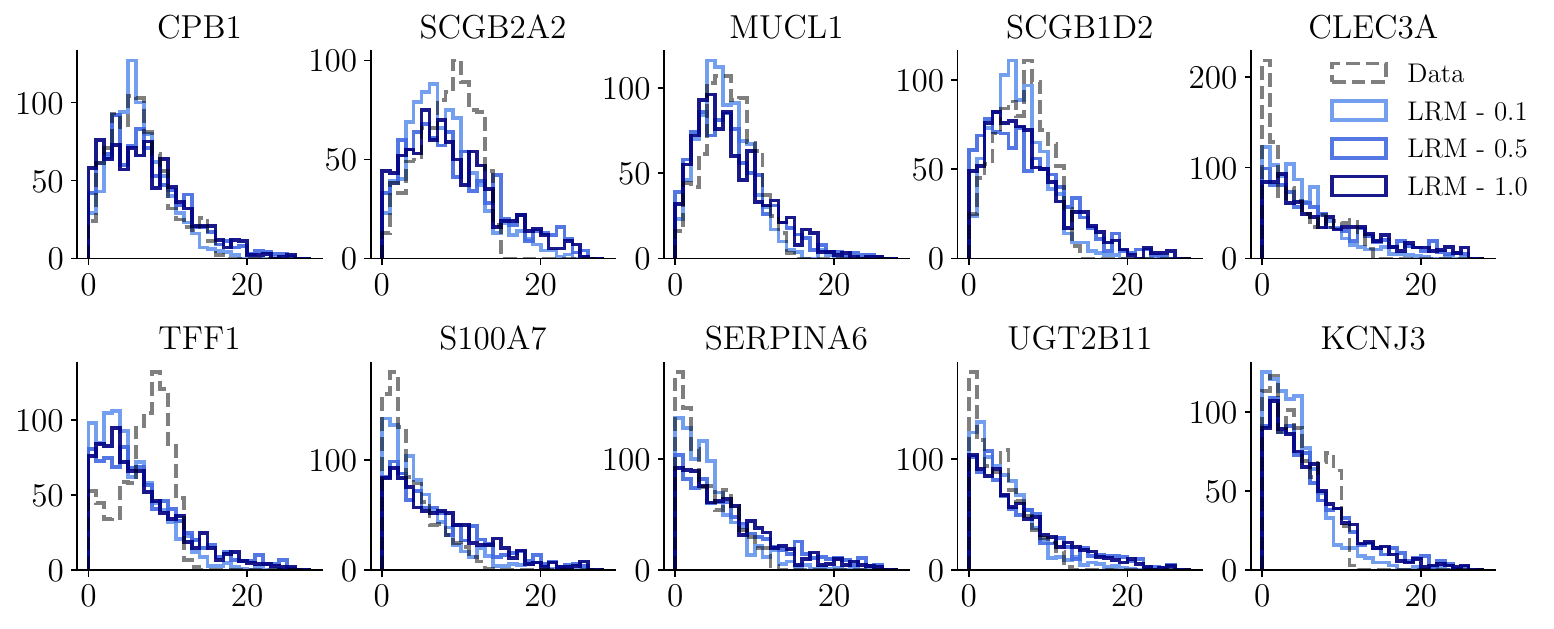}
    \caption{\textit{Sensitivity of LRM-Bayes' posterior predictive to $\alpha$ for BRCA dataset}.}
    \label{fig:BRCA-alpha-sensitivity}
\end{figure}

We now discuss the details for the horseshoe and spike-and-slab priors, which are for panel (b) of \Cref{fig:CMP-graphical-model} in the main text and its extension \Cref{fig:CMP-graphical-model-appendix} in the Appendix.
First, the horseshoe prior from \citet{makalic2015simple} is (with $\sigma=1$ fixed):
\begin{align*}
    &\theta_j \mid \lambda_j^2,\tau^2 
    \sim \mathcal N(0,\tau^2\lambda_j^2), \quad
    \lambda_j^2 \mid \psi_j 
    \sim \operatorname{IG}\left(\frac12,\frac{1}{\psi_j}\right),
    \quad
    \psi_j \sim \operatorname{IG}\left(\frac12,1\right), \\
    &\tau^2 \mid \xi 
    \sim \operatorname{IG}\left(\frac12,\frac{1}{\xi}\right),
    \quad
    \xi \sim \operatorname{IG}\left(\frac12,1\right).
\end{align*}
where $\operatorname{IG}(\alpha, \beta)$ is the inverse gamma distribution with parameters $\alpha, \beta > 0$.

With \( D_{\bm{\lambda},\tau} = \tau^2\operatorname{diag}(\lambda_1^2,\dots,\lambda_p^2), \) and \( \hat{\mathcal L}_n(\btheta) = \btheta^\top \Lambda_n \btheta - 2\btheta^\top \nu_n + C, \) the $\btheta$-block is
\begin{equation*}
    \btheta \mid \bm{\lambda},\tau,\{\x_i\}_{i=1}^n
    \sim
    \mathcal N(\mu_n,\Sigma_n),
\end{equation*}
where \(\Sigma_n := \left(2\beta n\Lambda_n + D_{\bm{\lambda},\tau}^{-1}\right)^{-1} \) and \(\mu_n := \Sigma_n(2\beta n\nu_n). \) The remaining Gibbs updates are

\begin{align*}
    &\lambda_j^2 \mid \theta_j,\tau^2,\psi_j
    \sim
    \operatorname{IG}\left(
    1,
    \frac{1}{\psi_j}+\frac{\theta_j^2}{2\tau^2}
    \right), \quad 
    \psi_j \mid \lambda_j^2
    \sim
    \operatorname{IG}\left(
    1,
    1+\frac{1}{\lambda_j^2}
    \right), \\
    &\tau^2 \mid \btheta,\bm{\lambda},\xi
    \sim
    \operatorname{IG}\left(
    \frac{p+1}{2},
    \frac{1}{\xi}
    +
    \frac12\sum_{j=1}^p \frac{\theta_j^2}{\lambda_j^2}
    \right), \quad
    \xi \mid \tau^2
    \sim
    \operatorname{IG}\left(
    1,
    1+\frac{1}{\tau^2}
    \right).
\end{align*}

Second, we consider the spike-and-slab prior applied only to the set
$\mathcal I$ of pairwise interaction parameters. 
Our specification is based on the variable-selection priors of \citet{kuo1998variable,george1997approaches}, but is adapted to impose the non-negativity required for the model to be normalizable. 

Specifically, we postulate
\begin{align*}
    &\theta_j = \gamma_j b_j,\qquad
    \gamma_j \sim \operatorname{Bernoulli}(\pi),\qquad
    b_j \sim \mathcal{N}_{+}(0,\sigma_{\mathrm{slab}}^2),
    \quad j\in\mathcal I,\\
    &\gamma_j=1,\qquad
    b_j\sim \mathcal N(0,1),
    \quad j\notin\mathcal I,
\end{align*}
where $\mathcal{N}_{+}$ denotes a Gaussian distribution truncated to the non-negative half-line. 
We fix $\pi=0.85$ and $\sigma_{\mathrm{slab}}=0.005$. Hence, $\theta_j=0$ exactly when $\gamma_j=0$.

Let \( D_{\bm{\gamma}}=\operatorname{diag}(\gamma_1,\ldots,\gamma_p), \btheta=D_{\bm{\gamma}}\mathbf b, \) and let \(V=\operatorname{diag}(v_1,\ldots,v_p)\) where \(v_j=\sigma_{\mathrm{slab}}^2\) if \(j\in\mathcal I\) and $v_j=1$ else.
With \( \hat{\mathcal L}_n(\btheta) = \btheta^\top\Lambda_n\btheta -2\btheta^\top\nu_n+C,\) the continuous Gibbs block is \( \mathbf b\mid\bm{\gamma},\{\mathbf x_i\}_{i=1}^n \sim \mathcal{N}_{\mathcal C}(\boldsymbol\mu_{\bm{\gamma}},\Sigma_{\bm{\gamma}}),\)
 and \(\Sigma_{\bm{\gamma}} := \left( 2\beta nD_{\bm{\gamma}}\Lambda_nD_{\bm{\gamma}} +V^{-1} \right)^{-1} \), \( \boldsymbol\mu_{\bm{\gamma}} := \Sigma_{\bm{\gamma}} \left(2\beta nD_{\bm{\gamma}}\nu_n\right) \). 
Conditional on $\bm{\gamma}$, the quadratic LRM loss is therefore conjugate to the Gaussian slab, up to truncation.
Exact zeros occur in individual posterior draws but not in the posterior mean; the latter is generally reduced in magnitude rather than exactly zero.

\begin{figure}
    \centering
    \includegraphics[width=1.0\linewidth]{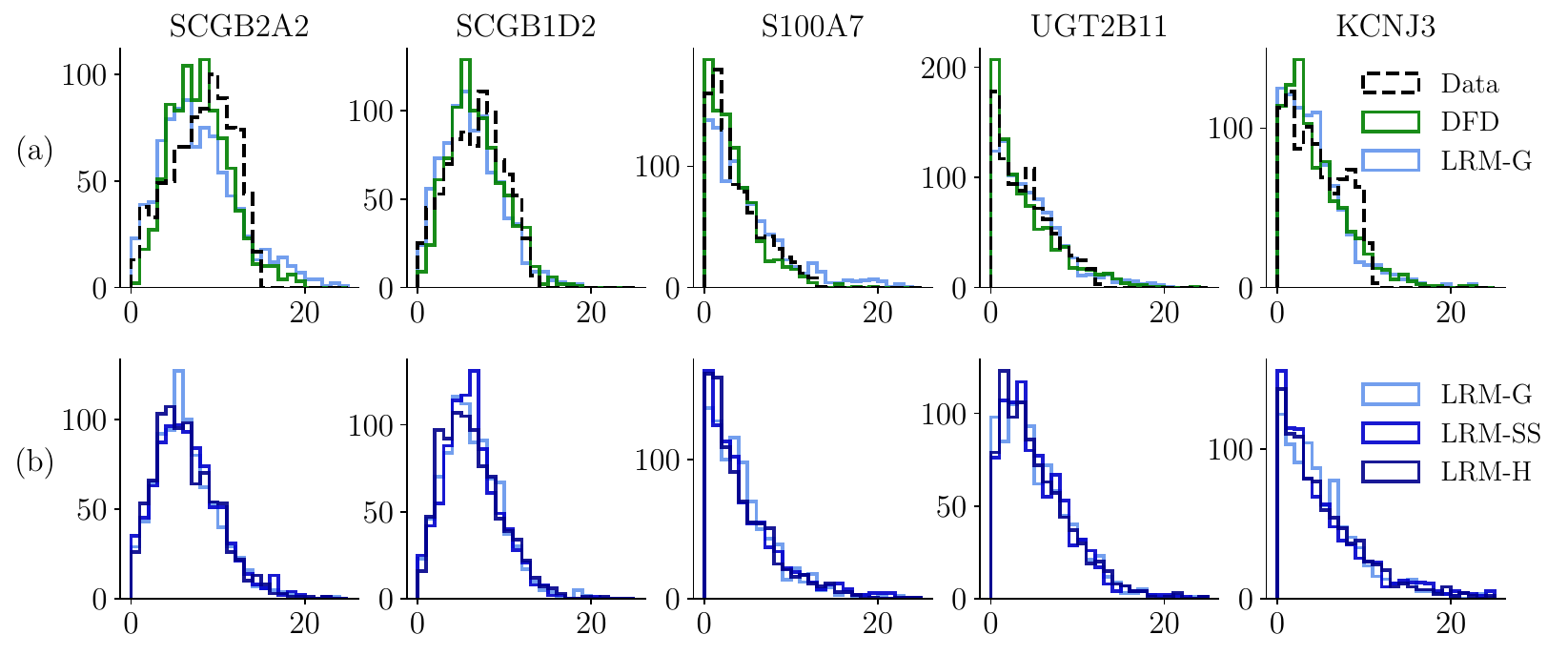}
    \caption{\textit{Posterior predictive distributions for the multivariate Conway-Maxwell-Poisson model.} Each column shows one of the remaining 5 genes not reported in Figure~\ref{fig:CMP-graphical-model}; all notation and conventions are unchanged.}
    \label{fig:CMP-graphical-model-appendix}
\end{figure}

\subsubsection{Additional Details for the Autoregressive Count Time Series Model for Crime Data in \Cref{sec:autoregressive-model-crime-data}}\label{appendix:autoregressive-model-crime-data}

In the  \textbf{Count Time Series Model for Crime Data} experiment, we take a dataset studied in \citet{chen2017bayesian}.
It can be obtained from \url{https://bocsar.nsw.gov.au/statistics-dashboards/open-datasets/criminal-offences-data.html} by selecting the ``Recorded Criminal Incidents by month - by LGA" file. 
We select the "Sexual offenses" category and sum the counts across the subcategories. 
Counts are monthly, from January 1995 to June 2025, with 366 observations in $[0, 32]$. Note that this dataset is actively updated---later downloads may therefore contain additional months.
For both inference methods, we fix the prior on $\btheta \sim \mathcal{N}(\mathbf{1}, 5 \, I_3)$ and $\varphi \sim \mathcal N(0,\,0.01^2)$, and convergence of the sampling chains is assessed using the Gelman-Rubin convergence diagnostic.
LRM uses $\alpha=1.0$ and $M(x) = \{x-1, x+1\}$. 
For LRM posterior inference on $(\btheta, \varphi)$  we use a Metropolis-Hastings within Gibbs scheme. 
At each iteration, we first sample $\btheta \mid \varphi$ from the Gaussian posterior given by the conjugate LRM posterior, reusing a single coverage-based calibration and conditional empirical PMF that are computed once at the start. 
We then update $\varphi \mid \btheta$ with a random-walk Metropolis kernel with proposal $\varphi' \sim \mathcal N(\varphi,\,0.01^2)$.
For each outer iteration, we perform 10 Metropolis updates for $\varphi$, and we run four independent chains of length 5000, discarding the first 3000 iterations as burn-in.
For the approximate MCMC method, we draw joint posterior samples of $(\btheta,\varphi)$ with a random-walk Metropolis-Hastings sampler. 
We use a Gaussian proposal to perturb $(\theta_0,\theta_1,\theta_2,\varphi)$.
We run four independent chains of length 100000 with different seeds, discarding the first 50000 iterations of each chain as burn-in and using the remaining draws for inference.
For the partial model, we fix $\varphi=0.0$, such that the LRM-Bayes posterior on $\btheta$ is conjugate.
Approximate MCMC is performed similarly as for the full model but without the $\varphi$ proposal; furthermore, the number of MCMC steps reduces to 12000 and burn-in to 5000.

\subsubsection{Additional Experiment on Simulated Autoregressive Count Time Series}

We compare the posterior over $\lambda_t$ induced by LRM-Bayes with that obtained using an MCMC approximation to the standard Bayesian posterior. The data are simulated from a CMP-INGARCH$(1,1)$ model with $T=300$, $\phi=0.1$, and $\btheta=(-0.1,0.5,1.5)$. For each posterior draw of $\btheta$, we reconstruct the implied path $(\lambda_t)_{t=1}^T$ and compare the resulting credible bands against the observed counts and the true simulated conditional intensity. As shown in \Cref{fig:lrm-vs-mcmc-lam-ingarch}, both methods yield very similar uncertainty estimates and recover the latent trajectory.

\begin{figure}
    \centering
    \includegraphics[width=0.9\linewidth]{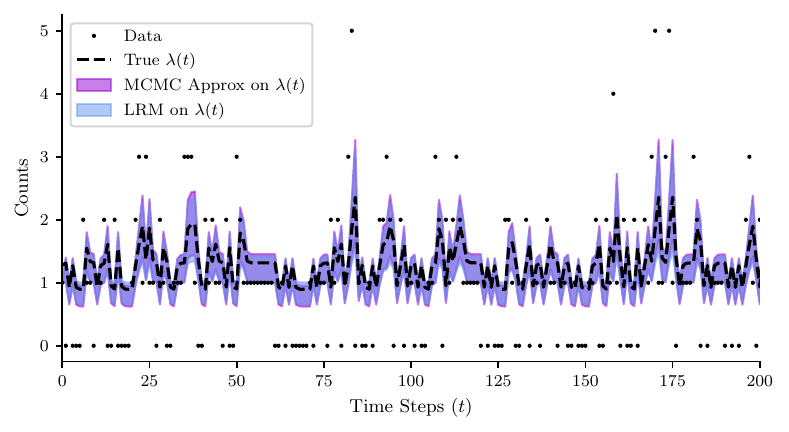}
    \caption{\textit{Posterior over $\lambda_t$ for Approximate Bayes and LRM-Bayes.} Points denote simulated counts, the dashed line denotes the true latent trajectory $\lambda_t$, and shaded regions show pointwise posterior credible intervals induced by posterior draws through the recursive equation for $\lambda_t$. We plot the first 200 points to help visually distinguish between the methods. }
    \label{fig:lrm-vs-mcmc-lam-ingarch}
\end{figure}

\subsection{Intractable Models on Lattices from \Cref{exp:intractable-models-on-lattices}}\label{appendix:intract-on-lattices}

\subsubsection{LRM for Markov Random Fields for \Cref{exp:intractable-models-on-lattices}} \label{appendix:LRM-for-MRFs}

We first show how the stationarity and Markovian assumptions can be leveraged for LRM in Markov random fields.
Recall that $\mathcal{S}$ is the set of states, and for any $\x \in \X$, $m=|M(\x)|=d( | \mathcal{S}|-1)$ where $d$ is the number of sites, and $| \mathcal{S}|$ the number of possible values at each site. The loss $\Lexp$ may be simplified by the Markovian assumption as follows:
\begin{equation*}
\begin{split}
\Lexp(\bm{\theta})
&= \mathbb E_{\x\sim\dgp}\!\left[\frac{1}{m}\sum_{\x'\in M(\x)}
\Big(\log\tfrac{\model(\x')}{\model(\x)}\Big)^2
-2\log\tfrac{\model(\x')}{\model(\x)}\log\tfrac{\dgp(\x')}{\dgp(\x)}\right] \\[2pt]
&= \mathbb E_{\x\sim\dgp}\!\left[\frac{1}{m}\sum_{i=1}^d\sum_{x_i' \in \mathcal{S}}
\Big(\log\tfrac{\model(x_i'\mid \text{nb}(x_i))}{\model(x_i\mid \text{nb}(x_i))}\Big)^2
-2\log\tfrac{\model(x_i'\mid \text{nb}(x_i))}{\model(x_i\mid \text{nb}(x_i))}\log\tfrac{\dgp(x_i'\mid \text{nb}(x_i))}{\dgp(x_i\mid \text{nb}(x_i))}\right].
\end{split}
\end{equation*}
Next, we use the linearity of expectation to get:
\begin{equation*}
\begin{split}
\Lexp(\bm{\theta})
&= \frac{1}{d}\sum_{i=1}^d
\mathbb E_{\x\sim\dgp}\left[\frac{1}{m}\sum_{x_i'\in \mathcal{S}}
\Big(\log\tfrac{\model(x_i'\mid \text{nb}(x_i))}{\model(x_i\mid \text{nb}(x_i))}\Big)^2
-2\log\tfrac{\model(x_i'\mid \text{nb}(x_i))}{\model(x_i\mid \text{nb}(x_i))}\log\tfrac{\dgp(x_i'\mid \text{nb}(x_i))}{\dgp(x_i\mid \text{nb}(x_i))}\right],
\end{split}
\end{equation*}
and by locality (the integrand depends only on $(x_i, \text{nb}(x_i))$, we marginalize the expectation to $(x_i, \text{nb}(x_i))$. Note that $x'_i$ also depends only on $x_i$. Therefore,
\begin{equation*}
\begin{split}
\Lexp(\bm{\theta})
&= \frac{1}{d}\sum_{i=1}^d
\mathbb E_{(x_i,\text{nb}(x_i))}\left[\frac{1}{|\mathcal{S}|}\sum_{x_i'\in \mathcal{S}}
\Big(\log\tfrac{\model(x_i'\mid \text{nb}(x_i))}{\model(x_i\mid \text{nb}(x_i))}\Big)^2
-2\log\tfrac{\model(x_i'\mid \text{nb}(x_i))}{\model(x_i\mid \text{nb}(x_i))}\log\tfrac{\dgp(x_i'\mid \text{nb}(x_i))}{\dgp(x_i\mid \text{nb}(x_i))}\right].
\end{split}
\end{equation*}
Finally, by stationarity (identically distributed sites, or equivalently translation invariance), we can collapse the site-average to a single-site expectation:
\begin{equation*}
\begin{split}
\Lexp(\bm{\theta})
&= \mathbb E_{(x,\text{nb}(x))}\left[\frac{1}{|\mathcal{S}|}\sum_{x'\in \mathcal{S}}
\Big(\log\tfrac{\model(x'\mid \text{nb}(x))}{\model(x\mid \text{nb}(x))}\Big)^2
-2\log\tfrac{\model(x'\mid \text{nb}(x))}{\model(x\mid \text{nb}(x))}\log\tfrac{\dgp(x'\mid \text{nb}(x))}{\dgp(x\mid \text{nb}(x))}\right].
\end{split}
\end{equation*}

\subsubsection{Additional Details for the Ising Model in \Cref{exp:Ising-model}}\label{appendix:Ising-model}

\begin{figure}[t]
    \centering
    
    \begin{subfigure}{0.26\textwidth}
        \centering
        \includegraphics[width=\linewidth]{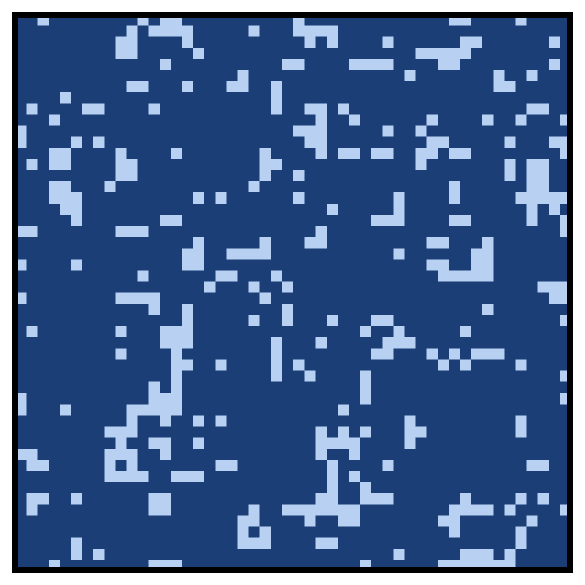}
        \caption{$50 \times 50$}
        \label{fig:first-grid-50}
    \end{subfigure}
    \hfill
    \begin{subfigure}{0.26\textwidth}
        \centering
        \includegraphics[width=\linewidth]{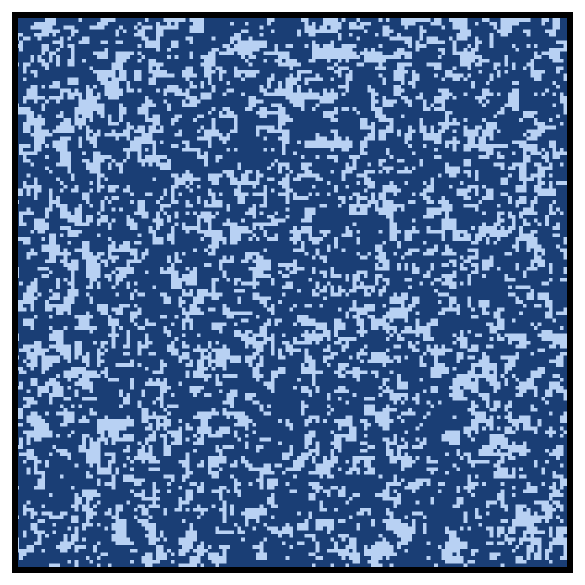}
        \caption{$150 \times 150$}
        \label{fig:second-grid-150}
    \end{subfigure}
    \hfill
    \begin{subfigure}{0.26\textwidth}
        \centering
        \includegraphics[width=\linewidth]{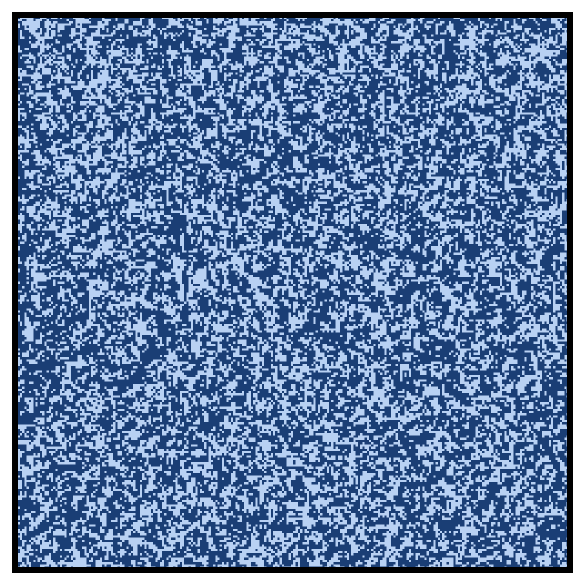}
        \caption{$250 \times 250$}
        \label{fig:third-grid-250}
    \end{subfigure}
    
    \caption{\textit{Three draws from an Ising model with increasing grid size.} Here, $\btheta=(0.15, 0.30).$}
    \label{fig:ising-model-samples}
\end{figure}

In the \textbf{Ising model} experiment, we compare the average of posterior means and computational cost of DFD, LRM and PL. The data is simulated from an Ising model on a 2D lattice with parameters $(\theta_1, \theta_2) = (0.15, 0.30)$. We investigate the results at increasing grid sizes $(50^2, 100^2, 150^2, 200^2, 250^2, 300^2)$, and provide sample visualizations of Ising model realizations for different grid sizes in \Cref{fig:ising-model-samples}. 
The simulation is performed via Gibbs sampling, which repeatedly picks a random lattice site and resamples the state from the site’s conditional distribution; one sweep corresponds to about one update per site, and we adjust the number of sweeps according to the grid size.
We assess the algorithm's convergence by monitoring the trace of the average magnetization $m_t=\frac{1}{d} \sum_i x_i^{(t)}$.
For all methods, the prior mean and prior covariance used are $\mu= (0.5, 0.5)^\top$ and $\Sigma=2 I_2$.

For PL and DFD, which both use MCMC sampling, we obtain 5000 MCMC samples, with 1000 warmup steps.
We run four chains and use the Gelman-Rubin statistic to assess how many samples are required for the chains to have mixed. 
For LRM, we use $\alpha=0.1$ and the matching set as specified in \Cref{sec:experiments} (with two states). Estimating $\beta$ is done as outlined in \Cref{sec:calibration-selection} and \Cref{appendix:posterior-calibration}, with $B=50$, and 95\% target coverage. 

For each grid size, the Ising model simulation is conducted 20 times. The average and standard deviation of the posterior means are computed for each method. To estimate computational cost, we run each method 5 times and compute the average cost and standard deviation (since computational time does not vary much, we do not need to run each method many times). 

Furthermore, to assess sensitivity to $\alpha$, we conduct an experiment where a single LRM-Bayes posterior is computed for data simulated from the same $(\theta_1, \theta_2)$ as for \Cref{fig:ising-model-samples} and a grid of size $100^2$. 
The results are in \Cref{fig:sensitivity-ising-model}, which show that the posterior is practically insensitive to $\alpha$ in this setting.

\begin{figure}[h!]
    \centering
    \includegraphics[width=0.5\linewidth]{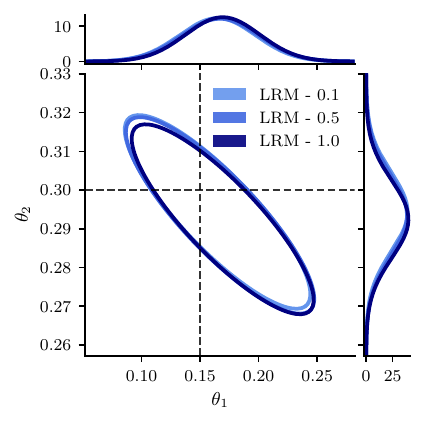}
    \caption{\textit{LRM-Bayes posterior's sensitivity to $\alpha$ for the Ising model}.}
    \label{fig:sensitivity-ising-model}
\end{figure}

\subsubsection{Additional Details for the Potts Model in \Cref{exp:Potts-model}}\label{appendix:Potts-model}
In the \textbf{Potts model} experiment, we study the Antarctic ice-bedsheet data.
This data was studied in \citet{lee2024steingradientdescentapproach}, and can be obtained from their GitHub repository \url{https://github.com/codinheesang/MCSVGD}, under the folder ``POTTS", with filename ``ice\_potts\_data.RData". 
The lattice is $171 \times 171$, and the ice thickness, $y$, is classified into four categories: $x_i=0$ for no ice, $x_i=1$ for $0 <y \leq 1000$, $x_i=2$ for $1000 < y \leq 2000 $, $x_i=3$ for $y>2000$.
All methods use a $\mathcal{N}(0, 10)$ prior. 
All MCMC methods retrieve 2000 samples with 500 burn-in steps. 
For auxiliary MCMC, we conduct 30 inner loops per MCMC step. 
For DFD, we estimate $\beta$ using the method outlined in \Cref{sec:methods} (i.e. same estimation procedure as for LRM).
For LRM, we use $\alpha=1$ in this setting because the data are sparse (i.e., many conditional probabilities are zero), which adds 1 count to each conditional outcome. 
Further, for numerical stability in LRM, we truncate terms with conditional probabilities below the 5\% quantile.

\subsubsection{Additional Simulated Experiment on the Potts Model}\label{appendix:add-exp-potts}
\begin{figure}
    \centering
    \includegraphics[width=0.9\linewidth]{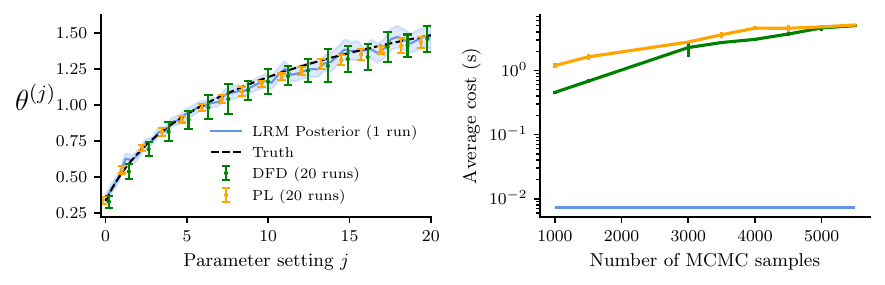}
    \caption{\textit{Simulated Potts model.} Left: inference across fixed parameter settings \(\theta^{(j)}\). The dashed line is the truth; blue shows LRM-Bayes posterior mean \(\pm 1.96\) standard deviations; error bars show mean \(\pm 1.96\) standard deviations across 20 runs for DFD-Bayes and pseudo-likelihood. Right: average computation time versus total MCMC iterations, with the horizontal line giving the average LRM-Bayes cost.}
    \label{fig:potts-model-exp-sim}
\end{figure}

We next study the performance of LRM-Bayes, DFD-Bayes, and pseudo-likelihood as a function of the true Potts interaction parameter \(\theta\).
This is informative because different values of \(\theta\) correspond to different regimes: smaller values yield weaker spatial dependence and more disordered grids, whereas larger values yield stronger spatial dependence and more coherent spatial regions.
For each value \(\theta^{(j)}\), we generate a \(100\times100\) Potts grid with \(|\mathcal{S}|=4\) states and compare posterior inference across the three methods.
LRM-Bayes is run once across the sequence, with \(\beta\) selected by the coverage-calibration procedure, while DFD-Bayes and pseudo-likelihood are repeated over 20 independent runs.
The results in \Cref{fig:potts-model-exp-sim} show that, like competing methods, LRM-Bayes covers the true parameter value across a range of regimes while being substantially more computationally efficient.

\subsubsection{Impact of the smoothing parameter on LRM-Bayes for small Potts grids}\label{appendix:alpha-vs-small-grids}

This experiment illustrates that, for smaller grids and hence sparser settings, the LRM-Bayes posterior can be more sensitive to the choice of smoothing parameter $\alpha$.
We simulate a $50 \times 50$ Potts model with $3$ states and true parameter $1/\theta=1.3$.
We then compute the LRM-Bayes posterior for $\alpha \in \{0,0.1,0.5,1,2\}$, using a Gaussian prior with mean $0.5$ and variance $10$.
For $\alpha=0$, terms involving unobserved matching points are truncated, rather than smoothed; we include this case to illustrate that, in some settings, this truncation-based version can still yield accurate posterior inference.
\Cref{fig:post-dens-for-vary-alpha} shows that, on smaller grids, changing $\alpha$ more visibly affects the estimated local ratios and, in turn, the posterior.
This is consistent with the role of $\alpha$ as a Laplace-smoothing parameter: larger $\alpha$ pulls empirical conditional probabilities toward uniform. This weakens the empirical local log-ratios, so the fitted interaction parameter $\theta$ tends to decrease.

\begin{figure}
    \centering
    \includegraphics[width=0.6\linewidth]{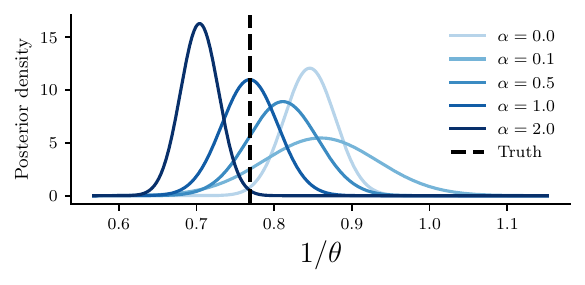}
    \caption{\textit{Sensitivity to smoothing for a small Potts grid.} We simulate a $50 \times 50$ Potts model with $3$ states and true parameter $1/\theta=1.3$. We plot the Gaussian LRM-Bayes posteriors for $\alpha \in \{0,0.1,0.5,1,2\}$. The dashed line denotes the true parameter value. The smaller grid makes the empirical conditional counts less stable, so the posterior is more visibly affected by the choice of $\alpha$.}
    \label{fig:post-dens-for-vary-alpha}
\end{figure}

\section{Extension: Robustness Through a Weighted LRM-Bayes Posterior}\label{appendix:robust-loss}

Robustness to outliers with generalized posteriors can be achieved with a weighted loss (for example, see \citet{altamirano2024robust, duran2024outlier, matsubara2022robust, laplante2025robust}). 
We can introduce robustness similarly with a weight function $w : \X \rightarrow (0, \infty)$ that downweights unreliable observations, extending the definition of the divergence as follows.
Suppose $q \in \AllPMFs$ and $p \in \AllPMFsGiven{q}$. Then, define the \emph{weighted LRM divergence} as:
\[
D^{\operatorname{LRM}}_w(q \| p) := \mathbb{E}_{\x \sim q} \left [ \frac{1}{|M(\x)|} \sum_{\x' \in M(\x)}   w(\x) \left(\log \frac{p(\x')}{p(\x)} - \log \frac{q(\x')}{q(\x)} \right)^2 \right].
\]
This remains a statistical divergence so long as $w(\x) > 0$ wherever $q(\x)>0$, and some regularity conditions are satisfied.
In particular, with the weighted log-ratio operator $\mathcal{R}^w_{j}[q](\x):=\sqrt{w(\x)} \log \frac{q(M_j(\x))}{q(\x)}$, define:
\[
\mathcal{Q}_q^{\text{adm}(w)}(\X):= \Big\{ p \in \AllPMFsGiven{q} \mid \mathcal{R}^w_{j}[p](\x) \in L^2(q, \mathbb{R}) \; \forall j \in \mathcal{J} \Big \}.
\]
Then, we would require $q, p \in\mathcal{Q}_q^{\text{adm}(w)}(\X) $ and \Cref{assumption:graph-connect} to hold to obtain $D_w^{\operatorname{LRM}}(q \| p) = 0 \iff q = p$.
Further, this weighted divergence would yield a loss similar to the one from \Cref{eq:log-ratio-matching-loss}, expressed as follows:
\begin{equation}
\hat{\mathcal{L}}^{\text{LRM}}_w(\bm{\theta}; \{\x_i \}_{i=1}^n) := \frac{1}{n} \sum_{i=1}^n \frac{w(\x_i)}{|M(\bm{x}_i)|} \sum_{\bm{x}' \in M(\bm{x}_i)}  \left(\log \frac{\model(\bm{x}')}{\model(\bm{x}_i)}\right)^2 -2 \log\frac{\model(\bm{x}')}{\model(\bm{x}_i)}\log \frac{\hat{q}(\bm{x}')}{\hat{q}(\bm{x}_i)} \label{eq:log-ratio-matching-loss-outliers}.
\end{equation}
A popular choice of weights that has been used frequently in the literature is the inverse multi-quadratic function \citep[e.g., see][]{matsubara2022robust, riabiz2022optimal, barp2019minimum}.
However, the shape of the IMQ is symmetric, and in many cases would not align well with the distribution of count data. 
We instead propose an alternative for count data problems: the product of Poisson marginals with a robust statistics parametrizing the densities.
This better reflects the semi-bounded domain of count data than the IMQ.
It is given by:
\begin{equation*}
w(\mathbf{x}; \bm{\mu}) = \prod_{j=1}^d \frac{e^{-\mu_j} \mu_j^{x_j}}{x_j!}
= \exp\left(-\sum_{j=1}^d \mu_j \right) \cdot \prod_{j=1}^d \frac{\mu_j^{x_j}}{x_j!}
\end{equation*}
where $\bm{\mu}:=(\mu_1, \dots, \mu_d)$, for $\mu_j := \text{med}(x_j)$, which denotes the median of $\{x_j^{(1)}, x_j^{(2)}, \dots, x_j^{(n)} \}$.
In \Cref{fig:robust-cmp-univariate}, we demonstrate the benefits of downweighting.
While the resulting CMP density estimates differ only slightly, the posterior derived from the robust empirical loss $\hat{\mathcal{L}}_w$ is considerably more stable and reliable.

\begin{figure}[t]
    \centering
    \includegraphics[width=\linewidth]{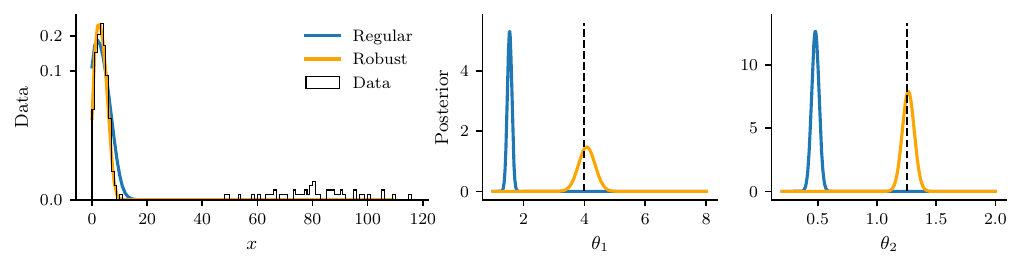}
    \caption{\textit{Posterior Fit with CMP Contaminated Data}. We simulate 1000 data points from a CMP model with $(\theta_1, \theta_2)=(4, 1.25)$ and introduce 5\% contamination from a high mean Poisson distribution. We obtain a posterior distribution with $\alpha=0$ using  $\hat{\mathcal{L}}^{\text{LRM}}_n$ (regular) and $\hat{\mathcal{L}}_w^{\text{LRM}}$ (robust); then, we show estimates of the CMP density with posterior means in the leftmost plot, and the marginals of $\theta_1$ and $\theta_2$ in the two rightmost plots. }
    \label{fig:robust-cmp-univariate}
\end{figure}

\section{Theoretical Assessment}\label{appendix:theoretical-assessment}

In this Section, we extend the theoretical guarantees of the LRM estimator to a broader class of models. To do so, additional regularity conditions are required.
The first set of conditions concerns the parameter space and the existence of the empirical minimizer.
\begin{assumption}
\label{assumption:general_assumptions}
The following regularity conditions apply: 
  \begin{assumpenum}
    \item  \label{assumption:theta-bounded}
    $\Theta \subseteq \mathbb{R}^p$ is  open, convex and bounded;
    \item \label{assumption:minimiser}
      The minimizer $\hat{\btheta}_n:= \arg \min_{\btheta \in \Theta} \Lhat(\btheta)$  exists for all $n$ large enough and $\dgp$-almost surely.
\end{assumpenum}
\end{assumption}

A further condition is imposed on the mapping $\btheta \mapsto \model(\cdot)$ to ensure sufficient smoothness. The required strength of this condition varies depending on the specific result.
\begin{assumption}
Assume that $\btheta  \mapsto \model(\x)$ is thrice continuously differentiable for any $\x\in\X$, and that for some $r \in \mathbb{N}$ and all $\rho =0, \dots, r$, one has
    \begin{equation}
       \max_{\x'\in M(\x)}\sup_{\btheta\in\Theta}\left\|\partial_{\btheta}^{\rho}\log \left (\frac{\model(\x')}{\model(\x)} \right)\right\| < K_{\rho}(\x),
\end{equation}
such that $K_{\rho}\in L^2(\dgp, \mathbb{R})$ and where we denote $\partial_{\btheta}^{0}f_{\btheta} := f_{\btheta}$.
\label{assumption:theta-diff}
\end{assumption}
These types of assumptions are standard in the literature on divergence-based estimation \citep{ghosh2016robust,barp2019minimum,matsubara2022robust,matsubara2024generalized}. In particular, \cref{assumption:theta-diff} can be viewed as a natural extension of Assumption 2 in \citet{matsubara2024generalized}, where integrability of the local probability ratios is required. Here, we extend this condition to the entire matching set.
Equipped with these assumptions, we can now derive the following results:
\begin{proposition}[a.s. Pointwise Convergence] Suppose \Cref{assumption:laplace-estimator}, \Cref{assumption:general_assumptions} and \Cref{assumption:theta-diff} for $r=0$ hold.
Then,
    $\Lhat(\btheta) \overset{\text{a.s.}}{\longrightarrow} \Lexp(\btheta)$, pointwise for all $\bm{\theta} \in \Theta$.
    \label{prop:loss_convergence}
\end{proposition}
\begin{proof}
    Define $f_{\btheta}(\x,\x')=\log\frac{\model(\x')}{\model(\x)}$, 
$g(\x,\x')=\log\frac{\dgp(\x')}{\dgp(\x)}$, and 
$\hat g(\x,\x')=\log\frac{\hat{q}_\alpha(\x')}{\hat{q}_\alpha(\x)}$.
Then
\begin{align*}
\Lhat(\btheta)&=\frac{1}{n}\sum_{i=1}^n\frac{1}{|M(\x_i)|}\sum_{\x'\in M(\x_i)}\!\left[f_{\btheta}(\x_i,\x')^2-2f_{\btheta}(\x_i,\x')\hat g(\x_i,\x')\right],\\
\Lexp(\btheta)&=\mathbb{E}_{\x\sim q}\!\left[\frac{1}{|M(\x)|}\sum_{\x'\in M(\x)}\!\left[f_{\btheta}(\x,\x')^2-2f_{\btheta}(\x,\x')g(\x,\x')\right]\right].
\end{align*}
Write $\Lhat(\btheta)=A_n-2\hat B_n$, where $A_n=\frac{1}{n}\sum_{i=1}^n\frac{1}{|M(\x_i)|}\sum_{\x'\in M(\x_i)}f_{\btheta}(\x_i,\x')^2$ and $\hat B_n=\frac{1}{n}\sum_{i=1}^n\frac{1}{|M(\x_i)|}\sum_{\x'\in M(\x_i)}f_{\btheta}(\x_i,\x')\hat g(\x_i,\x')$.

By the SLLN and \Cref{assumption:expectation}, $A_n\xrightarrow{\text{a.s.}}A:=\mathbb{E}_{\x\sim \dgp}\left[\frac{1}{|M(\x)|}\sum_{\x'\in M(\x)}f_{\btheta}(\x,\x')^2\right]$. For $\hat{B}_n$, we add and subtract $B_n:=\frac{1}{n}\sum_{i=1}^n\frac{1}{|M(\x_i)|}\sum_{\x'\in M(\x_i)}f_{\btheta}(\x_i,\x')g(\x_i,\x')$, to obtain $\hat B_n=B_n+R_n$ with $R_n=\frac{1}{n}\sum_{i=1}^n\frac{1}{|M(\x_i)|}\sum_{\x'\in M(\x_i)}f_{\btheta}(\x_i,\x')\!\big(\hat g(\x_i,\x')-g(\x_i,\x')\big)$. Again by the SLLN, $B_n\xrightarrow{\text{a.s.}}B :=\mathbb{E}_{\x\sim \dgp}\left[\frac{1}{|M(\x)|}\sum_{\x'\in M(\x)}f_{\btheta}(\x,\x')g(\x,\x')\right]$.

To handle the remainder term, we consider separately the two cases in \Cref{assumption:laplace-estimator}: (i) the finite case and (ii) the infinite case.

\noindent\textbf{Finite Case.}
If \Cref{assumption:laplace-estimator} (i) holds, let $\dgp^{\min}>0$ and set any $\varepsilon\le \dgp^{\min}$. By \Cref{lemma:laplace-estimator}, $\hat\Delta_\varepsilon:=\sup_{\dgp(\x)\ge\varepsilon}\max_{\x'\in M(\x)}|\hat g-g|\to0$ a.s., and here $\{\dgp\ge\varepsilon\}=\X$, so $|R_n|\le \hat\Delta_\varepsilon\, \frac{1}{n}\sum_{i=1}^n\frac{1}{|M(\x_i)|}\sum_{\x'\in M(\x_i)}|f_{\btheta}(\x_i,\x')|\xrightarrow{\text{a.s.}}0$, since the average converges a.s. to a finite constant by \Cref{assumption:theta-diff}.

\noindent\textbf{Infinite Case.}
Fix $\varepsilon>0$ and decompose $R_n=R_n^{(1)}(\varepsilon)+R_n^{(2)}(\varepsilon)$, where
\begin{align*}
R_n^{(1)}(\varepsilon)
&:=\frac{1}{n}\sum_{i=1}^n\mathbf 1\{\dgp(\x_i)\ge\varepsilon\} \frac{1}{|M(\x_i)|}\sum_{\x'\in M(\x_i)}f_{\btheta}(\x_i,\x')\big(\hat g(\x_i,\x')-g(\x_i,\x')\big),\\
R_n^{(2)}(\varepsilon)
&:=\frac{1}{n}\sum_{i=1}^n\mathbf 1\{\dgp(\x_i)<\varepsilon\} \frac{1}{|M(\x_i)|}\sum_{\x'\in M(\x_i)} f_{\btheta}(\x_i,\x')\big(\hat g(\x_i,\x')-g(\x_i,\x')\big).
\end{align*}

By \Cref{lemma:laplace-estimator}, $\Delta_\varepsilon :=\sup_{\dgp(\x)\ge\varepsilon}\max_{\x'\in M(\x)} |\hat g(\x,\x')-g(\x,\x')|\xrightarrow{\mathrm{a.s.}}0.$ Hence
\begin{align*}
|R_n^{(1)}(\varepsilon)|\le \Delta_\varepsilon\,
\frac{1}{n}\sum_{i=1}^n\mathbf 1\{\dgp(\x_i)\ge\varepsilon\} \frac{1}{|M(\x_i)|}\sum_{\x'\in M(\x_i)}|f_{\btheta}(\x_i,\x')|\xrightarrow{\mathrm{a.s.}}0,
\end{align*}
since the average converges a.s. by the SLLN. By \cref{assumption:laplace-estimator} (ii), for all $\x$ and $\x'\in M(\x)$, $|\hat g(\x,\x') - g(\x,\x')| \leq |\hat g(\x,\x')| + |g(\x,\x')| \leq (K_{\hat{q}_\alpha,n} + K_{\dgp})\,(1+\|\x\|^\gamma)$.
Since $\limsup_{n\to\infty} K_{\hat q_\alpha,n} < \infty$, define 
$K_{\hat q_\alpha} := \limsup_{n\to\infty} K_{\hat q_\alpha,n} + 1$. 
Then, by the definition of the limsup, there exists $N\in\mathbb{N}$ such that 
for all $n \ge N$, $K_{\hat{q}_\alpha,n} \le K_{\hat{q}_\alpha}$. Thus, for all $n \ge N$, $|\hat g(\x,\x') - g(\x,\x')| \leq (K_{\hat{q}_\alpha} + K_{\dgp})\,(1+\|\x\|^\gamma)$.
Hence,
\begin{align*}
|R_n^{(2)}(\varepsilon)| \leq \frac{(K_{\hat{q}_\alpha}+K_{\dgp})}{n}\sum_{i=1}^n \mathbf{1}\{\dgp(\x_i)<\varepsilon\}(1+\|\x_i\|^{\gamma}) \frac{1}{|M(\x_i)|}\sum_{\x'\in M(\x_i)}|f_{\btheta}(\x_i,\x')|.
\end{align*}
Define $H_{\btheta}(\x):= (K_{\hat{q}_\alpha}+K_{\dgp})(1+\|\x\|^{\gamma})\frac{1}{|M(\x)|}\sum_{\x'\in M(\x)}|f_{\btheta}(\x,\x')|$. Since $\dgp$ has sub-exponential tails (\cref{assumption:laplace-estimator} (ii)), it follows that $\mathbb{E}_{X\sim \dgp}\bigl[\|X\|^{k}\bigl] < \infty$ for every integer $k \ge 1$. Therefore, $H_{\btheta}\in L^1(\dgp,\R)$.
Hence, by the SLLN, $|R_n^{(2)}(\varepsilon)| \xrightarrow{\mathrm{a.s.}} \E_{\x\sim \dgp}\big[\mathbf 1\{\dgp(\x)<\varepsilon\}H_{\btheta}(\x)\big]$. Since $H_{\btheta}\in L^1(\dgp,\R)$ and  $\mathbf 1\{\dgp(\x)<\varepsilon\}\downarrow 0$ as $\varepsilon\downarrow 0$,
dominated convergence implies
\begin{align*}
\E_{\x\sim \dgp}\big[\mathbf{1}\{\dgp(\x)<\varepsilon\}H_{\btheta}(\x)\big]
\xrightarrow[\varepsilon\downarrow 0]{}0.
\end{align*}
Choosing a sequence $\varepsilon_n\downarrow 0$ slowly enough gives
$R_n^{(1)}(\varepsilon_n)\to 0$ and $R_n^{(2)}(\varepsilon_n)\to 0$ a.s., hence
$R_n\to 0$ a.s. In both cases, finite and countably infinite $\X$, we have  $A_n\to A$, $B_n\to B$, and $R_n\to0$ almost surely. Therefore
\begin{align*}
\Lhat(\btheta)=A_n-2\hat B_n=A_n-2(B_n+R_n)\xrightarrow{\mathrm{a.s.}}A-2B=\Lexp(\btheta).
\end{align*}
\end{proof}

\Cref{prop:loss_convergence} establishes almost sure pointwise convergence of the empirical loss under the weakest regularity conditions. Strengthening the assumptions slightly yields uniform convergence, as follows:
\begin{proposition}[a.s. Uniform Convergence]  Suppose \Cref{assumption:laplace-estimator}, \Cref{assumption:general_assumptions} and \Cref{assumption:theta-diff} for $r=1$ hold. Then,
    $\sup_{\bm{\theta} \in \Theta} | \Lhat(\btheta) - \Lexp(\btheta)| \overset{\text{a.s.}}{\longrightarrow} 0$.
    \label{prop:loss_convergence_unif}
\end{proposition}
\begin{proof}
By \Cref{prop:loss_convergence}, $\Lhat\to \Lexp$ pointwise a.s.\ on $\Theta$. 
To upgrade to uniform convergence, it suffices to show that the family 
$\{\Lhat : n\in\mathbb{N}\}$ is strongly stochastically equicontinuous 
\citep[Theorem~21.8]{davidson1994stochastic}.

By \citet[Theorem~21.10]{davidson1994stochastic}, the family 
$\{\Lhat : n\in\mathbb{N}\}$ is strongly stochastically equicontinuous 
on $\Theta$ if there exists a stochastic sequence $\{G_n:n\in\mathbb{N}\}$, independent of $\theta$, such that $|\Lhat(\btheta)-\Lhat(\btheta')| \leq G_n\|\btheta-\btheta'\|_2$, for all $\btheta,\btheta'\in\Theta$, and $\limsup_{n\to\infty}G_n<\infty$ almost surely. 

Since $\Lhat$ is an algebraic combination and composition of the log-likelihood ratio 
$\log(\model(\x')/\model(\x))$ and the estimator $\hat{q}_{\alpha}$, and $\model$ is continuously differentiable
in $\btheta$ by \Cref{assumption:theta-diff}, it follows that $\Lhat$ is
continuously differentiable in $\btheta$. Moreover, $\Theta\subset\mathbb R^d$ is convex
by \Cref{assumption:theta-bounded}, so the mean value theorem 
\citep[Theorem~9.19]{rudin1976principles} yields
\begin{align*}
|\Lhat(\btheta)-\Lhat(\btheta')| \leq\sup_{\vartheta\in\Theta}\|\nabla_{\btheta} \Lhat(\vartheta)\|_2 \cdot \|\btheta-\btheta'\|_2, \qquad\forall \btheta,\btheta'\in\Theta.    
\end{align*}
Hence we can define $G_n := \sup_{\btheta\in\Theta}\|\nabla_{\btheta} \Lhat(\btheta)\|_2$. It remains to show that $\limsup_{n\to\infty}G_n<\infty$ almost surely. Recall $\Lhat(\btheta)=A_n(\btheta)-2\widehat{B}_n(\btheta)$. Differentiating and using the product/chain rules,
\begin{align*}
\nabla_{\btheta} A_n(\btheta)
&= \frac{2}{n}\sum_{i=1}^n\frac{1}{|M(\x_i)|}\sum_{\x'\in M(\x_i)}
 f_{\btheta}(\x_i,\x')\,\nabla_{\btheta} f_{\btheta}(\x_i,\x'),
\\
\nabla_{\btheta} \hat B_n(\btheta)
&= \frac{1}{n}\sum_{i=1}^n\frac{1}{|M(\x_i)|}\sum_{\x'\in M(\x_i)}
\nabla_{\btheta} f_{\btheta}(\x_i,\x')\,\hat{g}(\x_i,\x').
\end{align*}
Thus
\begin{align*}
\nabla_{\btheta} \Lhat(\btheta)
= \frac{2}{n}\sum_{i=1}^n\frac{1}{|M(\x_i)|}\sum_{\x'\in M(\x_i)}
\Big( f_{\btheta}(\x_i,\x')\,\nabla_{\btheta} f_{\btheta}(\x_i,\x') 
      - \nabla_{\btheta} f_{\btheta}(\x_i,\x')\,\hat{g}(\x_i,\x') \Big).
\end{align*}

Taking Euclidean norms and applying the triangle inequality,
\begin{align*}
\|\nabla_{\btheta} \Lhat(\btheta)\|_2
&\leq \frac{2}{n}\sum_{i=1}^n\frac{1}{|M(\x_i)|}\sum_{\x'\in M(\x_i)}
\big|f_{\btheta}(\x_i,\x')\big|\;\big\|\nabla_{\btheta} f_{\btheta}(\x_i,\x')\big\|_2 \\
&\quad + \frac{2}{n}\sum_{i=1}^n\frac{1}{|M(\x_i)|}\sum_{\x'\in M(\x_i)}
\big\|\nabla_{\btheta} f_{\btheta}(\x_i,\x')\big\|_2\;\big|\hat{g}(\x_i,\x')\big|.
\end{align*}

By \Cref{assumption:theta-diff}, there exist measurable functions $K_0,K_1\in L^1(\dgp, \mathbb{R})$ such that
\begin{align*}
\sup_{\x'\in M(\x)}\sup_{\btheta\in\Theta}\big|f_{\btheta}(\x,\x')\big| \le K_0(\x),\qquad
\sup_{\x'\in M(\x)}\sup_{\btheta\in\Theta}\big\|\nabla_{\btheta}f_{\btheta}(\x,\x')\big\|_2 \le K_1(\x).
\end{align*}
Hence, for all $\btheta\in\Theta$,
\begin{align*}
\|\nabla_{\btheta} \Lhat(\btheta)\|_2
\le \frac{2}{n}\sum_{i=1}^n E_1(\x_i) + \frac{2}{n}\sum_{i=1}^n \hat{E}_2(\x_i), 
\end{align*}
where $E_1(\x):= K_0(\x)K_1(\x)$ and $\hat{E}_2(\x):=\frac{K_1(\x)}{|M(\x)|}\sum_{\x'\in M(\x)} |\hat{g}(\x,\x')|$.

Similarly to \Cref{prop:loss_convergence}, by \Cref{assumption:laplace-estimator}, there exists a constant $K_{\hat{q}_{\alpha}}<\infty$ such that for large enough $n$, all $x$, and all $\x'\in M(\x)$, $|\hat{g}(\x,\x')|\leq K_{\hat{q}_{\alpha}}(1+\|\x\|^{\gamma})$.
Thus
\begin{align*}
\widehat{E}_2(\x) \le \frac{K_1(\x)}{|M(\x)|}\sum_{\x'\in M(\x)} K_{\hat{q}_{\alpha}}(1+\|\x\|^{\gamma})
= K_1(\x)K_{\hat{q}_{\alpha}}(1+\|\x\|^{\gamma}):= J_2(\x),
\end{align*}
and we have $\|\nabla_{\btheta} \Lhat(\btheta)\|_2\le \frac{2}{n}\sum_{i=1}^n E_1(\x_i) +\frac{2}{n}\sum_{i=1}^n J_2(\x_i)$,

for all $\btheta\in\Theta$. By \Cref{assumption:theta-diff},
$E_1,J_2\in L^1(\dgp, \mathbb{R})$, so by the SLLN,
\begin{align*}
\frac{1}{n}\sum_{i=1}^n E_1(\x_i)\xrightarrow{\text{a.s.}}\E_{\x\sim \dgp}[E_1(\x)],\qquad
\frac{1}{n}\sum_{i=1}^n J_2(\x_i)\xrightarrow{\text{a.s.}}\E_{\x\sim \dgp}[J_2(\x)],
\end{align*}
both finite. Therefore
\begin{align*}
\limsup_{n\to\infty}
\sup_{\btheta\in\Theta}\|\nabla_{\btheta} \Lhat(\btheta)\|_2
\le 2\big(\mathbb{E}_{\x\sim \dgp}[E_1(\x)]+\mathbb{E}[J_2(\x)]\big) <\infty\qquad\text{a.s.}
\end{align*}
That is, $\limsup_{n\to\infty}G_n<\infty$ almost surely.

We have (i) pointwise a.s.\ convergence $\Lhat(\btheta)\to\Lexp(\btheta)$  for all $\theta$ by \Cref{prop:loss_convergence}, and (ii) strong stochastic equicontinuity  of $\{\Lhat:n\in\mathbb{N}\}$ via the bound on $\{G_n:n\in\mathbb{N}\}$. 
Hence, by \citet[Theorem~21.8]{davidson1994stochastic}, $\sup_{\btheta\in\Theta} \big|\Lhat(\btheta)-\Lexp(\btheta)\big| \xrightarrow{\textnormal{a.s.}} 0$.
\end{proof}

Our next results establish the consistency and asymptotic normality of the estimator $\hat{\btheta}_n$ that minimizes the LRM loss. 
Common in the literature on M-estimators, and less restrictive than assuming convexity, we assume the existence of empirical minimizers and the uniqueness of the population minimizer as follows:
\begin{assumption}
    There exist minimizers $\hat{\btheta}_n$ of $\Lhat$ for all $n \in \mathbb{N}$, and there is a unique $\bm{\theta}_\star$ such that $\Lexp(\bm{\theta}_\star) < \inf_{\{\bm{\theta} \in \Theta : \| \bm{\theta} - \bm{\theta}_\star \| \geq \epsilon \}} \Lexp(\bm{\theta})$ for any $\epsilon > 0$.
    \label{assumption:minimisers}
\end{assumption}
The uniqueness of $\btheta_\star$ holds automatically in the well-specified setting, i.e., when there exists $\btheta_0$ such that $p_{\btheta_0} = \dgp$, since the LRM defines a divergence whenever the matching set $M$ induces a connected graph.
\begin{proposition} [Consistency]
Suppose \Cref{assumption:laplace-estimator} holds, and \Cref{assumption:general_assumptions}, \Cref{assumption:theta-diff} for $r=1$, and  \Cref{assumption:minimisers}. Then, $\hat{\btheta}_n \overset{\text{a.s.}}{\longrightarrow} \btheta_\star$.
\label{prop:consistency}
\end{proposition}
\begin{proof}
    
By \Cref{prop:loss_convergence} and \Cref{assumption:theta-diff}, each 
$\Lhat$ and $\Lexp$ is continuous on $\Theta$ almost surely, and
\begin{equation*}
    \sup_{\bm{\theta}\in\Theta} 
    \bigl| \Lhat(\bm{\theta}) - \Lexp(\bm{\theta}) \bigr| 
    \xrightarrow{\text{a.s.}} 0.
\end{equation*}
From \Cref{assumption:minimisers}, $\Lexp$ has a unique and well-separated minimizer 
$\bm{\theta}_*\in\Theta$. That is, for every $\epsilon>0$,
\begin{equation*}
    \Lexp(\bm{\theta}_*) 
    < \inf_{\{\bm{\theta}\in\Theta:\|\bm{\theta}-\bm{\theta}_*\|\ge\epsilon\}} 
    \Lexp(\bm{\theta}).
\end{equation*}
Thus, almost surely:
(1) $\Theta\subset\mathbb{R}^p$ is open and bounded,
(2) $\Lhat$ and $\Lexp$ are continuous,
(3) $\Lhat\to \Lexp$ uniformly on $\Theta$, and
(4) $\Lexp$ has a unique, well-separated minimizer $\bm{\theta}_*$.
These are exactly the hypotheses of \citet[Lemma~7]{matsubara2022robust}), therefore any sequence of empirical minimizers $\{\hat{\bm{\theta}}_n:n\in\mathbb{N}\}$ with 
$\hat{\bm{\theta}}_n \in \arg\min_{\bm{\theta}\in\Theta} \Lhat(\bm{\theta})$ 
for all sufficiently large $n$ satisfies $\hat{\bm{\theta}}_n \xrightarrow{\text{a.s.}} \bm{\theta}_*$, which establishes consistency.

\end{proof}
With additional regularity conditions, the asymptotic normality of the estimator can be established.
\begin{proposition}[Asymptotic normality]
Suppose Assumption~\ref{assumption:laplace-estimator} holds, together with Assumptions~\ref{assumption:general_assumptions}, \ref{assumption:theta-diff} for $r=3$, and \ref{assumption:minimisers}. Let $ \mathbf H_{\star} := \nabla^2_{\btheta}\,\Lexp (\btheta)\big|_{\btheta=\btheta_\star}$. Then there exists a function $\varphi_{\btheta_\star}\in L^2(\dgp,\R)$ and  $\E_{\x\sim\dgp}[\varphi_{\btheta_\star}(\x)]=\bm 0$ such that
\begin{align*}
\sqrt{n}\,\nabla_{\btheta} \Lhat (\btheta_\star) = \frac{1}{\sqrt{n}}\sum_{i=1}^n \varphi_{\btheta_\star}(\x_i) + o_p(1).
\end{align*}
Let $\mathbf{J}_\star := \E_{\x\sim \dgp}\big[\varphi_{\btheta_\star}(\x)\varphi_{\btheta_\star}(\x)^\top\big]$. Then $\hat{\btheta}_n$ satisfies
\begin{align*}
\sqrt{n}\,(\hat{\btheta}_n-\btheta_\star) \xrightarrow{d} \mathcal{N}\!\left(\mathbf{0},\mathbf{H}_{\star}^{-1} \mathbf{J}_{\star} \mathbf{H}_{\star}^{-1}\right).
\end{align*}
\label{prop:clt}
\end{proposition}

\begin{proof}
Since $\hat{\btheta}_n$ minimizes $\Lhat$, the first-order condition gives $\nabla_{\btheta}\Lhat(\hat{\btheta}_n)=\mathbf 0$. By a second-order Taylor expansion of $\Lhat$ around $\btheta_\star$, there exists $\tilde{\btheta}_n$ on the line segment between $\hat{\btheta}_n$ and $\btheta_\star$
such that
\begin{equation*}
\mathbf{0}  = \nabla_{\btheta} \Lhat(\btheta_\star)  + \nabla^2_{\btheta} \Lhat(\btheta_\star)(\hat{\btheta}_n-\btheta_\star)  + (\hat{\btheta}_n-\btheta_\star)\cdot\nabla_{\btheta}^3\Lhat(\tilde{\btheta}_n),
\end{equation*}
where $\nabla_{\btheta}^3\Lhat$ is the third-order derivative tensor. Rearranging,
\begin{equation}
\sqrt{n}(\hat{\btheta}_n-\btheta_\star) = -\Big(\nabla^2_{\btheta} \Lhat(\btheta_\star) +(\hat{\btheta}_n-\btheta_\star)\cdot\nabla_{\btheta}^3\Lhat(\tilde{\btheta}_n)\Big)^{-1} \sqrt{n}\nabla_{\btheta} \Lhat(\btheta_\star).
\label{eq:theta-Taylor}
\end{equation}
To obtain the CLT it suffices to show:
\begin{itemize}[itemsep=0pt,topsep=0pt]
\item[(i)] 
$\nabla^2_{\btheta} \Lhat(\btheta_\star) +(\hat{\btheta}_n-\btheta_\star)\cdot\nabla_{\btheta}^3\Lhat(\tilde{\btheta}_n) \xrightarrow{p} \mathbf{H}_\star$;
\item[(ii)]
$\sqrt{n}\,\nabla_{\btheta} \Lhat(\btheta_\star) \xrightarrow{d}\mathcal N(\mathbf{0},\mathbf{J}_\star)$,
for some finite covariance matrix $\mathbf{J}_\star$.
\end{itemize}
Since $\mathbf H_\star$ is symmetric positive definite by \Cref{assumption:minimisers}, invertibility of the matrix in parentheses is ensured for large $n$, and Slutsky's theorem applied to \eqref{eq:theta-Taylor} then yields the desired CLT for $\hat{\btheta}_n$.

Write $\Lhat(\btheta)=A_n(\btheta)-2B_n(\btheta) - 2R_n(\btheta)$ and $\Lexp(\btheta)=A(\btheta)-2B(\btheta)$. Differentiating twice,$\nabla^2_{\btheta}\Lhat(\btheta) = \nabla^2_{\btheta}A_n(\btheta) - 2\nabla^2_{\btheta}B_n(\btheta)- 2\nabla^2_{\btheta}R_n(\btheta)$, and analogously $\nabla^2_{\btheta}\Lexp(\btheta)=\nabla^2_{\btheta}A(\btheta)-2\nabla^2_{\btheta}B(\btheta)$.

By Assumption~\ref{assumption:theta-diff} and \ref{assumption:general_assumptions}, the envelopes for $f_{\btheta},\nabla_{\btheta} f_{\btheta},\nabla_{\btheta}^2 f_{\btheta}$ are integrable,
so the SLLN yields $\nabla^2_{\btheta}A_n(\btheta)\xrightarrow{\text{a.s.}}\nabla^2_{\btheta}A(\btheta)$ and  $\nabla^2_{\btheta}B_n(\btheta)\xrightarrow{\text{a.s.}}\nabla^2_{\btheta}B(\btheta)$, in particular at $\btheta=\btheta_\star$.

It remains to show that $\nabla^2_{\btheta}R_n(\btheta_\star)\to 0$ almost surely. By Assumption~\ref{assumption:theta-diff},
$\|\nabla^2_{\btheta}f_{\btheta}(\x,\x')\|\le K_2(\x)$ with $K_2\in L^2(\dgp, \mathbb{R})$, giving
\begin{align*}
\|\nabla^2_{\btheta}R_n(\btheta)\| \le \frac{1}{n}\sum_{i=1}^n K_2(\x_i) \frac{1}{|M(\x_i)|}\sum_{\x'\in M(\x_i)} \big|\hat g(\x_i,\x')-g(\x_i,\x')\big|.
\end{align*}
Fix $\varepsilon>0$ and split according to $\dgp(\x_i)\ge\varepsilon$ or $\dgp(\x_i)<\varepsilon$: $\|\nabla^2_{\btheta}R_n(\btheta)\|\le T_{n,1}(\varepsilon)+T_{n,2}(\varepsilon)$, with
\begin{align*}
T_{n,1}(\varepsilon) &:= \frac{1}{n}\sum_{i=1}^n K_2(\x_i)\mathbf 1\{\dgp(\x_i)\ge\varepsilon\}\frac{1}{|M(\x_i)|}\sum_{\x'\in M(\x_i)} \big|\hat g(\x_i,\x')-g(\x_i,\x')\big|,\\
T_{n,2}(\varepsilon) &:= \frac{1}{n}\sum_{i=1}^n K_2(\x_i)\mathbf 1\{\dgp(\x_i)<\varepsilon\}\frac{1}{|M(\x_i)|}\sum_{\x'\in M(\x_i)}\big|\hat g(\x_i,\x')-g(\x_i,\x')\big|.
\end{align*}
On $\{\dgp\ge\varepsilon\}$ we have the truncated error $\frac{1}{|M(\x)|}\sum_{\x'\in M(\x)} \big|\hat g(\x,\x')-g(\x,\x')\big| \le \Delta_{\varepsilon}^{(n)}$, where $\Delta_{\varepsilon}^{(n)} := \sup_{\{\x:\,\dgp(\x)\ge\varepsilon\}}\max_{\x'\in M(\x)} \big|\hat g(\x,\x')-g(\x,\x')\big|\xrightarrow{\text{a.s.}}0$, by \Cref{lemma:laplace-estimator}. Hence
\begin{align*}
T_{n,1}(\varepsilon) \le \Delta_{\varepsilon}^{(n)}\,\frac{1}{n}\sum_{i=1}^n K_2(\x_i)\mathbf 1\{\dgp(\x_i)\ge\varepsilon\} \xrightarrow{\text{a.s.}} 0.
\end{align*}

For the tail term, we use the same technique used in \cref{prop:loss_convergence} to show that $T_{n,2}(\varepsilon) \to 0$ as $\varepsilon\downarrow0$. A diagonal argument then yields $\nabla^2_{\btheta}R_n(\btheta_\star)\to0$ a.s., and thus $\nabla^2_{\btheta} \Lhat(\btheta_\star)\xrightarrow{\text{a.s.}} \mathbf H_\star$.

Finally, Assumption~\ref{assumption:theta-diff} with $r=3$ provides an integrable envelope $K_3$ for $\nabla^3_{\btheta}f_{\btheta}$, so the SLLN yields $\nabla^3_{\btheta}\Lhat(\btheta)=O_p(1)$ uniformly in a neighborhood of $\btheta_\star$. Consistency of $\hat\btheta_n$ implies $\|\hat\btheta_n-\btheta_\star\|=o_p(1)$, hence $(\hat{\btheta}_n-\btheta_\star)\cdot\nabla_{\btheta}^3\Lhat(\tilde{\btheta}_n)=o_p(1).$
Together, this shows $\nabla^2_{\btheta} \Lhat(\btheta_\star) +(\hat{\btheta}_n-\btheta_\star)\cdot\nabla_{\btheta}^3\Lhat(\tilde{\btheta}_n)
\xrightarrow{p} \mathbf H_\star$, establishing (i).

Now for (ii) we define for any pmf $p \in \AllAdmPMFsGiven{\dgp}$,
\begin{align*}
\ell_{\btheta}(\x;p) := \frac{1}{|M(\x)|}\sum_{\x'\in M(\x)} \Big(f_{\btheta}(\x,\x')^2 - 2 f_{\btheta}(\x,\x')\log\frac{p(\x')}{p(\x)}\Big).
\end{align*}
If we freeze the pmf inside the loss at $p=\dgp$, define $\hat{\mathcal{L}}^{\mathrm{LRM}}_{n, \mathrm{dgp}}(\btheta) := \frac{1}{n}\sum_{i=1}^n \ell_{\btheta}(\x_i;\dgp)$,
so that
\begin{align*}
\nabla_{\btheta}\hat{\mathcal{L}}^{\mathrm{LRM}}_{n, \mathrm{dgp}}(\btheta_\star) = \frac{1}{n}\sum_{i=1}^n \varphi^{\mathrm{dgp}}_{\btheta_\star}(\x_i).
\end{align*}
Since $\btheta_\star$ minimizes $\Lexp$, we have $\E_{\x\sim \dgp}[\varphi^{\mathrm{dgp}}_{\btheta_\star}(\x)]=\bm 0$, and \Cref{assumption:general_assumptions} and \Cref{assumption:theta-diff} imply $\varphi^{\mathrm{dgp}}_{\btheta_\star}\in L^2(\dgp,\R)$.
Hence the multivariate CLT gives
\begin{align*}
\frac{1}{\sqrt{n}}\sum_{i=1}^n \varphi^{\mathrm{dgp}}_{\btheta_\star}(\x_i) \xrightarrow{d}\ \mathcal N(\bm 0,\mathbf J_\star^{\mathrm{dgp}}),
\quad
\mathbf J_\star^{\mathrm{dgp}} := \E_{\x\sim \dgp}[\varphi^{\mathrm{dgp}}_{\btheta_\star}(\x)\varphi^{\mathrm{dgp}}_{\btheta_\star}(\x)^\top].
\end{align*}

It remains to incorporate the effect of replacing $g$ by $\hat{g}$ in the loss. Recall that $\nabla_{\btheta}\Lhat(\btheta_\star) =\nabla_{\btheta}\hat{\mathcal{L}}^{\mathrm{LRM}}_{n, \mathrm{dgp}}(\btheta_\star) -2\,\nabla_{\btheta}R_n(\btheta_\star)$. Thus
\begin{align*}
\sqrt{n}\,\nabla_{\btheta}\Lhat(\btheta_\star) =\frac{1}{\sqrt{n}}\sum_{i=1}^n\varphi^{\mathrm{dgp}}_{\btheta_\star}(\x_i) -2\sqrt{n}\,\nabla_{\btheta}R_n(\btheta_\star).
\end{align*}
We now show that the second term on the right-hand side admits a linear representation plus an $o_p(1)$ remainder.

Fix $\varepsilon>0$ and define $S_\varepsilon:=\{\x\in\X:\dgp(\x)\ge\varepsilon\}$, which is finite. Define the truncated empirical loss
\begin{align*}
\hat{\mathcal{L}}_\varepsilon(\btheta) := \frac{1}{n}\sum_{i=1}^n  \mathbf{1}\{\x_i\in S_\varepsilon\}\frac{1}{|M(\x_i)|}\sum_{\x'\in M(\x_i)\cap S_\varepsilon} \Big(f_{\btheta}(\x_i,\x')^2 - 2 f_{\btheta}(\x_i,\x')\hat g(\x_i,\x')\Big),
\end{align*}
and its gradient $\nabla_{\btheta}\hat{\mathcal{L}}_\varepsilon(\btheta)$. On $S_\varepsilon$, $\dgp(\x)\ge\varepsilon$ and $\hat q_\alpha(\x)\to \dgp(\x)$ a.s.; hence for $n$ large we have $\hat q_\alpha(\x)\ge\varepsilon/2$ for all $\x\in S_\varepsilon$, almost surely. In particular, on $S_\varepsilon$ the map $\hat{q}_{\alpha,\varepsilon}:=\big(\hat{q}_\alpha(\x)\big)_{\x\in S_\varepsilon} \mapsto\nabla_{\btheta}\hat{\mathcal{L}}_\varepsilon(\btheta_\star)$ is a smooth function of the finite vector $\hat{q}_{\alpha,\varepsilon}$. Let $q_\varepsilon:=(\dgp(\x))_{\x\in S_\varepsilon}$ be the truncated true mass function and let
\begin{align*}
\mathcal{L}_\varepsilon(\btheta; p_\varepsilon) :=\sum_{\x\in S_\varepsilon}p(\x)\,
\frac{1}{|M(\x)|}\sum_{\x'\in M(\x)\cap S_\varepsilon} \Big(f_{\btheta}(\x,\x')^2 -2f_{\btheta}(\x,\x')\log\frac{p(\x')}{p(\x)}\Big)
\end{align*}
denote the corresponding truncated population loss. By the assumptions on $f_{\btheta}$ and the truncation, the map $ p_\varepsilon\mapsto\nabla_{\btheta}\mathcal L_\varepsilon(\btheta_\star;p_\varepsilon)$ is continuously differentiable in a neighborhood of $q_\varepsilon$, with bounded Jacobian. 

For each $\x\in S_\varepsilon$ we have $\sqrt{n}\big(\hat q_\alpha(\x)-\dgp(\x)\big) =\frac{1}{\sqrt{n}}\sum_{i=1}^n\big(\mathbf 1\{\x_i=\x\}-\dgp(\x)\big)+o_p(1)$, where the $o_p(1)$ term comes from the deterministic correction of order $O(n^{-1/2})$. By the multivariate delta method applied at $q_\varepsilon$, there exists a vector of coefficients $(c_\varepsilon(\mathbf{y}))_{\mathbf{y}\in S_\varepsilon}$ such that
\begin{align*}
\sqrt{n}\,\Big( \nabla_{\btheta}\widehat{\mathcal L}_\varepsilon(\btheta_\star) -\nabla_{\btheta}\mathcal L_\varepsilon(\btheta_\star; q_\varepsilon) \Big)
=\frac{1}{\sqrt{n}}\sum_{i=1}^n \sum_{\mathbf{y}\in S_\varepsilon}c_\varepsilon(\mathbf{y})\big(\mathbf 1\{\x_i=\mathbf{y}\}-\dgp(\mathbf{y})\big) +o_p(1).
\end{align*}
But $\nabla_{\btheta}\mathcal L_\varepsilon(\btheta_\star;q_\varepsilon) =\E_{\x\sim \dgp}\big[\nabla_{\btheta}\ell_{\btheta_\star}(\x;\dgp)\,\mathbf 1\{\x\in S_\varepsilon\}\big]$, and by definition,
\begin{align*}
\nabla_{\btheta}\hat{\mathcal{L}}^{\mathrm{LRM}}_{n, \mathrm{dgp}}(\btheta_\star) =\frac{1}{n}\sum_{i=1}^n\varphi^{\mathrm{dgp}}_{\btheta_\star}(\x_i).
\end{align*}
Hence, for each fixed $\varepsilon>0$, $\sqrt{n}\nabla_{\btheta}\hat{\mathcal L}_\varepsilon(\btheta_\star) =\frac{1}{\sqrt{n}}\sum_{i=1}^n\varphi_{\btheta_\star,\varepsilon}(\x_i)+o_p(1)$, where
\begin{align*}
\varphi_{\btheta_\star,\varepsilon}(\x) :=\varphi^{\mathrm{dgp}}_{\btheta_\star}(\x)\,\mathbf{1}\{\x\in S_\varepsilon\} +\sum_{\mathbf{y}\in S_\varepsilon}c_\varepsilon(\mathbf{y})\big(\mathbf{1}\{\x=\mathbf{y}\}-\dgp(\mathbf{y})\big).
\end{align*}
Since $S_\varepsilon$ is finite and the coefficients $c_\varepsilon(y)$ are bounded, and since $\varphi^{\mathrm{dgp}}_{\btheta_\star}\in L^2(\dgp,\R)$, there exists $H_1\in L^2(\dgp,\R)$ and a constant $C_\varepsilon$ such that $\|\varphi_{\btheta_\star,\varepsilon}(\x)\|\le C_\varepsilon H_1(\x)$ for all $\x$. In particular, $\varphi_{\btheta_\star,\varepsilon}\in L^2(\dgp,\R)$ and the multivariate CLT yields
\begin{align*}
\frac{1}{\sqrt{n}}\sum_{i=1}^n\varphi_{\btheta_\star,\varepsilon}(\x_i) \xrightarrow{d}\ \mathcal{N}(\bm{0},\mathbf{J}_\star(\varepsilon)), \qquad
\mathbf J_\star(\varepsilon) :=\E_{\x\sim \dgp}\big[\varphi_{\btheta_\star,\varepsilon}(\x) \varphi_{\btheta_\star,\varepsilon}(\x)^\top\big].
\end{align*}

Now, to control the tail, we define $\Delta_n(\varepsilon) :=\sqrt{n}\Big( \nabla_{\btheta}\Lhat(\btheta_\star) -\nabla_{\btheta}\hat{\mathcal{L}}_\varepsilon(\btheta_\star)\Big)$. Only indices with $\x_i\notin S_\varepsilon$, i.e.\ $\dgp(\x_i)<\varepsilon$, contribute to $\Delta_n(\varepsilon)$. Using the expressions for the gradients of $\Lhat$ and $\hat{\mathcal{L}}_\varepsilon$, the polynomial growth of $f_{\btheta_\star}$ and $\nabla_{\btheta}f_{\btheta_\star}$, we obtain
\begin{align*}
\|\Delta_n(\varepsilon)\| \le \frac{1}{\sqrt{n}}\sum_{i=1}^n H_2(\x_i)\,\mathbf 1\{\dgp(\x_i)<\varepsilon\}
\end{align*}
for some $H_2\in L^2(\dgp,\R)$.
Consequently $\E\big[\|\Delta_n(\varepsilon)\|^2\big] \leq \E_{\x\sim \dgp}\big[H_2(\x)^2\mathbf 1\{\dgp(\x)<\varepsilon\}\big] \xrightarrow[\varepsilon\downarrow 0]{}0$ by dominated convergence. Hence, for every $\delta>0$,
\begin{align*}
\lim_{\varepsilon\downarrow0}\sup_n \mathbb{P}\big(\|\Delta_n(\varepsilon)\|>\delta\big)=0.
\end{align*}
Thus, the tail contribution to the score can be made arbitrarily small in probability, uniformly in $n$, by choosing $\varepsilon$ sufficiently small.

For $0<\varepsilon'\le\varepsilon$, observe that $\varphi_{\btheta_\star,\varepsilon}(\x)=\varphi_{\btheta_\star,\varepsilon'}(\x)$ whenever $\dgp(\x)\ge\varepsilon$; both are bounded by a constant multiple of $H_2(\x)$.
Therefore
\begin{align*}
\|\varphi_{\btheta_\star,\varepsilon}(\x)-\varphi_{\btheta_\star,\varepsilon'}(\x)\|\le C\,H_2(\x)\,\mathbf 1\{\dgp(\x)<\varepsilon\},
\end{align*}
and hence $\E_{\x\sim \dgp}\bigl[\|\varphi_{\btheta_\star,\varepsilon}(\x)-\varphi_{\btheta_\star,\varepsilon'}(\x)\|^2\bigr]\le C^2\,\E_{\x\sim \dgp}\big[H_2(\x)^2\mathbf 1\{\dgp(\x)<\varepsilon\}\big]\xrightarrow[\varepsilon\downarrow0]{}0$. Thus $\{\varphi_{\btheta_\star,\varepsilon}\}_{\varepsilon>0}$ is Cauchy in $L^2(\dgp,\R)$, and there exists $\varphi_{\btheta_\star}\in L^2(\dgp,\R)$ such that
\begin{align*}
\varphi_{\btheta_\star,\varepsilon}\xrightarrow[\varepsilon\downarrow0]{L^2(\dgp,\R)}
\varphi_{\btheta_\star}.
\end{align*}
In particular, $\E_{\x\sim \dgp}[\varphi_{\btheta_\star}(\x)]=\bm 0$ by continuity of expectation on $L^2$. Choose a sequence $\varepsilon_n\downarrow0$ slowly enough so that $\Delta_n(\varepsilon_n)\to0$ in probability. Then
\begin{align*}
\sqrt{n}\,\nabla_{\btheta}\Lhat(\btheta_\star) =\sqrt{n}\,\nabla_{\btheta}\hat{\mathcal{L}}_{\varepsilon_n}(\btheta_\star) +\Delta_n(\varepsilon_n) 
=\frac{1}{\sqrt{n}}\sum_{i=1}^n\varphi_{\btheta_\star,\varepsilon_n}(\x_i)+o_p(1).
\end{align*}
Since $\varphi_{\btheta_\star,\varepsilon_n}\to\varphi_{\btheta_\star}$ in $L^2(\dgp, \mathbb{R})$ and $\{\x_i\}_{i=1}^{n}$ are i.i.d., the CLT and Slutsky’s theorem imply
\begin{align*}
\frac{1}{\sqrt{n}}\sum_{i=1}^n \big(\varphi_{\btheta_\star,\varepsilon_n}(\x_i) -\varphi_{\btheta_\star}(\x_i)\big)\xrightarrow{p}0,
\end{align*}
and $\frac{1}{\sqrt{n}}\sum_{i=1}^n\varphi_{\btheta_\star}(\x_i) \xrightarrow{d} \mathcal{N}(\bm 0,\mathbf J_\star)$, and $\mathbf{J}_\star:=\E_{\x\sim \dgp}[\varphi_{\btheta_\star}(\x)\varphi_{\btheta_\star}(\x)^\top]$. Combining these, we obtain $\sqrt{n}\nabla_{\btheta}\Lhat(\btheta_\star) =\frac{1}{\sqrt{n}}\sum_{i=1}^n\varphi_{\btheta_\star}(\x_i)+o_p(1)$, which proves (ii) with the stated $\mathbf{J}_\star$.

Part (i) showed that $\nabla^2_{\btheta} \Lhat(\btheta_\star) +(\hat{\btheta}_n-\btheta_\star)\cdot\nabla_{\btheta}^3\Lhat(\tilde{\btheta}_n) \xrightarrow{p} \mathbf{H}_\star$, and part (ii) established $\sqrt{n}\,\nabla_{\btheta} \Lhat(\btheta_\star) \xrightarrow{d} \mathcal{N}(\mathbf 0,\mathbf{J}_\star)$. Applying Slutsky’s theorem to the representation~\eqref{eq:theta-Taylor}, we conclude that
\begin{align*}
\sqrt{n}\,(\hat{\btheta}_n-\btheta_\star) \xrightarrow{d} \mathcal{N}\left(\mathbf{0}, \mathbf H_{\star}^{-1} \mathbf{J}_{\star} \mathbf{H}_{\star}^{-1}\right),
\end{align*}
which completes the proof.
\end{proof}
Given the technical results established for the LRM estimator, we can now derive a posterior consistency and a Bernstein–von Mises theorem for the corresponding LRM generalized posterior. To do so, an additional assumption on the prior is required. This assumption---commonly referred to as the prior mass condition---requires that the prior assigns positive probability mass to the loss minimizer.

\begin{theorem} [Consistency \& Bernstein–von Mises]
Suppose \Cref{assumption:laplace-estimator}, \Cref{assumption:general_assumptions}, \Cref{assumption:theta-diff} for $r=3$, and  \Cref{assumption:minimisers} hold. Moreover, the prior $\pi$ admits a density that is continuous at $\btheta_\star$ with $\pi(\btheta_\star)>0$
Let $B_\epsilon(\btheta_\star)=\{\btheta\in\Theta\,:\, \|\btheta-\btheta_{\star}\|_{2}\leq\epsilon\}$. Then, for any $\epsilon>0$,
\begin{equation*}
    \int_{B_\epsilon(\btheta_\star)}\hat{\pi}_M^\beta(\btheta)d\btheta \;\xrightarrow{\text{a.s}}\; 1.
\end{equation*}
Let $\tilde{\pi}_M$ the p.d.f. of the random variable $\tilde{\btheta}_n := \sqrt{n}(\btheta - \hat\btheta_n)$ for $\btheta \sim \hat{\pi}_M^\beta$, viewed as a p.d.f. on $\mathbb{R}^{p}$.
Let $\mathbf H_{\star} := \nabla^2_{\btheta}\,\Lexp (\btheta)\big|_{\btheta=\btheta_\star}$. If $\mathbf H_\star$ is non-singular,
    \begin{equation*}
    \int_{\mathbb R^p}
    \left|
    \tilde\pi_{M}(\tilde{\btheta}_n) - \frac{1}{\det(2\pi \mathbf H_\star^{-1})^{1/2}}
    \exp\!\left(-\tfrac{1}{2}\btheta^\top \mathbf H_\star\btheta\right)
    \right|\,d\btheta \;\xrightarrow{\text{a.s}}\; 0.
\end{equation*}
    \label{prop:bvm}
\end{theorem}
\begin{proof}
    We verify, almost surely, the conditions of \citet[Theorem~4]{miller2021asymptotic} for the loss sequence $\{\Lhat:n\in\mathbb{N}\}$.

\noindent\textbf{1. Prior Mass.}
By assumption, the prior $\pi$ admits a density continuous at $\btheta_\star$ with $\pi(\btheta_\star)>0$. This matches the prior conditions in \citet[Theorem~4]{miller2021asymptotic}.

\noindent\textbf{2. Consistency of the minimizer.}
Let $\hat\btheta_n\in\arg\min_{\btheta\in\Theta}\Lhat(\btheta)$ denote the empirical minimizer. By \Cref{prop:consistency}, we have $\hat\btheta_n \xrightarrow{\text{a.s.}} \btheta_\star$.

\noindent\textbf{3. Local quadratic expansion of $\Lhat$.}
\Cref{assumption:theta-diff} with $r=3$ ensures that $\Lhat$ is three times continuously differentiable in $\btheta$ on $\Theta$. For $\btheta$ in a neighborhood of $\hat\btheta_n$, a second-order Taylor expansion of $\Lhat$ around $\hat\btheta_n$ gives
\begin{align*}
  \Lhat(\btheta)
= \Lhat(\hat\btheta_n)
  + \nabla_{\btheta}\Lhat(\hat\btheta_n)^\top(\btheta-\hat\btheta_n)
  + \tfrac{1}{2}(\btheta-\hat\btheta_n)^\top \mathbf H_n(\tilde\btheta_n)\,(\btheta-\hat\btheta_n),  
\end{align*}

where $\tilde\btheta_n$ lies on the line segment between $\btheta$ and $\hat\btheta_n$, and $\mathbf{H}_n(\vartheta):=\nabla^2_{\btheta}\Lhat(\vartheta)$.

Since $\hat\btheta_n$ minimizes $\Lhat$, we have 
$\nabla_{\btheta}\Lhat(\hat\btheta_n)=\mathbf 0$, so this simplifies to
\begin{align*}
\Lhat(\btheta)
& = \Lhat(\hat\btheta_n)
  + \tfrac{1}{2}(\btheta-\hat\btheta_n)^\top \mathbf{H}_n(\tilde\btheta_n)\,(\btheta-\hat\btheta_n). \\
& = \Lhat(\hat\btheta_n)
  + \tfrac{1}{2}(\btheta-\hat\btheta_n)^\top \mathbf{H}_n(\hat\btheta_n)\,(\btheta-\hat\btheta_n)
  + r_n(\btheta-\hat\btheta_n),
\end{align*}
with remainder $r_n(\btheta-\hat\btheta_n)
= \tfrac{1}{2}(\btheta-\hat\btheta_n)^\top\big(\mathbf{H}_n(\tilde\btheta_n)-\mathbf{H}_n(\hat\btheta_n)\big)(\btheta-\hat\btheta_n)$. By the mean value theorem and \Cref{assumption:theta-diff} with $r=3$, there exists a constant $C_n$ such that $|r_n(\btheta-\hat\btheta_n)| \le C_n \|\btheta-\hat\btheta_n\|_2^3$,
for all $\btheta$ in a fixed neighborhood of $\btheta_\star$. In \Cref{prop:clt} we have shown that $\limsup_{n\to\infty}\; C_n < \infty$ a.s.,  so the remainder is $O(\|\btheta-\hat\btheta_n\|_2^3)$ with an almost surely bounded coefficient. This provides the local quadratic approximation required by \citet[Theorem~4]{miller2021asymptotic}.

\noindent\textbf{4. Convergence of the Hessian.}
By \Cref{prop:pointwise} and \Cref{assumption:theta-diff},  $\Lexp(\btheta)$ is twice continuously differentiable in a neighborhood of $\btheta_\star$. Its Hessian $\mathbf H_\star := \nabla^2_{\btheta} \Lexp(\btheta)\big|_{\btheta=\btheta_\star}$ is symmetric by Schwarz’s theorem \citep[Theorem~9.41]{rudin1976principles}.  
In \Cref{prop:clt} we also showed that $\nabla^2_{\btheta}\Lhat(\btheta_\star) \xrightarrow{\text{a.s.}} \nabla^2_{\btheta}\Lexp(\btheta_\star)=\mathbf H_\star$. Since $\btheta_\star$ is the minimizer of $\Lexp$ by \Cref{assumption:minimisers}, $\mathbf H_\star$ is positive semidefinite. Under the assumption that $\mathbf H_\star$ is non-singular, it is positive definite. This verifies the condition required in \citet[Theorem~4]{miller2021asymptotic}.

\noindent\textbf{5. Separation of minimizers.}
By the uniform convergence of $\Lhat$ to $\Lexp$ established in \Cref{prop:loss_convergence_unif}, $\sup_{\btheta\in\Theta}|\Lhat(\btheta)-\Lexp(\btheta)| \xrightarrow{\text{a.s.}} 0$. Hence, for any fixed $\varepsilon>0$,
\begin{align*}
\liminf_{n\to\infty}\Big(\inf_{\{\btheta:\|\btheta-\btheta_\star\|\ge\varepsilon\}}  \Lhat(\btheta)-\Lhat(\hat\btheta_n)\Big)
&\ge \liminf_{n\to\infty}\inf_{\{\|\btheta-\btheta_\star\|\ge\varepsilon\}} 
\Lhat(\btheta) - \limsup_{n\to\infty}\Lhat(\hat\btheta_n) \\
&= \inf_{\{\|\btheta-\btheta_\star\|\ge\varepsilon\}} \Lexp(\btheta) - \Lexp(\btheta_\star) > 0 \quad \text{a.s.},
\end{align*}
where the last inequality uses \Cref{assumption:minimisers}. This is exactly the separation-of-minimizers condition in \citet[Theorem~4]{miller2021asymptotic}.

1–5 verify, almost surely, the hypotheses of the consistency and Bernstein–von Mises theorem for loss sequences in \citet[Theorem~4]{miller2021asymptotic}. Therefore, the posterior $\hat\pi_M^\beta$ satisfies $\int_{B_\epsilon( \btheta_\star)}\hat{\pi}_M^\beta(\btheta)\,d\btheta 
\xrightarrow{\text{a.s.}} 1$ for all $\epsilon>0$, and the posterior density $\tilde\pi_M$ satisfies
\begin{align*}
    \int_{\mathbb R^p}\left|
    \tilde\pi_{M}(\tilde\btheta_n) - \frac{1}{\det(2\pi \mathbf H_\star^{-1})^{1/2}}
    \exp\left(-\tfrac{1}{2}\btheta^\top \mathbf H_\star\btheta\right)
    \right|d\btheta \xrightarrow{\text{a.s.}}\; 0.
\end{align*}
This establishes the consistency and the Bernstein-von Mises result.
\end{proof}

\section{Proofs of Theoretical Results}\label{appendix:proofs-of-theory}

\subsection{\Cref{theorem:divergence}}\label{appendix:proof-divergence}
\begin{proof}
Suppose $q \in \AllPMFs$ such that $q \in \AllAdmPMFsGiven{q}$ and  $p \in \AllAdmPMFsGiven{q}$. We first show that if $q=p$, then clearly
\begin{equation*}
   \log \frac{q(\bm{x}')}{q(\bm{x})}= \log \frac{p(\bm{x}')}{p(\bm{x})}  \qquad \forall \x,\x'\in\mathcal{X}
\end{equation*}
and so every term in the definition of $\LRM(q \|p)$ is zero. Hence $\LRM(q\|p)=0$.

Now, for the converse, assume $\LRM(q||p) = 0$. Since every term in the expectation  is non-negative, this implies:
$$ \sum_{\bm{x}' \in M(\bm{x})}  \left (\log \frac{p(\bm{x}')}{p(\bm{x})} - \log \frac{q(\bm{x}')}{q(\bm{x})} \right)^2 = 0 \qquad \forall x\in\mathcal{X}.$$
Therefore, $\frac{q(\x)}{q(\x')} = \frac{p(\x)}{p(\x')}$ for all $\x\in \X, \x'\in M(\x)$. From this equation we note that $\forall \x \in \X$, and $\x'\in M(\x)$ we can write:
\begin{equation}
    q(\x) = q(\x') \frac{q(\x)}{q(\x')} = q(\x') \frac{p(\x)}{p(\x')}. \label{eq:pq-iterative}
\end{equation}
In particular, this equality does not necessarily hold for arbitrary $\x' \in \mathcal{X}$. We require $\x' \in M(\x)$. However, if \Cref{assumption:graph-connect} holds; that is, $G$ is connected, it can be extended iteratively. 

Suppose $G$ is connected by \Cref{assumption:graph-connect}. Then, for $\x, \x^\star\in\X$, there exists a path $\x, \x_1, \x_2, \dots, \x_{n-1},\x^\star$ for $\x_i \in \X$, $i=1, \dots, n-1$, $n \in \mathbb{N}$, such that $\x_1\in M(\x), \x_2\in M(\x_1),..., \x^\star \in M(\x_{n-1})$. Therefore, we can apply \Cref{eq:pq-iterative} iteratively so that
\begin{equation*}
    q(\x) = q(\x_1) \frac{p(\x)}{p(\x_1)} = q(\x_2) \frac{p(\x_1)}{p(\x_2)} \frac{p(\x)}{p(\x_1)} = q(\x_2) \frac{p(\x)}{p(\x_2)}=...=q(\x^\star) \frac{p(\x)}{p(\x^\star)}
\end{equation*}
Thus, for any two points $\x, \x^\star \in\X$:
\begin{equation}
    q(\x) =  p(\x) \frac{q(\x^\star)}{p(\x^\star)} \label{eq:pq-star}
\end{equation}

We now proceed by contradiction. Suppose there is $\x_0\in\mathcal{X}$ such that $q(\x_0)\neq p(\x_0)$. Since the graph $G$ is connected, for every $\x\in\X$, there exists a path (depending on $\x$) from $\x$ to $\x_0$. Then, by \Cref{eq:pq-star}, $q(\x) = p(\x) \frac{q(\x_0)}{p(\x_0)}$ for all $\x\in\X$. Finally:
\begin{equation*}
    1 = \sum_{\x\in\X}q(\x)=\sum_{\x\in\X}p(\x) \frac{q(\x_0)}{p(\x_0)} =  \frac{q(\x_0)}{p(\x_0)}\sum_{\x\in\X}p(\x) = \frac{q(\x_0)}{p(\x_0)} \neq 1
\end{equation*}
which is a contradiction. Therefore, when $G$ is connected, there are no $\x_0 \in \mathcal{X}$ such that $q(\x_0) \neq p(\x_0)$, implying $q=p$.
\end{proof}
\subsection{\Cref{prop:exp_fam}}\label{appendix:exp_fam}
\begin{proof}
From \Cref{eq:exp-family-model}, the exponential family model can be expressed as  $\log p^{\exp}_{\btheta} (\x) = \bm{\eta}(\btheta)^\top \cdot \mathbf{T}(\x) + B(\x) - \log Z(\btheta)$.
Therefore, for $\x' \in \X$, the log-ratio is
\[\log \frac{p^{\exp}_{\bm{\theta}}(\x')}{p^{\exp}_{\bm{\theta}}(\x)} = \bm{\eta}^\top \cdot \left (\mathbf{T}(\x') - \mathbf{T}(\x) \right) + B(\x') - B(\x):= \bm{\eta}^\top \Delta \mathbf{T}(\x', \x) + \Delta B(\x', \x),
\]
where we write for shorthand $\bm{\eta} := \bm{\eta}(\btheta)$, and define $\Delta \mathbf{T}(\x', \x):= \mathbf{T}(\x') - \mathbf{T}(\x)$, $\Delta B(\x', \x):=B(\x') - B(\x)$. Therefore, with data-generating process $\dgp \in \AllPMFs$,  observations $\{\x_i \}_{i=1}^n \overset{i.i.d.}{\sim}\dgp$, 
estimate $\hat{q} \in \AllPMFsGiven{\dgp}$ and for any two $\x, \x' \in \operatorname{supp}(\dgp)$, we compute the summand of the loss $\Lhat$ from \Cref{eq:log-ratio-matching-loss} as follows:
\begin{equation*}
\begin{split}
    &\left (\log \frac{p^{\exp}_{\bm{\theta}}(\x')}{p^{\exp}_{\bm{\theta}}(\x)} \right)^2 - 2 \log \frac{p^{\exp}_{\bm{\theta}}(\x')}{p^{\exp}_{\bm{\theta}}(\x)}\log \frac{\hat{q}(\x')}{\hat{q}(\x)} = \\
    & \bm{\eta}^\top \Delta \mathbf{T}(\x', \x) \Delta \mathbf{T}(\x', \x)^\top  \bm{\eta} - 2 \bm{\eta}^\top \cdot \Delta \mathbf{T}(\x', \x) \left (\log \frac{\hat{q}(\x')}{\hat{q}(\x)} - \Delta B(\x', \x) \right) + C(\x, \x'),
\end{split}
\end{equation*}
where $C(\x, \x')$ is a function not depending on $\bm{\eta}$. With matching set $M$, we define 
\begin{equation*}
    \begin{split}
        &\bm{\Lambda}_n := \frac{1}{n} \sum_{i=1}^n \frac{1}{|M(\x_i)|} \sum_{\x' \in M(\x_i)} \Delta \mathbf{T}(\x', \x) \Delta \mathbf{T}(\x', \x)^\top \\
        &\bm{\nu}_n:= \frac{1}{n}\sum_{i=1}^n \frac{1}{|M(\x_i)|} \sum_{\x' \in M(\x_i)} \Delta \mathbf{T}(\x', \x) \left (\log \frac{\hat{q}(\x')}{\hat{q}(\x_i)} - \Delta B(\x', \x)  \right),
    \end{split}
\end{equation*}
and obtain that the loss is given by
\begin{equation*}
\begin{split}
    \Lhat(\bm{\theta}) &= \frac{1}{n} \sum_{i=1}^n \frac{1}{|M(\bm{x}_i)|} \sum_{\bm{x}' \in M(\bm{x}_i)}  \left(\log \frac{p^{\exp}_{\bm{\theta}}(\bm{x}')}{p^{\exp}_{\bm{\theta}}(\bm{x}_i)}\right)^2 -2 \log\frac{p^{\exp}_{\bm{\theta}}(\bm{x}')}{p^{\exp}_{\bm{\theta}}(\bm{x}_i)}\log \frac{\hat{q}(\bm{x}')}{\hat{q}(\bm{x}_i)}  \\
    &= \bm{\eta}(\btheta)^\top \bm{\Lambda}_n \bm{\eta}(\btheta) - 2 \bm{\eta}(\btheta)^\top \bm{\nu}_n,
\end{split}
\end{equation*}
where we reintroduce the notation $\bm{\eta}(\btheta)$ to emphasize the dependency of the loss on $\btheta$. Now, if we assume an exponentially quadratic prior \(\pi(\bm{\eta}) \propto \exp (-\frac{1}{2}(\bm{\eta} - \bm{\mu} )^\top \mathbf{\Sigma}^{-1}  (\bm{\eta} - \bm{\mu} ) ) \), the generalized log-ratio matching posterior is obtained as:
\begin{equation*}
\begin{split}
    \log \pi_{M}  &\propto -\frac{1}{2}(\bm{\eta} - \bm{\mu} )^\top \mathbf{\Sigma}^{-1}  (\bm{\eta} - \bm{\mu} ) -\beta n  \Lhat(\btheta)\\
    &\propto -\frac{1}{2}(\bm{\eta}^\top  \mathbf{\Sigma}^{-1} \bm{\eta} - 2 \bm{\eta}^\top \mathbf{\Sigma}^{-1} \bm{\mu}) - \beta n \left (\bm{\eta}^\top \bm{\Lambda}_n \bm{\eta} - 2 \bm{\eta}^\top \bm{\nu}_n \right) \\
    &= -\frac{1}{2} \left ( \bm{\eta}^\top (\mathbf{\Sigma}^{-1} + 2 \beta n \bm{\Lambda}_n) \bm{\eta}  - 2 \bm{\eta}^\top (\mathbf{\Sigma}^{-1} \bm{\mu} + 2 \beta n \bm{\nu}_n) \right),
\end{split}
\end{equation*}
where on the first line, we write the definition of the (log) generalized posterior on $\bm{\eta}$; on the second line we expand the prior, replace the loss by its quadratic form, and absorb into $\propto$ constants not depending on $\bm{\eta}$; and on the last line, we group terms. This last line implies $\pi_{M}$ is 
$ \mathcal{N}(\bm{\mu}_n, \mathbf{\Sigma}_n)$ with $\mathbf{\Sigma}_n := \left ( \mathbf{\Sigma}^{-1} + 2 \beta n \bm{\Lambda}_n \right)^{-1}$ and $\bm{\mu}_n := \mathbf{\Sigma}_n \left ( \mathbf{\Sigma}^{-1} \bm{\mu} + 2 \beta n \bm{\nu}_n\right)$.
\end{proof}
\subsection{\Cref{lemma:laplace-estimator}}
\begin{proof}
For any $\varepsilon>0$, the set
$S_\varepsilon := \{\x\in\X:\ \dgp(\x)\ge\varepsilon\}$ is finite.
Indeed, since $\dgp$ is a probability mass function,
\begin{align*}
1 = \sum_{\x\in\X} \dgp(\x) \geq \sum_{\x\in S_\varepsilon} \dgp(\x)
\geq \sum_{\x\in S_\varepsilon} \varepsilon
= \varepsilon |S_\varepsilon|.
\end{align*}
Thus $|S_\varepsilon|\le 1/\varepsilon < \infty$. Fix any $\x\in S_\varepsilon$.  
By definition, $\hat q_\alpha(\x)=(C_n(\x)+\alpha\tildebasePMF(\x))/(n+\alpha \ZbasePMF)$ and $\ZbasePMF=\sum_{\x\in\X}\tildebasePMF(\x)$.
Since $C_n(\x)=\sum_{i=1}^n \mathbf 1\{\x_i=\x\}$,
the strong law of large numbers gives $C_n(\x)/n\ \longrightarrow\ \dgp(\x)$ a.s. Dividing numerator and denominator of $\hat q_\alpha(\x)$ by $n$ yields
\begin{align*}
\hat q_\alpha(\x) =\frac{C_n(\x)/n + \alpha\,\tildebasePMF(\x)/n}{1+\alpha \ZbasePMF/n},
\end{align*}
and since $\tildebasePMF(\x)/n\to 0$ and $\alpha \ZbasePMF/n\to 0$, we obtain $\hat q_\alpha(\x) \rightarrow \dgp(\x)$ a.s. for each $\x\in S_\varepsilon$.

For $\x\in S_\varepsilon$ we have $\dgp(\x)\ge\varepsilon$, hence $\dgp(\x)\in[\varepsilon,1]$.
Since $\hat q_\alpha(\x)\to \dgp(\x)$ a.s., there exists a $N_{\x}$
such that for all $n\ge N_{\x}$, $\hat q_\alpha(\x)\in[\varepsilon/2,1]$.

Because $S_\varepsilon$ is finite, we can take $N:=\max_{\x\in S_\varepsilon} N_{\x}$, which is finite a.s.  Thus for all $n\ge N$, $\hat q_\alpha(\x)\in[\varepsilon/2,1]$, $\forall \x\in S_\varepsilon$. The function $u\mapsto \log u$ is continuously differentiable on $[\varepsilon/2,1]$
and satisfies $|\log u - \log v|\leq\frac{2}{\varepsilon}|u-v|$ for all $u,v\in[\varepsilon/2,1]$. Therefore, for all $\x\in S_\varepsilon$, $\bigl|\log\hat q_\alpha(\x)-\log \dgp(\x)\bigr|\leq\frac{2}{\varepsilon}\,|\hat q_\alpha(\x)-\dgp(\x)|$,  and taking $\max_{\x\in S_\varepsilon}$ gives $\max_{\x\in S_\varepsilon} \bigl|\log\hat q_\alpha(\x)-\log \dgp(\x)\bigr| \xrightarrow{\text{a.s.}}0$. Fix $\x\in S_\varepsilon$ and any $\x'\in M(\x)$. We have
\begin{align*}
\left|\log\frac{\hat{q}_{\alpha}(\x')}{\hat{q}_{\alpha}(\x)}-\log\frac{\dgp(\x')}{\dgp(\x)}\right|\leq\bigl|\log\hat{q}_{\alpha}(\x')-\log \dgp(\x')\bigr| + \bigl|\log\hat{q}_{\alpha}(\x)-\log \dgp(\x)\bigr|.
\end{align*}
Now take the maximum over $\x'\in M(\x)$, then the supremum over $\x\in S_\varepsilon$.
Since $S_\varepsilon$ is finite and $M(\x)$ is finite for each $\x$, all maxima are finite, and the right-hand side converges to $0$ a.s. Hence
\begin{align*}
\sup_{\x\in S_\varepsilon}\max_{\x'\in M(\x)}\left|\log\frac{\hat q_\alpha(\x')}{\hat q_\alpha(\x)}-\log\frac{\dgp(\x')}{\dgp(\x)}\right|\xrightarrow{\text{a.s.}}\;0.
\end{align*}

This proves the desired truncated convergence.
\end{proof}

\subsection{\Cref{theorem:bvm}}\label{appendix:bvm}

Before proving \Cref{theorem:bvm}, we establish the following auxiliary result.

\begin{proposition}
Suppose \Cref{assumption:laplace-estimator} and $\model$ is an exponential family as in \Cref{eq:exp-family-model} with natural parameter $\eta(\btheta)=\btheta$. Let $E\subseteq \Theta$ be open and bounded. Then $\Lhat(\btheta)\overset{\text{a.s.}}{\longrightarrow}\Lexp(\btheta)$ pointwise for all $\btheta\in E$.
\label{prop:pointwise}
\end{proposition}
\begin{proof}
We first verify that natural exponential families satisfy \Cref{assumption:theta-diff} with $r=0$.

The model is an exponential family of the form $\model(\x) := \exp \left( \btheta^\top \mathbf{T}(\x) + B(\x) - \log Z(\btheta) \right)$. For any $\x,\x'\in\X$ and $\btheta\in\Theta$ we have 
\begin{align*}
\log\frac{\model(\x')}{\model(\x)} = \bigl(\log B(\x') - \log B(\x)\bigr) + \btheta^\top\bigl(\mathbf{T}(\x') - \mathbf{T}(\x)\bigr),
\end{align*}
since the terms $Z(\btheta)$ cancel in the ratio. Fix an open, bounded set $E\subset\Theta$. Because $E$ is bounded, there exists a finite constant $C_E := \sup_{\btheta\in E} \|\btheta\| < \infty$. Then, for any $\x,\x'\in\X$ and $\btheta\in E$,
\begin{align*}
\left|\log\frac{\model(\x')}{\model(\x)}\right|
\le \bigl|\log B(\x') - \log B(\x)\bigr|
    + C_E\,\|\mathbf{T}(\x') - \mathbf{T}(\x)\|.
\end{align*}

Define $K_0(\x) := \max_{\x' \in M(\x)} \left( \bigl|\log B(\x') - \log B(\x)\bigr| + C_E\,\|\mathbf{T}(\x') - \mathbf{T}(\x)\|\right)$. Then, for all $\x\in\X$, $\x'\in M(\x)$ and $\btheta\in E$, $\sup_{\btheta\in E} \left| \log\frac{\model(\x')}{\model(\x)} \right| \le K_0(\x)$, which is \Cref{assumption:theta-diff} with $r = 0$. By \Cref{assumption:parametric-model}, the $B$ and $T$ satisfy conditions ensuring that $K_0 \in L^2(\dgp,\mathbb{R})$. Hence \Cref{assumption:theta-diff} holds with $r=0$ for natural exponential families on $E$. Therefore, all conditions of \Cref{prop:loss_convergence} are satisfied on $E$, and we obtain $\Lhat(\btheta)\xrightarrow{\text{a.s.}}\Lexp(\btheta)$
 for all $\btheta\in E$.
\end{proof}
We now prove \Cref{theorem:bvm} by verifying the conditions of \citet[Theorem~5]{miller2021asymptotic} a.s.

\begin{proof}
\noindent\textbf{1. Prior mass.}
By assumption, the prior $\pi$ admits a continuous density on $\btheta$ and \ $\pi(\btheta_\star) > 0$. This verifies the required prior mass condition.

\noindent\textbf{2. Pointwise convergence of the empirical loss.}
Let $E\subseteq \Theta$ be an open and bounded set such that $\btheta_\star\in E$. By \Cref{prop:pointwise}, we have $\Lhat(\btheta) \xrightarrow{\text{a.s.}} \Lexp(\btheta)$ for all $\btheta \in E$. Thus, the empirical loss converges pointwise a.s. to its population counterpart in $E$.

\noindent\textbf{3. Regularity of the empirical loss $\Lhat$.}
By \cref{prop:exp_fam}, the $\Lhat(\btheta)$ is quadratic in $\btheta$.  In particular, it is convex and possesses uniformly bounded third derivatives; in fact, these derivatives are zero. This matches the smoothness requirement of \citet{miller2021asymptotic}.

\noindent\textbf{4. Regularity of the population loss $\Lexp$.}
Since $\btheta_\star = \arg\min_{\btheta\in\Theta}\Lexp(\btheta)$, we have first-order optimality: $\nabla_{\btheta} \Lexp(\btheta_\star) = \mathbf 0$.
Furthermore, the Hessian at the minimizer, $\mathbf{H}_{\star} := \nabla_{\btheta}^2 \Lexp(\btheta_\star)$, is positive definite.

(1)-(4) verify, almost surely, all hypotheses of the consistency and Bernstein-von Mises result for posteriors based on loss sequences \citep[Theorem~5]{miller2021asymptotic}.  
Consequently, the posterior $\hat\pi_M^\beta$ concentrates at $\btheta_\star$: $\int_{B_\epsilon(\btheta_\star)} \hat\pi_M^\beta(\btheta)\,d\btheta 
\xrightarrow{\text{a.s.}} 1$ for all  $\epsilon > 0$.

Moreover, the rescaled posterior density $\tilde\pi_M$ converges to a Gaussian distribution with covariance $\mathbf H_\star^{-1}$:
\begin{align*}
\int_{\mathbb{R}^p}
\left|\tilde\pi_M(\tilde{\btheta}_n)-\frac{1}{\det(2\pi \mathbf{H}_{\star}^{-1})^{1/2}}
\exp\left(-\tfrac{1}{2}\btheta^\top \mathbf{H}_{\star} \btheta\right)\right|
d\btheta\xrightarrow{\text{a.s.}} 0.
\end{align*}

This completes the proof of \Cref{theorem:bvm}.
\end{proof}

\end{document}